\documentclass{article}

\PassOptionsToPackage{numbers, compress}{natbib}

\usepackage[final]{staix_2026}         

\usepackage[utf8]{inputenc}
\usepackage[T1]{fontenc}
\usepackage{hyperref}
\usepackage{url}
\usepackage{booktabs}
\usepackage{amsfonts}
\usepackage{nicefrac}
\usepackage{microtype}
\usepackage{xcolor}
\usepackage{graphicx}
\usepackage{amsmath,amssymb,amsthm}
\usepackage{float}

\usepackage{subcaption}

\newtheorem{theorem}{Theorem}
\newtheorem{lemma}{Lemma}
\newtheorem{proposition}{Proposition}

\title{Conditional-Independence-Regularized Distributional Autoencoders for Mixed-Type Data}

\author{%
  Siyuan Tang \\
  Department of Statistics \\
  University of Michigan \\
  \texttt{siyuat@umich.edu}
  \And
  Gongjun Xu \\
  Department of Statistics \\
  University of Michigan \\
  \texttt{gongjun@umich.edu}
  \And
  Ji Zhu \\
  Department of Statistics \\
  University of Michigan \\
  \texttt{jizhu@umich.edu}
}

\begin{document}

\maketitle


\begin{abstract}
Mixed-type data containing both numerical and categorical variables arise in many scientific and real-world applications. Existing representation learning and generative modeling approaches typically focus either on reconstruction accuracy or unconditional data generation, but often fail to recover the full conditional distribution of the data while preserving interpretable structural relationships between heterogeneous variable types. In this work, we introduce Conditional-Independence-Regularized Distributional Autoencoders, a framework for learning low-dimensional representations of mixed-type data through conditional distribution matching and structural regularization. Our method combines an energy-score-based objective for numerical variables, a likelihood-based objective for categorical variables, and an auxiliary conditional independence regularization term encouraging the learned representation to capture the dependence between numerical and categorical components. We provide theoretical analysis showing that the optimal representation balances unexplained numerical variability, conditional entropy of categorical variables, and residual conditional dependence. Empirically, the proposed method achieves strong performance on both synthetic and real-world datasets, substantially improving categorical distribution recovery, achieving competitive overall conditional distribution recovery, and preserving mixed-type dependence structure. The code has been made available at \href{https://github.com/siyuat-ui/Conditional-Independence-Regularized-DAE-public}{\color{red} GitHub}.
\end{abstract}


\section{Introduction}

Data with mixed variable types arise in many modern applications, including healthcare, finance, education, and other scientific studies \citep{johnson2016mimic, kovalerchuk2000data, 5524021, borisov2022deep, uci_repository}. Such data typically consist of both numerical and categorical variables, often exhibiting complex and heterogeneous dependencies that are difficult to model accurately \citep{pmlr-v202-kotelnikov23a, tabsyn}. Denote the observed data by $X = (X^{\mathrm{num}}, X^{\mathrm{cat}})$, where $X^{\mathrm{num}}$ and $X^{\mathrm{cat}}$ represent the numerical and categorical components, respectively. A challenge is to learn a low-dimensional representation $Z = e(X)$, where $e(\cdot)$ is an embedding map, that preserves essential structure while capturing the full variability of the underlying data distribution \citep{jiang2026representation}.

Classical approaches to representation learning, such as principal component analysis (PCA) \citep{hastie2009elements, greenacre2022principal} and autoencoders (AE) \citep{rumelhart1986learning, hinton2006reducing}, focus primarily on reconstruction quality, often through minimizing mean squared error. While effective for capturing low-dimensional structure, these methods inherently target conditional means and therefore may fail to recover the full distribution of the data. As a result, reconstructed samples tend to underestimate variability and distort higher-order properties such as tail behavior \citep{shen2024distributional}.


Advances in deep generative modeling aim to address distributional fidelity by learning mappings from simple latent distributions to complex data distributions. Methods such as variational autoencoders (VAE) \citep{kingma2014autoencoding}, generative adversarial networks (GAN) \citep{goodfellow2014generative}, normalizing flows \citep{rezende2015variational}, diffusion models \citep{ho2020denoising}, and their numerous variants \citep{tolstikhin2018wasserstein, higgins2017betavae, pmlr-v70-arjovsky17a, pmlr-v139-nichol21a} have demonstrated strong performance in modeling complex distributions. However,
VAE- and GAN-based methods do not guarantee recovery of the conditional distribution $P(X \mid Z=e(X))$. Meanwhile, normalizing flows and diffusion models require the latent space to have the same dimension as the data space, making them less suitable for dimensionality reduction.
\citep{shen2024distributional}.

Distributional Principal Autoencoder (DPA) \citep{shen2024distributional, leban2025distributional} provides a promising alternative by directly targeting conditional distribution matching rather than mean reconstruction. In particular, DPA learns an encoder--decoder pair, based on energy scores \citep{gneiting2007strictly, szekely2023energy}, such that the decoder recovers the full conditional distribution of the data given the latent representation. While this ensures distributionally faithful reconstruction, the formulation is primarily tailored to numerical data and does not naturally extend to mixed-type settings with categorical variables, where the absence of a natural metric structure and the discrete nature of the data require fundamentally different modeling and objective design.


To address this issue, we propose to explicitly separate the modeling of numerical and categorical variables. Numerical variables naturally admit metric-based distributional objectives, such as energy scores,  whereas categorical variables are discrete and are more naturally modeled through likelihood-based objectives. Consequently, we seek a representation $Z=e(X)$ such that the conditional distributions of the numerical and categorical components given Z can be approximately modeled separately through $P(X^{\mathrm{num}} \mid Z)$ and $P(X^{\mathrm{cat}} \mid Z)$, allowing each component to be reconstructed using objectives tailored to its data type. This perspective is closely related to latent variable modeling, where heterogeneous observed variables are generated from some latent representation and modeled conditionally independently on that representation \citep{bishop2006pattern, koller2009probabilistic}. Following this perspective, we seek a representation $Z=e(X)$ under which the numerical and categorical components can be reconstructed through decoder objectives. The conditional-independence structure underlying this factorization is expected to become increasingly accurate as the dimension of $Z=e(X)$ grows. In practice, however, the representation is typically chosen to be relatively low-dimensional for compression and interpretability.

To this end, we introduce \emph{Conditional-Independence-Regularized Distributional Autoencoders}, a novel framework that augments distributional autoencoders with an explicit regularization mechanism encouraging conditional independence. Our approach leverages auxiliary predictive models to quantify the extent to which one component of the data provides additional predictive power beyond the learned representation, and penalizes such residual dependence. This can be interpreted as a probing mechanism that enforces structural constraints on the latent space.
The proposed method is tailored to mixed-type data. We design encoder and decoder architectures that jointly model numerical and categorical variables, while learning representations that balance distributional fidelity and structural disentanglement. The resulting objective combines (i) a distributional loss for numerical variables, (ii) a likelihood-based loss for categorical variables, and (iii) a conditional independence regularization term.
From a theoretical perspective, we show that the proposed objective preserves the desirable properties of distributional autoencoders while introducing a principled trade-off between distributional fidelity and structural independence. In particular, the optimal representation can be interpreted as balancing unexplained variability in the numerical component, uncertainty in the categorical component, and residual dependence between them.
Empirically, we demonstrate the effectiveness of our approach on both synthetic and real-world datasets. On simulated data, the method successfully recovers the conditional distributions of both numerical and categorical components while preserving the underlying dependence structure. On real datasets, it consistently improves categorical distribution recovery and achieves strong performance in preserving mixed-type associations compared to standard AE, VAE, and DPA baselines.

The remainder of the paper is organized as follows.  Section \ref{sec:Method} introduces the proposed framework and describe the corresponding learning objective and training procedure. 
Section \ref{sec:theory} presents the theoretical analysis of the proposed method.
Section \ref{sec:Experiments} reports experimental results on both synthetic and real-world datasets, and Section \ref{sec:Discussion} concludes with a discussion.


 
\section{Method} \label{sec:Method}

\subsection{Problem setup}

Let $X = (X^{\mathrm{num}}, X^{\mathrm{cat}})$ denote a random vector in 
$\mathcal{X}$, where $X^{\mathrm{num}}$ is a $d_{\mathrm{num}}$-dimensional numerical vector and $X^{\mathrm{cat}} = (X^{\mathrm{cat},1}, \dots, X^{\mathrm{cat},d_{\mathrm{cat}}})$ is a collection of $d_{\mathrm{cat}}$-categorical variables. Each $X^{\mathrm{cat},j}$ takes values in a finite set of size $K_j$.
We assume that the data are drawn from an unknown distribution $P^\ast$. Given a dataset $\{X_i\}_{i=1}^n$ sampled from $P^\ast$, our goal is to learn a representation map $e: \mathcal{X} \to \mathbb{R}^{d}$, together with a decoder $d$ that models the conditional distribution of $X$ given the latent representation $Z=e(X)$. Our approach is guided by two objectives:
 (i) {\it Distributional reconstruction.}
The learned decoder should accurately recover the conditional distribution of the data given the representation, so that the model captures the full conditional law $P (X \vert e(X) )$ rather than only conditional means.
 (ii) {\it Conditional independence regularization.}
The representation $Z$ should encourage approximate conditional independence between the numerical and categorical components, i.e., $X^{\mathrm{num}} \perp\!\!\!\perp X^{\mathrm{cat}} \mid Z = e(X)$. Intuitively, this encourages the representation to capture the relevant dependence between the numerical and categorical components while removing redundant interactions. As a result, the representation becomes more interpretable and the modeling task is simplified. We would like to clarify that the numerical--categorical partition is induced by the modeling objective rather than assumed to be an intrinsic structural property of the data-generating mechanism. Numerical and categorical variables usually require fundamentally different reconstruction objectives and decoder designs. Numerical variables naturally admit metric-based distributional objectives, whereas categorical variables are more naturally modeled through likelihood-based objectives. Consequently, we employ separate decoder components for the numerical and categorical variables in this work. This design allows the model to explicitly recover categorical conditional distributions, which is not naturally accommodated by the original DPA framework developed primarily for numerical data.

\subsection{Model architecture and learning objective}

We adopt an encoder--decoder architecture to model the conditional distribution of $X$ given the latent representation $Z = e(X)$.
The encoder $e(\cdot)$ maps the input $X = (X^{\mathrm{num}}, X^{\mathrm{cat}})$ to a latent representation $Z \in \mathbb{R}^{d}$, where both numerical and categorical variables are tokenized before being passed through a neural network. The decoder models the conditional distribution of $X$ given $Z$ and is decomposed as $d(Z, \varepsilon) = \big(d^{\mathrm{num}}(Z, \varepsilon),\; d^{\mathrm{cat}}(Z)\big)$,
where $\varepsilon$ is an auxiliary noise variable. The numerical component $d^{\mathrm{num}}(Z, \varepsilon)$ aims to capture variability in $X^{\mathrm{num}} \mid Z$, while the categorical component $d^{\mathrm{cat}}(Z)$ outputs probability vectors (via softmax) for categorical variables.

We next introduce the objective used to learn the encoder and decoder.


{\bf Reconstruction of numerical variables.} Following \citet{shen2024distributional}, we adopt a distributional loss based on energy scores to model the conditional distribution of $X^{\mathrm{num}}$ given $Z$. Let $\varepsilon$ and $\varepsilon'$ be independent noise variables from the standard Gaussian distribution. We define
\begin{equation}
\begin{aligned}
\mathcal{L}_{\mathrm{num}} 
= {} & \mathbb{E}_{X \sim P^\ast} \, \mathbb{E}_{\varepsilon} 
\big[ \| X^{\mathrm{num}} - d^{\mathrm{num}}(e(X), \varepsilon) \|_2^\beta \big] \\
& - \frac{1}{2} \, \mathbb{E}_{X \sim P^\ast} \, \mathbb{E}_{\varepsilon, \varepsilon'} 
\big[ \| d^{\mathrm{num}}(e(X), \varepsilon) - d^{\mathrm{num}}(e(X), \varepsilon') \|_2^\beta \big],
\end{aligned} \nonumber
\end{equation}
where $\beta \in (0,2)$ is a fixed hyperparameter. As shown in \citet{szekely2003statistics}, the energy score is a strictly proper scoring rule for probability distributions on Euclidean spaces, implying that the objective is minimized when the learned conditional distribution matches the true conditional distribution.

{\bf  Reconstruction of categorical variables.}
For categorical variables, we model the conditional distribution of $X^{\mathrm{cat}}$ given $Z$ using a likelihood-based objective. Let $d^{\mathrm{cat}}(\cdot \mid Z)$ denote the model distribution. We define
\begin{equation}
\mathcal{L}_{\mathrm{cat}} 
= \mathbb{E}_{X \sim P^\ast} \big[ - \log d^{\mathrm{cat}}(X^{\mathrm{cat}} \mid e(X)) \big]. \nonumber
\end{equation}
This objective corresponds to minimizing the conditional Kullback--Leibler (KL) divergence $\mathbb{E}_{Z} \big[ \mathrm{KL}\big( P(X^{\mathrm{cat}} \mid Z) \,\|\, d^{\mathrm{cat}}(\cdot \mid Z) \big) \big]$
up to an additive constant, and therefore encourages the decoder to recover the true conditional distribution of $X^{\mathrm{cat}}$. In practice, the specification of $d^{\mathrm{cat}}(\cdot \mid Z)$ may depend on the scientific context and available domain knowledge. A commonly used approach is to model the categorical variables conditionally independently given $Z$, with separate softmax outputs for each categorical component \citep{koller2009probabilistic, bishop2006pattern}.

{\bf Conditional independence regularization.}
To encourage the conditional independence structure, we introduce auxiliary predictive models that quantify the extent to which $X^{\mathrm{num}}$ provides additional information about $X^{\mathrm{cat}}$ beyond $Z$.
Specifically, let $q_\theta(X^{\mathrm{cat}} \mid Z)$ and $q_\psi(X^{\mathrm{cat}} \mid Z, X^{\mathrm{num}})$ denote two conditional models parameterized by neural networks. The first model predicts $X^{\mathrm{cat}}$ using only $Z$, while the second additionally has access to $X^{\mathrm{num}}$. We define the conditional independence regularizer as
\begin{equation}
    \mathcal{L}_{\mathrm{CInd}} = \inf_{\theta} \, \mathbb{E} \big[ - \log q_\theta(X^{\mathrm{cat}} \mid e(X)) \big] - \inf_{\psi} \, \mathbb{E} \big[ - \log q_\psi(X^{\mathrm{cat}} \mid e(X), X^{\mathrm{num}}) \big]. \nonumber
\end{equation}
Intuitively, this regularizer estimates the additional predictive information that $X^{\mathrm{num}}$ provides beyond $Z=e(X)$. Minimizing this gap encourages the representation to capture the relevant dependence between the numerical and categorical components, thereby promoting conditional independence. We note that the conditional independence assumption may not hold exactly in real-world datasets. In such cases, the proposed regularization should be viewed as encouraging approximate independence, trading off fidelity and structural simplicity. The strength of this trade-off is controlled by the regularization.

We emphasize that the proposed regularizer should be viewed as a soft structural inductive bias, rather than an assumption that conditional independence holds exactly in mixed-type datasets. From a representation-learning perspective, the conditional independence factorization is expected to become more accurate as the dimension of $Z=e(X)$ increases. In practice, however, $Z$ is typically chosen to be relatively low-dimensional for compression and interpretability. Consequently, the goal is not to enforce exact conditional independence, but rather to encourage the representation to capture the shared factors responsible for cross-type dependence. Moreover, since the optimization problem is highly nonconvex, the regularizer also serves as an inductive bias that favors representations exhibiting stronger conditional independence. As demonstrated by the sensitivity analysis and additional simulation studies in Appendix \ref{appendix:sensitivity_original_simulation} and Appendix \ref{appendix:structured_factor_simulation}, including the term $\mathcal{L}_{\mathrm{CInd}}$ properly often improves performance.

The final objective combines reconstruction and regularization:
\begin{equation}
    \mathcal{L}_{\mathrm{total}} = \lambda_{\mathrm{num}} \mathcal{L}_{\mathrm{num}} + \lambda_{\mathrm{cat}} \mathcal{L}_{\mathrm{cat}} + \lambda_{\mathrm{CInd}} \mathcal{L}_{\mathrm{CInd}}, \label{eq:population objective}
\end{equation}
where $\lambda_{\mathrm{num}}, \lambda_{\mathrm{cat}}, \lambda_{\mathrm{CInd}} > 0$ are tuning parameters with the constraint $\lambda_{\mathrm{num}} + \lambda_{\mathrm{cat}} + \lambda_{\mathrm{CInd}} = 1$. Note that the constraint $\lambda_{\mathrm{num}}+\lambda_{\mathrm{cat}}+\lambda_{\mathrm{CInd}}=1$ is imposed only for normalization since multiplying all three weights by a common positive constant leaves the optimizer unchanged. The normalization removes an unnecessary scaling degree of freedom and allows the hyperparameters to be interpreted as relative trade-off weights.
Overall, the objective balances three competing goals: accurate reconstruction of the numerical variables, likelihood-based modeling of the categorical variables, and enforcement of conditional independence through the auxiliary models.

\subsection{Training algorithm}


In practice, the population losses are approximated using empirical averages over the training data. Given a dataset $\{X_i\}_{i=1}^n$ and $m$ Monte Carlo samples of the noise variable $\epsilon$ for each observation, we define the empirical versions of $\mathcal{L}_{\mathrm{num}}, \mathcal{L}_{\mathrm{cat}}$, and $\mathcal{L}_{\mathrm{CInd}}$ as
\begin{align}
\widehat{\mathcal{L}}_{\mathrm{num}} 
& = \frac{1}{n} \sum_{i=1}^n \Bigg[
\frac{1}{m} \sum_{j=1}^m \left\| X_i^{\mathrm{num}} - d^{\mathrm{num}}(e(X_i), \varepsilon_{i,j}) \right\|_2^\beta \nonumber \\
&\qquad - \frac{1}{2m(m-1)} \sum_{j=1}^m \sum_{j' \neq j} \left\| d^{\mathrm{num}}(e(X_i), \varepsilon_{i,j}) - d^{\mathrm{num}}(e(X_i), \varepsilon_{i,j'}) \right\|_2^\beta
\Bigg], \nonumber \\
\widehat{\mathcal{L}}_{\mathrm{cat}} 
&= -\frac{1}{n} \sum_{i=1}^n \log d^{\mathrm{cat}} \big( X_i^{\mathrm{cat}} \mid e(X_i) \big), \nonumber \\
\widehat{\mathcal{L}}_{\mathrm{CInd}} 
&= - \sup_{\theta} \frac{1}{n} \sum_{i=1}^n \log q_\theta \big( X_i^{\mathrm{cat}} \mid e(X_i) \big)
+ \sup_{\psi} \frac{1}{n} \sum_{i=1}^n \log q_\psi \big( X_i^{\mathrm{cat}} \mid e(X_i), X_i^{\mathrm{num}} \big). \nonumber
\end{align}

The empirical objective is then given by $\widehat{\mathcal{L}}_{\mathrm{total}} = \lambda_{\mathrm{num}} \widehat{\mathcal{L}}_{\mathrm{num}}  + \lambda_{\mathrm{cat}} \widehat{\mathcal{L}}_{\mathrm{cat}} + \lambda_{\mathrm{CInd}} \widehat{\mathcal{L}}_{\mathrm{CInd}}$.
The training proceeds by alternating the following updates:
  
    (i). \textit{Update auxiliary models.} 
    Fix the encoder and decoder, and update $\theta$ and $\psi$ to maximize $\frac{1}{n} \sum_{i=1}^n \log q_\theta \big( X_i^{\mathrm{cat}} \mid e(X_i) \big)$ and $\frac{1}{n} \sum_{i=1}^n \log q_\psi \big( X_i^{\mathrm{cat}} \mid e(X_i), X_i^{\mathrm{num}} \big)$.

    (ii). \textit{Update encoder and decoder.}
    Fix $\theta$ and $\psi$, and update the parameters of $e(\cdot)$, $d^{\mathrm{num}}(\cdot)$, and $d^{\mathrm{cat}}(\cdot)$ to minimize $\widehat{\mathcal{L}}_{\mathrm{total}}$.
    
In practice, stochastic gradient descent with mini-batches is used for all updates. We typically perform multiple updates of the auxiliary models per update of the encoder and decoder to better approximate the inner optimization. In practice, the tuning parameters $\lambda_{\mathrm{num}}$, $\lambda_{\mathrm{cat}}$, and $\lambda_{\mathrm{CInd}}$ control the trade-off between numerical reconstruction, categorical reconstruction, and conditional independence regularization. These parameters can be selected using validation performance based on reconstruction quality and preservation of dependence structure.

\section{Theoretical analysis} \label{sec:theory}


In this section, we analyze the population objective \eqref{eq:population objective} and characterize the representations induced by the proposed framework at the population level. In particular, we study the optimal decoder and encoder under suitable expressivity assumptions, and show that the resulting representation admits a natural interpretation balancing conditional distribution reconstruction and conditional independence regularization.


{\bf Assumptions.}
We impose the following assumptions:
\begin{itemize}
    \item[(A1)] (\textit{Expressive decoder}) For any representation map $e(\cdot)$, there exist decoder functions such that $d^{\mathrm{num}}(Z, \varepsilon) \overset{d}{=} X^{\mathrm{num}} \mid Z$ and $d^{\mathrm{cat}}(\cdot \mid Z) = P(X^{\mathrm{cat}} \mid Z)$.

    \item[(A2)] (\textit{Expressive auxiliary models}) The function classes for $q_\theta$ and $q_\psi$ are sufficiently rich so that, for almost every input, $q_\theta(\cdot \mid Z) = P(X^{\mathrm{cat}} \mid Z)$ and $q_\psi(\cdot \mid Z, X^{\mathrm{num}}) = P(X^{\mathrm{cat}} \mid Z, X^{\mathrm{num}})$ can be attained.
    
    \item[(A3)] (\textit{Finite $\beta$-moments}) The numerical component satisfies $\mathbb{E} \| X^{\mathrm{num}} \|_2^\beta < \infty$. Moreover, for every admissible numerical decoder $d^{\mathrm{num}}$, it holds that $\mathbb{E} \| d^{\mathrm{num}}(z,\varepsilon) \|_2^\beta < \infty$ for $P_Z\text{-almost every } z$.
\end{itemize}

Assumptions (A1) and (A2) are expressivity assumptions ensuring that the decoder and auxiliary model classes are sufficiently rich to recover the corresponding conditional distributions at the population level. Similar population-level expressivity assumptions are commonly adopted, either explicitly or implicitly, in deep generative and distributional representation learning models, including \citet{kingma2014autoencoding,tolstikhin2018wasserstein}, and \citet{shen2024distributional}. In practice, these assumptions are approximated through sufficiently expressive neural network architectures. Assumption (A3) is a mild regularity condition required for the energy-score-based objective to be well defined and finite.

We first recall a key property of the energy score which states that the energy score is a strictly proper scoring rule.

\begin{lemma}[Propositions 1 and 2 in \citet{szekely2003statistics}] \label{lemma:energy score is a strictly proper scoring rule}
Let $P$ and $P^\ast$ be probability distributions on $\mathbb{R}^d$, and let $\beta \in (0,2)$. Suppose that
$\mathbb{E}_{X^\ast \sim P^\ast} \| X^\ast \|_2^\beta < \infty$ and $\mathbb{E}_{X \sim P} \| X \|_2^\beta < \infty$.
Then the quantity
$
\mathbb{E}_{X^\ast \sim P^\ast,\, X \sim P} \|X^\ast - X\|_2^\beta 
- \frac{1}{2} \mathbb{E}_{X, X' \sim P} \|X - X'\|_2^\beta
$ is minimized if and only if $P = P^\ast$,
where $X^\ast$ and $X$ are independent and $X'$ is an independent copy of $X$.
\end{lemma}

This lemma implies that $\mathcal{L}_{\mathrm{num}}$ is minimized when the conditional distribution induced by $d^{\mathrm{num}}(Z, \varepsilon)$ matches the true conditional distribution of $X^{\mathrm{num}}$ given $Z$.

We next characterize the optimal decoder.


\begin{proposition} \label{proposition:optimal decoder for any fixed encoder}
Under Assumptions (A1) and (A3), for any fixed encoder $e(\cdot)$, the corresponding optimal decoder $d^\ast = \left(d^{\mathrm{num}\,*},  d^{\mathrm{cat}\,*}\right) \in \arg \min_{d} \mathcal{L}_{\mathrm{total}}$ satisfies
$$
d^{\mathrm{num}\,*}(e(x), \varepsilon)  \overset{d}{=} X^{\mathrm{num}} \mid Z = e(x) \mbox{ and } d^{\mathrm{cat}\,*}(\cdot \mid e(x)) = P(X^{\mathrm{cat}} \mid Z = e(x))
$$
for $P^\ast$-almost every $x$.
\end{proposition}

The first statement follows from the strict propriety of the energy score, while the second follows from the fact that minimizing $\mathcal{L}_{\mathrm{cat}}$ is equivalent to minimizing the conditional KL divergence up to an additive term independent of the decoder.




The optimal representation is then characterized as follows.

\begin{theorem} \label{theorem:main}
Under Assumptions (A1)--(A3), let $(e^\ast,d^\ast)$ be a minimizer of $\mathcal{L}_{\mathrm{total}}
=
\lambda_{\mathrm{num}}\mathcal{L}_{\mathrm{num}}
+
\lambda_{\mathrm{cat}}\mathcal{L}_{\mathrm{cat}}
+
\lambda_{\mathrm{CInd}}\mathcal{L}_{\mathrm{CInd}},$
where $\lambda_{\mathrm{num}}, \lambda_{\mathrm{cat}}, \lambda_{\mathrm{CInd}} > 0$ are fixed constants with $\lambda_{\mathrm{num}} + \lambda_{\mathrm{cat}} + \lambda_{\mathrm{CInd}} = 1$.
Then the following statements hold:

\begin{itemize}
    \item[(i)] (\textit{Optimal decoder}) The decoder $d^\ast = \left(d^{\mathrm{num}\,*},  d^{\mathrm{cat}\,*}\right)$ recovers the true conditional distributions of the data given the representation in the sense that for $P^\ast$-almost every $x$,
    \begin{align*}
    d^{\mathrm{num}\,*}(e^\ast(x), \varepsilon) \overset{d}{=} X^{\mathrm{num}} \mid Z = e^\ast(x) \mbox{ and } 
    d^{\mathrm{cat}\,*}(\cdot \mid e^\ast(x)) = P(X^{\mathrm{cat}} \mid Z = e^\ast(x)).
    \end{align*}

    \item[(ii)] (\textit{Optimal encoder}) The optimal encoder $e^\ast$ is characterized by
    \begin{equation} \label{eq:optimal encoder}
    \begin{aligned}
    e^\ast \in \arg\min_e \Bigg\{
    &\lambda_{\mathrm{num}}
    \frac12
    \mathbb{E}_{X \sim P^\ast}
    \mathbb{E}_{Y,Y' \sim P(X^{\mathrm{num}} \mid e(X))}
    \|Y-Y'\|_2^\beta \\
    &\qquad
    +
    \lambda_{\mathrm{cat}} \,
    H(X^{\mathrm{cat}} \mid e(X))
    +
    \lambda_{\mathrm{CInd}} \,
    I(X^{\mathrm{num}}; X^{\mathrm{cat}} \mid e(X))
    \Bigg\},
    \end{aligned}        
    \end{equation}
    where $H(\cdot \mid \cdot)$ denotes the conditional entropy and $I\big( X^{\mathrm{num}} ; X^{\mathrm{cat}} \mid e(X) \big) = H(X^{\mathrm{cat}} \mid e(X)) - H(X^{\mathrm{cat}} \mid e(X), X^{\mathrm{num}})$ is the conditional mutual information.
\end{itemize}
\end{theorem}



Theorem \ref{theorem:main} provides a population-level characterization of the representation learned by the proposed objective. In contrast to classical autoencoder objectives that primarily target conditional means, the first term in the reduced objective \eqref{eq:optimal encoder} corresponds to unexplained conditional variability in the numerical component and arises naturally from the strict propriety of the energy score. The second term measures the residual uncertainty of the categorical variables given the representation through conditional entropy, while the third term explicitly penalizes residual dependence between numerical and categorical components via conditional mutual information. More broadly, theoretical analysis of representation learning for mixed-type data remains limited. Compared to the existing DPA framework \citep{shen2024distributional}, whose theoretical characterization focuses primarily on conditional distributional reconstruction for numerical variables, our result provides a richer characterization tailored to mixed-type settings by incorporating both categorical uncertainty and dependence disentanglement into the representation objective. As a result, the learned representation is characterized not only by conditional distributional reconstruction, but also by its ability to capture and separate dependence between heterogeneous variable types. 
All the proofs are provided in Appendix \ref{appendix:Proofs}.

\section{Experiments} \label{sec:Experiments}

To evaluate reconstruction performance, we compare the following methods: \textbf{CI-Reg} (our method), \textbf{CI-Reg-0} (an ablated version of our method using the same architecture but without the conditional independence regularizer), standard \textbf{AE} and \textbf{VAE} baselines, and a single-decoder \textbf{DPA} model proposed by \citet{shen2024distributional} that does not separately handle numerical and categorical components. We  point out that CI-Reg-0 already differs from DPA in two important respects: (i) the use of separate numerical and categorical decoder heads, and (ii) the likelihood-based objective for categorical variables. Our objective is representation learning for mixed-type data with conditional distribution recovery rather than synthetic data generation or joint density estimation. To the best of our knowledge, there is currently no established benchmark specifically targeting this problem. Consequently, we choose DPA, standard autoencoder baselines, and carefully designed ablations as the most directly relevant comparisons. This choice is also consistent with the original DPA paper, where conditional distribution recovery was compared against PCA and standard AE. We further include a diffusion-based generative model, DDPM \citep{ho2020denoising, song2021denoising}, as a representative modern generative baseline.

We report the average categorical total variation (TV) distance across all categorical components since accurately recovering categorical distributions is an important aspect of modeling mixed-type data and a main contribution of our work. We additionally report the overall energy distance, computed on the concatenation of the numerical variables and one-hot encoded categorical variables, as a measure of overall conditional distribution recovery. Due to space limitations, additional evaluations, including dependence-based evaluation (summarized by the relative Frobenius norm errors between the true and reconstructed association matrices), component-wise energy distances, and Kolmogorov--Smirnov (KS) statistics for numerical variables, are deferred to Appendix \ref{appendix:Additional results}.

\subsection{Simulations} \label{subsec:Simulations}

We consider a synthetic mixed-type latent variable model. Let $Z^\ast \in \mathbb{R}^{d_z}$ be the true latent factor sampled from a Gaussian mixture model with three components. The numerical vector is generated as $X^{\mathrm{num}} = b + A_1 Z^\ast + \sin(A_2 Z^\ast) + \varepsilon$, $\varepsilon \sim \mathcal{N}(0,\sigma^2 I)$,
which introduces nonlinear dependence on $Z^\ast$. For the categorical component, each variable $X_j^{\mathrm{cat}}$, $j=1,\dots,d_{\mathrm{cat}}$, is generated from $X_j^{\mathrm{cat}} \sim \mathrm{Categorical}(p_j)$ where $p_j = \mathrm{softmax}(W_j Z^\ast + b_j)$.
The observed sample is $X = \bigl(X^{\mathrm{num}}, X_1^{\mathrm{cat}}, \dots, X_{d_{\mathrm{cat}}}^{\mathrm{cat}}\bigr)$.
Within each scenario, the parameters $b$, $A_1$, $A_2$, $\sigma$, $W_j$, and $b_j$ are fixed across all 30 repetitions.
We consider true latent dimensions $d_z \in \{3,5,10\}$ and vary the number of categories per categorical variable as $K \in \{2,3,5,8\}$. Each dataset contains $15$ numerical variables and $5$ categorical variables, with training sample size $n=2000$. The encoder dimension is fixed at $d = 8$ for all methods. Further details are provided in Appendix \ref{appendix:Experimental details-Simulations}.



Figures \ref{fig:simulation_tv_grid} and \ref{fig:simulation_energy_grid} summarize the results.
For categorical TV distance, CI-Reg and CI-Reg-0 achieve the best performance in all the settings, demonstrating that  explicitly modeling categorical features substantially improves recovery of the categorical conditional distributions. For energy distance, CI-Reg remains highly competitive and is typically on par with or better than CI-Reg-0. These results indicate that enforcing conditional independence does not degrade overall conditional distribution matching or second-order structure, while still improving the discrete component.
In contrast, AE and VAE perform substantially worse across all metrics, reflecting their focus on mean reconstruction rather than conditional distribution recovery. The DPA baseline also underperforms compared to structured approaches, highlighting the benefit of separating numerical and categorical components.

To further assess the robustness of the proposed framework, we conduct two additional simulation studies, with details and results deferred to Appendix \ref{appendix:sensitivity_original_simulation} and Appendix \ref{appendix:structured_factor_simulation}. First, we perform a sensitivity analysis over the conditional independence regularization parameter $\lambda_{\mathrm{CInd}}$ on representative settings from the original simulation study. The results indicate that performance varies smoothly across a broad range of regularization strengths and that the proposed method is not overly sensitive to the precise choice of $\lambda_{\mathrm{CInd}}$. Second, we consider a more challenging simulation study with a structured latent-factor partition, where some latent factors affect only numerical variables, some affect only categorical variables, and others affect both. We additionally investigate latent-dimension misspecification by considering under-specified, correctly specified, and over-specified representation dimensions. Across all settings, CI-Reg continues to achieve strong performance on categorical, numerical, and mixed-type reconstruction metrics, while preserving the underlying latent cluster structure.

\begin{figure}[htbp]
\centering
\includegraphics[width=\textwidth]{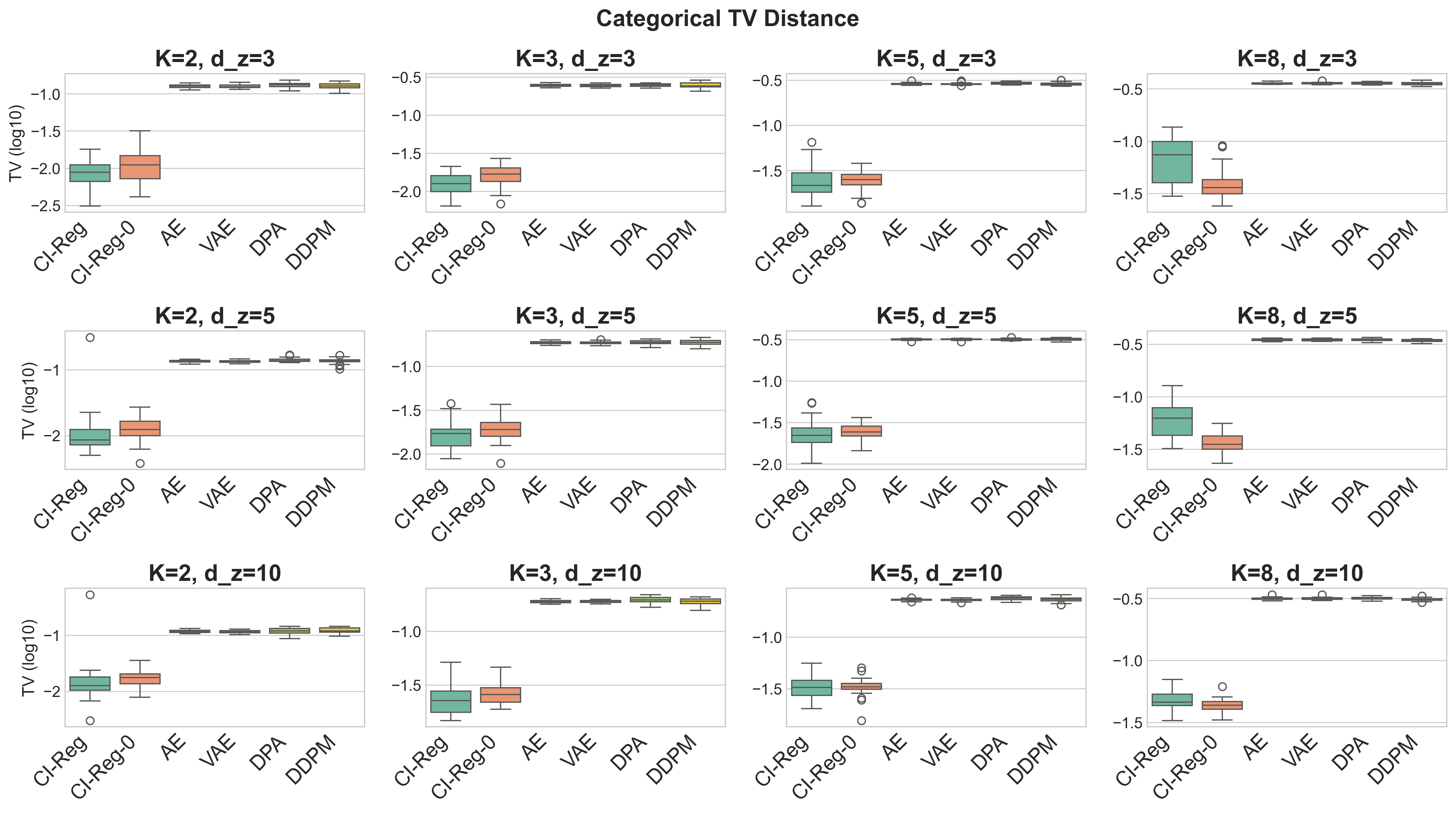}
\caption{
Categorical total variation distance (log-scale) across simulation settings. Lower values indicate more accurate recovery of the categorical conditional distributions. CI-Reg achieves the best performance in most settings.
}
\label{fig:simulation_tv_grid}
\end{figure}

\begin{figure}[htbp]
\centering
\includegraphics[width=\textwidth]{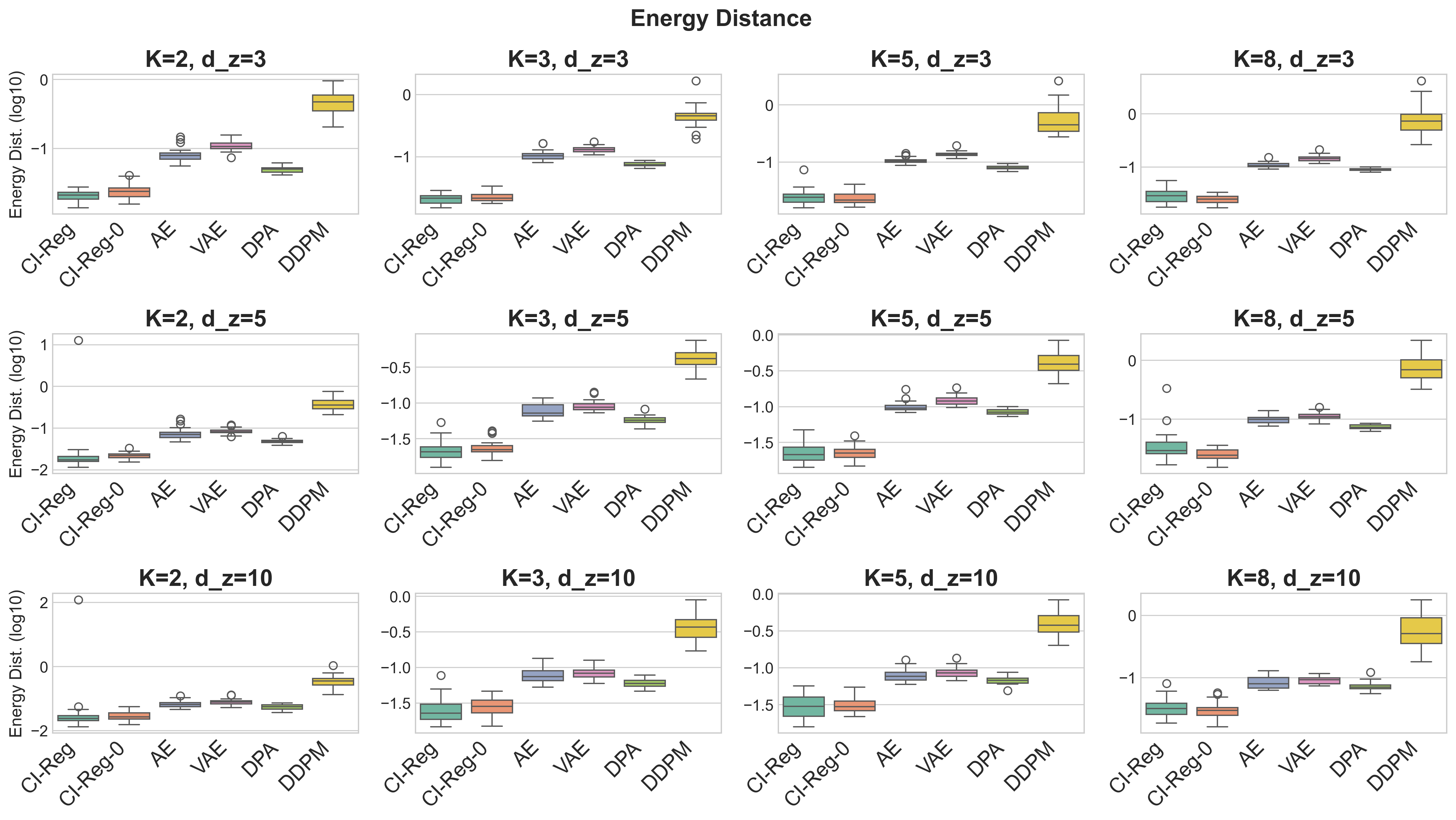}
\caption{
Overall energy distance (log-scale) between the true and learned conditional distributions across simulation settings. Lower values indicate better conditional distribution recovery.
}
\label{fig:simulation_energy_grid}
\end{figure}

\subsection{Real data experiments}

We consider 4 real-world public datasets consisting of both numerical and categorical features: Adult, Insurance, Shoppers, and PM2.5 \citep{uci_repository}. These datasets arise from diverse application domains including demographics, healthcare, online behavior, and environmental monitoring.

For model selection, we tune hyperparameters on the validation set, choosing the configuration that minimizes the combined association relative Frobenius error. The selected model is then retrained on the union of the training and validation sets and evaluated on the test set. All experiments are repeated over 10 independent runs with different random splits. Further details on the datasets and preprocessing procedures are deferred to Appendix \ref{appendix:Experimental details-Real data experiments}.


\begin{table}[htbp]
\caption{Categorical TV distance on real-world datasets (mean $\pm$ std). Lower is better.}
\label{tab:real_tv_distance}
\centering
\small
\setlength{\tabcolsep}{2pt}
\begin{tabular}{lcccccc}
\toprule
Dataset & CI-Reg & CI-Reg-0 & AE & VAE & DPA & DDPM \\
\midrule
Adult & $\mathbf{0.071 \pm 0.029}$ & $0.079 \pm 0.032$ & $0.385 \pm 0.001$ & $0.389 \pm 0.001$ & $0.385 \pm 0.001$ & $0.379 \pm 0.005$ \\
Insurance & $\mathbf{0.028 \pm 0.008}$ & $\mathbf{0.028 \pm 0.008}$ & $0.074 \pm 0.019$ & $0.071 \pm 0.012$ & $0.075 \pm 0.013$ & $0.080 \pm 0.024$ \\
Shoppers & $\mathbf{0.021 \pm 0.005}$ & $0.031 \pm 0.009$ & $0.365 \pm 0.004$ & $0.372 \pm 0.004$ & $0.365 \pm 0.004$ & $0.364 \pm 0.006$ \\
PM2.5 & $\mathbf{0.027 \pm 0.004}$ & $0.030 \pm 0.004$ & $0.041 \pm 0.002$ & $0.044 \pm 0.002$ & $0.041 \pm 0.002$ & $0.047 \pm 0.002$ \\
\bottomrule
\end{tabular}
\end{table}


\begin{table}[htbp]
\caption{Overall energy distance (mean $\pm$ std). Lower is better.}
\label{tab:real_energy_distance}
\centering
\small
\setlength{\tabcolsep}{2pt}
\begin{tabular}{lcccccc}
\toprule
Dataset & CI-Reg & CI-Reg-0 & AE & VAE & DPA & DDPM \\
\midrule
Adult & $\mathbf{0.047 \pm 0.034}$ & $0.060 \pm 0.044$ & $0.405 \pm 0.002$ & $0.423 \pm 0.002$ & $0.405 \pm 0.002$ & $0.420 \pm 0.011$ \\
Insurance & $0.013 \pm 0.003$ & $\mathbf{0.012 \pm 0.002}$ & $0.034 \pm 0.007$ & $0.090 \pm 0.012$ & $0.034 \pm 0.007$ & $0.056 \pm 0.024$ \\
Shoppers & $\mathbf{0.004 \pm 0.002}$ & $0.007 \pm 0.003$ & $0.227 \pm 0.005$ & $0.240 \pm 0.005$ & $0.229 \pm 0.006$ & $0.280 \pm 0.030$ \\
PM2.5 & $\mathbf{0.001 \pm 0.000}$ & $0.002 \pm 0.000$ & $0.009 \pm 0.000$ & $0.023 \pm 0.002$ & $0.009 \pm 0.000$ & $0.034 \pm 0.012$ \\
\bottomrule
\end{tabular}
\end{table}

Tables \ref{tab:real_tv_distance} and \ref{tab:real_energy_distance} summarize the reconstruction performance across the real-world datasets. Overall, CI-Reg and CI-Reg-0 consistently outperform AE, VAE, and DPA, indicating that the proposed split-decoder framework provides a substantial improvement in mixed-type conditional distribution recovery.

For categorical TV distance, CI-Reg and CI-Reg-0 achieve the best performance on all four datasets. In particular, the errors are substantially smaller than those of AE, VAE, and DPA, often by an order of magnitude. These results suggest that explicitly modeling categorical variables through a dedicated likelihood-based objective greatly improves recovery of the categorical conditional distributions.

For the real-world datasets considered in this paper, numerical features often have substantially different scales and variances. Therefore, all energy distances are computed after standardizing the numerical variables, ensuring that each numerical feature contributes equally to the metric and preventing variables with larger scales from dominating the evaluation. Under this evaluation protocol, a similar pattern is observed for the overall energy distance. CI-Reg achieves the lowest energy distance on three of the four datasets, while CI-Reg-0 performs best on the Insurance dataset. More importantly, both variants consistently outperform AE, VAE, and DPA by a large margin. This indicates that the proposed framework effectively recovers the overall conditional distribution of mixed-type data rather than only its conditional mean.

In contrast, AE and VAE exhibit substantially larger errors across both metrics. This behavior is consistent with their reconstruction objectives, which primarily target conditional means rather than full conditional distributions. The DPA baseline generally performs better than AE and VAE, reflecting the advantage of distributional reconstruction objectives, but still underperforms compared to the proposed framework designed specifically for mixed-type data.

\section{Discussion} \label{sec:Discussion}


We proposed Conditional-Independence-Regularized Distributional Autoencoders, a framework for representation learning and distribution recovery in mixed-type data. By combining distributional reconstruction objectives with an explicit conditional independence regularization mechanism, the proposed approach learns latent representations that preserve the conditional law of the data while encouraging structural disentanglement between numerical and categorical components. Theoretical analysis shows that the learned representation balances unexplained variability in numerical variables, uncertainty in categorical variables, and residual conditional dependence. Empirically, the method achieves strong performance across both synthetic and real-world datasets, consistently improving categorical distribution recovery, achieving competitive overall distribution recovery, and preserving mixed-type dependence structure.
From a computational perspective, CI-Reg-0 has training complexity comparable to DPA while consistently achieving substantially better reconstruction performance. We acknowledge that the conditional independence regularization in CI-Reg introduces additional computational overhead through the auxiliary models. However, the resulting performance improvements suggest that this additional cost can be worthwhile in practice.

Although the focus of this paper is conditional distribution recovery rather than downstream predictive performance, the learned representations may also be useful for subsequent tasks such as clustering, generation, and prediction. In particular, the additional structured-factor simulation study in Appendix \ref{appendix:structured_factor_simulation} shows that the learned embeddings successfully recover the latent cluster structure of the data (Figure \ref{fig:structured_factor_tsne_of_embeddings}). This suggests that preserving conditional distributions and dependence structure can simultaneously yield informative low-dimensional representations. A systematic investigation of downstream utility remains an interesting direction for future work.
Another interesting direction for future work is to extend the current conditional independence regularization beyond the numerical--categorical setting considered in this work. More generally, one could impose conditional independence constraints between multiple groups of variables. This extension may provide a principled way to incorporate domain knowledge into representation learning while further improving interpretability and structural disentanglement.



\bibliographystyle{plainnat}
\bibliography{references_current_version}

\newpage
\appendix


\section*{Appendix}

The appendix is organized as follows. Appendix \ref{appendix:Proofs} contains proofs of the theoretical results presented in Section \ref{sec:theory}. Appendix \ref{appendix:Experimental details} provides additional experimental details for both the simulation studies and real-world datasets, including model architectures, training procedures, and hyperparameter configurations. Appendix \ref{appendix:Additional results} presents additional experimental results, supplementary figures and tables, a sensitivity analysis with respect to the regularization parameter $\lambda_{\mathrm{CInd}}$,  as well as a more challenging simulation study with a structured latent-factor partition, which are omitted from the main text due to space limitations.

\section{Proofs} \label{appendix:Proofs}

In this appendix, we prove the theoretical results stated in Section \ref{sec:theory}. Throughout, let
\[
X = (X^{\mathrm{num}}, X^{\mathrm{cat}}) \sim P^\ast, 
\qquad Z = e(X),
\]
and recall the assumptions stated in Section \ref{sec:theory}.

\begin{proof}[Proof of Lemma \ref{lemma:energy score is a strictly proper scoring rule}] See Proposition 2 and its proof in \citet{szekely2003statistics}.
\end{proof}


\begin{proof}[Proof of Proposition \ref{proposition:optimal decoder for any fixed encoder}]
Fix an encoder $e(\cdot)$, and write $Z=e(X)$. We first consider the numerical decoder.

For $P_Z$-almost every latent value $z$, let $P_z^\ast$ denote a version of the conditional law of $X^{\mathrm{num}}$ given $Z=z$, and let $P_{d,z}$ denote the distribution induced by $d^{\mathrm{num}}(z,\varepsilon)$. Then
\begin{align*}
\mathcal{L}_{\mathrm{num}}
&=
\mathbb{E}_{Z}
\Bigg[
\mathbb{E}_{X^{\mathrm{num}} \sim P_z^\ast,\; Y \sim P_{d,z}}
\|X^{\mathrm{num}}-Y\|_2^\beta
-
\frac12
\mathbb{E}_{Y,Y' \sim P_{d,z}}
\|Y-Y'\|_2^\beta
\Bigg].
\end{align*}
Under Assumption (A3), we have $\mathbb{E}\|X^{\mathrm{num}}\|_2^\beta < \infty$, and hence for $P_Z$-almost every $z$, the conditional law $P_z^\ast$ has finite $\beta$-moment. Moreover, for any admissible decoder, the distribution $P_{d,z}$ also has finite $\beta$-moment. Therefore Lemma 1 applies.

By Lemma \ref{lemma:energy score is a strictly proper scoring rule}, for $P_Z$-almost every $z$, the inner quantity is minimized if and only if
\[
P_{d,z} = P_z^\ast,
\]
that is,
\[
d^{\mathrm{num}\,*}(z,\varepsilon)\overset{d}{=} X^{\mathrm{num}} \mid Z=z.
\]
Therefore,
\[
d^{\mathrm{num}\,*}(e(x),\varepsilon)\overset{d}{=} X^{\mathrm{num}} \mid Z=e(x)
\]
for $P^\ast$-almost every $x$.

Next consider the categorical decoder. For $P_Z$-almost every $z$, let
\[
p_z^\ast(\cdot) = P(X^{\mathrm{cat}}\mid Z=z),
\qquad
p_{d,z}(\cdot) = d^{\mathrm{cat}}(\cdot \mid z).
\]
Then
\begin{align*}
\mathcal{L}_{\mathrm{cat}}
&=
\mathbb{E}\big[-\log d^{\mathrm{cat}}(X^{\mathrm{cat}}\mid Z)\big] \\
&=
\mathbb{E}_{Z}
\Big[
\mathbb{E}_{X^{\mathrm{cat}} \sim p_z^\ast}
\big[-\log p_{d,z}(X^{\mathrm{cat}})\big]
\Big].
\end{align*}

For $P_Z$-almost every $z$, we decompose the inner expectation as
\begin{align*}
\mathbb{E}_{X^{\mathrm{cat}} \sim p_z^\ast}
\big[-\log p_{d,z}(X^{\mathrm{cat}})\big]
&=
\sum_x p_z^\ast(x)\, \big(-\log p_{d,z}(x)\big) \\
&=
\sum_x p_z^\ast(x)\log \frac{p_z^\ast(x)}{p_{d,z}(x)}
\;-\;
\sum_x p_z^\ast(x)\log p_z^\ast(x) \\
&=
\mathrm{KL}\big(p_z^\ast \,\|\, p_{d,z}\big)
+
H\big(p_z^\ast\big),
\end{align*}
where
\[
H(p_z^\ast) = - \sum_x p_z^\ast(x)\log p_z^\ast(x)
\]
denotes the entropy of $p_z^\ast$.

Therefore,
\[
\mathcal{L}_{\mathrm{cat}}
=
\mathbb{E}_{Z}
\Big[
H\big(p_z^\ast\big)
+
\mathrm{KL}\big(p_z^\ast \,\|\, p_{d,z}\big)
\Big].
\]
Since $H(p_z^\ast)$ does not depend on the decoder and the KL divergence is nonnegative, the loss is minimized if and only if
\[
p_{d,z}=p_z^\ast
\]
for $P_Z$-almost every $z$. Equivalently,
\[
d^{\mathrm{cat}\,*}(\cdot \mid e(x))
=
P(X^{\mathrm{cat}} \mid Z=e(x))
\]
for $P^\ast$-almost every $x$.
\end{proof}

The role of the conditional independence regularization is understood by the following proposition.

\begin{proposition} \label{proposition:role of the conditional independence regularization}
Under Assumption (A2), for any encoder $e(\cdot)$, the conditional independence term satisfies
\[
\mathcal{L}_{\mathrm{CInd}} 
= \inf_{\theta} \mathbb{E}\big[-\log q_\theta(X^{\mathrm{cat}} \mid Z)\big]
- \inf_{\psi} \mathbb{E}\big[-\log q_\psi(X^{\mathrm{cat}} \mid Z, X^{\mathrm{num}})\big]
= I\big( X^{\mathrm{num}} ; X^{\mathrm{cat}} \mid Z \big),
\]
where the conditional mutual information is defined as
\[
I\big( X^{\mathrm{num}} ; X^{\mathrm{cat}} \mid Z \big)
=
H(X^{\mathrm{cat}} \mid Z)
-
H(X^{\mathrm{cat}} \mid Z, X^{\mathrm{num}}),
\]
with $H(\cdot \mid \cdot)$ denoting conditional entropy.
\end{proposition}


\begin{proof}[Proof of Proposition \ref{proposition:role of the conditional independence regularization}] 

The proof relies on standard information-theoretic identities relating conditional entropy, KL divergence, and conditional mutual information; see, for example, \citet{10.1007/978-3-642-04138-9_30} and \citet{batina2011mutual}.

Fix an encoder $e(\cdot)$ and let $Z=e(X)$. For each $z$, denote
\[
p_z^\ast(\cdot) = P(X^{\mathrm{cat}} \mid Z=z).
\]

For any $\theta$, we have
\begin{align*}
\mathbb{E}\big[-\log q_\theta(X^{\mathrm{cat}} \mid Z)\big]
&=
\mathbb{E}_Z \Big[
\mathbb{E}_{X^{\mathrm{cat}} \sim p_z^\ast}
\big[-\log q_\theta(X^{\mathrm{cat}} \mid z)\big]
\Big].
\end{align*}

For each fixed $z$, we decompose the inner expectation as
\begin{align*}
\mathbb{E}_{X^{\mathrm{cat}} \sim p_z^\ast}
\big[-\log q_\theta(X^{\mathrm{cat}} \mid z)\big]
&=
\sum_x p_z^\ast(x)\, \big(-\log q_\theta(x \mid z)\big) \\
&=
\sum_x p_z^\ast(x)\log \frac{p_z^\ast(x)}{q_\theta(x \mid z)}
\;-\;
\sum_x p_z^\ast(x)\log p_z^\ast(x) \\
&=
\mathrm{KL}\big(p_z^\ast \,\|\, q_\theta(\cdot \mid z)\big)
+
H(p_z^\ast),
\end{align*}
where
\[
H(p_z^\ast) = -\sum_x p_z^\ast(x)\log p_z^\ast(x)
\]
denotes the entropy of $p_z^\ast$.

Therefore,
\begin{align*}
\mathbb{E}\big[-\log q_\theta(X^{\mathrm{cat}} \mid Z)\big]
=
\mathbb{E}_Z \Big[
H(p_z^\ast)
+
\mathrm{KL}\big(p_z^\ast \,\|\, q_\theta(\cdot \mid z)\big)
\Big].
\end{align*}

Since the KL divergence is nonnegative, we obtain the lower bound
\[
\mathbb{E}\big[-\log q_\theta(X^{\mathrm{cat}} \mid Z)\big]
\ge
\mathbb{E}_Z\big[H(p_z^\ast)\big]
=
H(X^{\mathrm{cat}} \mid Z),
\]
where the conditional entropy is defined as
\[
H(X^{\mathrm{cat}} \mid Z)
=
\mathbb{E}_Z\big[H(p_z^\ast)\big].
\]

By Assumption (A2), the model class for $q_\theta$ is sufficiently rich so that there exists $\theta^\ast$ satisfying
\[
q_{\theta^\ast}(\cdot \mid z) = p_z^\ast(\cdot)
\quad \text{for almost every } z.
\]
For such $\theta^\ast$, the KL divergence vanishes, and hence
\[
\inf_\theta \mathbb{E}\big[-\log q_\theta(X^{\mathrm{cat}} \mid Z)\big]
=
H(X^{\mathrm{cat}} \mid Z).
\]

An analogous argument yields
\[
\inf_\psi \mathbb{E}\big[-\log q_\psi(X^{\mathrm{cat}} \mid Z, X^{\mathrm{num}})\big]
=
H(X^{\mathrm{cat}} \mid Z, X^{\mathrm{num}}).
\]

Therefore,
\begin{align*}
\mathcal{L}_{\mathrm{CInd}}
&=
\inf_\theta \mathbb{E}\big[-\log q_\theta(X^{\mathrm{cat}} \mid Z)\big]
-
\inf_\psi \mathbb{E}\big[-\log q_\psi(X^{\mathrm{cat}} \mid Z,X^{\mathrm{num}})\big] \\
&=
H(X^{\mathrm{cat}} \mid Z)
-
H(X^{\mathrm{cat}} \mid Z,X^{\mathrm{num}}) \\
&=
I(X^{\mathrm{num}}; X^{\mathrm{cat}} \mid Z),
\end{align*}
where the last equality follows from the definition of conditional mutual information.

In particular,
\[
\mathcal{L}_{\mathrm{CInd}} \ge 0,
\]
with equality if and only if
\[
X^{\mathrm{num}} \perp\!\!\!\perp X^{\mathrm{cat}} \mid Z.
\]
\end{proof}


\begin{proof}[Proof of Theorem \ref{theorem:main}]
Let $(e^\ast,d^\ast)$ be a minimizer of
\[
\mathcal{L}_{\mathrm{total}}
=
\lambda_{\mathrm{num}}\mathcal{L}_{\mathrm{num}}
+
\lambda_{\mathrm{cat}}\mathcal{L}_{\mathrm{cat}}
+
\lambda_{\mathrm{CInd}}\mathcal{L}_{\mathrm{CInd}}.
\]

Fix an arbitrary encoder $e(\cdot)$ and denote $Z=e(X)$. We first minimize the objective over the decoder and auxiliary models.

By Proposition \ref{proposition:optimal decoder for any fixed encoder}, for each fixed $z$, the optimal numerical decoder induces the conditional distribution of $X^{\mathrm{num}}$ given $Z=z$. Substituting this into $\mathcal{L}_{\mathrm{num}}$, we obtain
\begin{align*}
\inf_d \mathcal{L}_{\mathrm{num}}
&=
\mathbb{E}_Z \Bigg[
\frac12
\mathbb{E}_{Y,Y' \sim P(X^{\mathrm{num}} \mid Z)}
\|Y-Y'\|_2^\beta
\Bigg] \\
&=
\frac12
\mathbb{E}_{X \sim P^\ast}
\mathbb{E}_{Y,Y' \sim P(X^{\mathrm{num}} \mid e(X))}
\|Y-Y'\|_2^\beta.
\end{align*}

Again by Proposition \ref{proposition:optimal decoder for any fixed encoder}, the optimal categorical decoder recovers the conditional distribution of $X^{\mathrm{cat}}$ given $Z$. Therefore,
\[
\inf_d \mathcal{L}_{\mathrm{cat}}
=
H(X^{\mathrm{cat}} \mid Z)
=
H(X^{\mathrm{cat}} \mid e(X)).
\]

Moreover, Proposition \ref{proposition:role of the conditional independence regularization} shows that
\[
\mathcal{L}_{\mathrm{CInd}}
=
I(X^{\mathrm{num}};X^{\mathrm{cat}}\mid Z)
=
I(X^{\mathrm{num}};X^{\mathrm{cat}}\mid e(X)).
\]

Combining the above, after optimizing over the decoder and auxiliary models, the encoder is characterized by the reduced objective
\begin{align*}
e^\ast \in \arg\min_e \Bigg\{
&\lambda_{\mathrm{num}}
\frac12
\mathbb{E}_{X \sim P^\ast}
\mathbb{E}_{Y,Y' \sim P(X^{\mathrm{num}} \mid e(X))}
\|Y-Y'\|_2^\beta \\
&\qquad
+
\lambda_{\mathrm{cat}} \,
H(X^{\mathrm{cat}} \mid e(X))
+
\lambda_{\mathrm{CInd}} \,
I(X^{\mathrm{num}}; X^{\mathrm{cat}} \mid e(X))
\Bigg\}.
\end{align*}

This shows that the optimal encoder balances three terms: unexplained variability in the numerical component, conditional entropy of the categorical component, and residual dependence between the two components.
\end{proof}

\section{Experimental details} \label{appendix:Experimental details}

\subsection{Simulations} \label{appendix:Experimental details-Simulations}

\paragraph{Evaluation protocol.}
To evaluate conditional distribution recovery, we use exactly the same latent variables $Z^\ast$ to generate both training and test samples, with independent noise realizations in the test set. Thus, the comparison isolates the discrepancy between the true conditional distribution $P(X \mid Z^\ast)$ and the learned conditional distribution induced by the decoder, rather than conflating it with latent mismatch.

\paragraph{Simulation setup.}
We consider a grid of simulation scenarios defined by varying the number of categories per categorical variable $K \in \{2,3,5,8\}$ and the true latent dimension $d_z \in \{3,5,10\}$, resulting in 12 distinct settings. Each dataset consists of $n=2000$ samples with $15$ numerical and $5$ categorical variables. The latent variable $Z^\ast$ is generated from a 3-component Gaussian mixture model with mixing proportions $(0.1, 0.2, 0.7)$, cluster separation $5.0$, and unit variance. Observed variables are generated via a nonlinear transformation of $Z^\ast$ with additive Gaussian noise.

The DGP parameters $b$, $A_1$, $A_2$, $\{W_j\}$, and $\{b_j\}$ are generated once per scenario and held fixed across all 30 repetitions. Specifically, the bias vector $b \in \mathbb{R}^{d_{\mathrm{num}}}$ is drawn from $\mathcal{N}(0, 4I)$, the linear loading matrix $A_1 \in \mathbb{R}^{d_{\mathrm{num}} \times d_z}$ from $\mathcal{N}(0, I)$, and the nonlinear loading matrix $A_2 \in \mathbb{R}^{d_{\mathrm{num}} \times d_z}$ from $\mathcal{N}(0, 0.25I)$. For each categorical feature $j$, the weight matrix $W_j \in \mathbb{R}^{d_z \times K_j}$ is drawn from $\mathcal{N}(0, 0.25I)$ and the bias vector $b_j \in \mathbb{R}^{K_j}$ from $\mathcal{N}(0, I)$. The noise standard deviation is fixed at $\sigma = 1.0$. Within each repetition, the latent variables $Z^\ast$ are resampled from the GMM and fresh noise is drawn for both the numerical and categorical feature generation.

We compare the following methods: AE, VAE, CI-Reg, CI-Reg-0, DPA, and DDPM. The first five methods are representation-learning approaches and therefore share the same encoder architecture with latent dimension fixed to $8$, while decoder architectures differ according to each method. DDPM is a generative baseline and does not employ an explicit low-dimensional latent representation.

\paragraph{Model architecture and training details.}
For AE, VAE, CI-Reg, CI-Reg-0, and DPA, we use a shared encoder consisting of a multilayer perceptron with three hidden layers of width 64 and ReLU activations, mapping inputs to an 8-dimensional latent representation. Decoder architectures vary by method: split-decoder models use separate networks for numerical and categorical variables, while DPA uses a unified decoder network.

All representation-learning methods are trained using Adam with learning rate $10^{-3}$ and batch size 128 for up to 200 epochs. For CI-Reg models, we consistently set raw (unnormalized) weights $\lambda_{\mathrm{num}} = 0.06$, $\lambda_{\mathrm{cat}} = 0.04$, and $\lambda_{\mathrm{CInd}} = 0.03$, with 100 auditor updates per training step. The diversity parameter is fixed to $m=100$, and the default choice $\beta=1$ is used throughout.

The DDPM implementation uses a standard denoising diffusion model with $T=1000$ diffusion steps, a linear noise schedule ranging from $10^{-4}$ to $0.02$, and a four-layer denoising MLP with hidden dimensions $(512,512,512,512)$. The model is trained using Adam with learning rate $10^{-3}$, batch size $128$, and early stopping. The DDPM configuration is fixed across all experiments in the work.


\subsection{Real data experiments} \label{appendix:Experimental details-Real data experiments}

\paragraph{Datasets.}
We briefly describe the real-world datasets used in our experiments:

\begin{itemize}
    \item \href{https://archive.ics.uci.edu/dataset/2/adult}{\textbf{Adult.}} The Adult Census Income dataset contains demographic and employment-related attributes of individuals. The original task is to predict whether an individual's annual income exceeds \$50{,}000.
    \item \href{https://www.kaggle.com/datasets/mirichoi0218/insurance}{\textbf{Insurance.}}The Insurance dataset contains demographic and health-related attributes, including age, sex, body mass index (BMI), number of children, smoking status, and region. The original task is to predict individual medical insurance charges.
    \item \href{https://archive.ics.uci.edu/dataset/468/online+shoppers+purchasing+intention+dataset}{\textbf{Shoppers.}} The Online Shoppers Purchasing Intention dataset contains information on users' web browsing sessions. The original task is to predict whether a session results in a purchase.
    \item \href{https://archive.ics.uci.edu/dataset/381/beijing+pm2+5+data}{\textbf{PM2.5.}} The PM2.5 dataset contains hourly measurements of PM2.5 concentration recorded at the U.S. Embassy in Beijing, along with meteorological variables from Beijing Capital International Airport. The original task is to predict the PM2.5 concentration.
\end{itemize}


\begin{table}[htbp]
\caption{Details of datasets used in the real-world experiments.}
\label{tab:dataset_stats}
\centering
\setlength{\tabcolsep}{4pt}
\begin{tabular}{lccccccc}
\toprule
Dataset & \# Total & \# Train & \# Validation & \# Test & \# Num & \# Cat \\
\midrule
Adult & 45,222 & 28,941 & 7,236 & 9,045 & 6 & 9 \\
Insurance & 1,338 & 749 & 321 & 268 & 4 & 3 \\
Shoppers & 12,330 & 7,891 & 1,973 & 2,466 & 10 & 8 \\
PM2.5 & 41,757 & 23,383 & 10,022 & 8,352 & 7 & 5 \\
\bottomrule
\end{tabular}
\end{table}

\paragraph{Experimental details.}

We remove any rows containing missing values. All experiments are repeated over 10 independent runs with different random splits; the statistics reported in Table \ref{tab:dataset_stats} correspond to each individual run.
Model selection is performed using the validation set, where we tune the conditional independence regularization parameter $\lambda_{\mathrm{CInd}}$ by minimizing the combined association relative Frobenius error. The selected model is then retrained on the union of the training and validation sets and evaluated on the test set.

We note that the $\lambda_{\mathrm{CInd}}$ values reported below correspond to the raw (unnormalized) weights used during hyperparameter tuning. The constraint $\lambda_{\mathrm{num}}+\lambda_{\mathrm{cat}}+\lambda_{\mathrm{CInd}}=1 $ is imposed only for normalization but is not necessary for model training. Since multiplying all three weights by a common positive constant leaves the optimizer unchanged, the normalization removes an unnecessary scaling degree of freedom and allows the hyperparameters to be interpreted as relative trade-off weights.

For the Adult dataset, we use a large-capacity architecture with latent dimension $d = 8$ and noise dimension $d_{\epsilon}=512$. The encoder consists of 3 hidden layers of width 512, while the split decoder uses deep networks for both numerical and categorical components. DPA baseline uses a comparable decoder architecture with 4 layers of width 512. Models are trained for up to 200 epochs using Adam with learning rate $10^{-3}$ and batch size 256, with early stopping (patience 30). We apply gradient clipping to stabilize training. The (unnormalized) conditional independence regularization parameter $\lambda_{\mathrm{CInd}}$ is selected from $\{0, 0.01, 0.03, 0.05, 0.1, 0.3\}$ based on validation performance, while $\lambda_{\mathrm{num}}=0.06$ and $\lambda_{\mathrm{cat}}=0.04$ are fixed.

For the Insurance dataset, we use a latent dimension $d = 16$ and $d_{\epsilon}=128$. The encoder consists of 4 hidden layers of width 256, and the decoder networks are correspondingly shallower and narrower. Models are trained for up to 200 epochs using Adam with learning rate $10^{-3}$ and batch size 64, with early stopping (patience 30). Gradient clipping is applied to stabilize training. The (unnormalized) conditional independence regularization parameter $\lambda_{\mathrm{CInd}}$ is tuned over $\{ 0, 0.01, 0.03, 0.05, 0.1, 0.3, 0.5 \}$ using the validation set, while $\lambda_{\mathrm{num}}=0.06$ and $\lambda_{\mathrm{cat}}=0.04$ are fixed.

For the Shoppers dataset, we use the same large-capacity architecture as in Adult, with latent dimension $d = 8$ and $d_{\epsilon}=512$. The encoder consists of 4 hidden layers of width 512, and the decoder networks follow the same deep configuration for both split and single-decoder variants. Models are trained for up to 200 epochs using Adam with learning rate $10^{-3}$ and batch size 256, with early stopping (patience 30). Gradient clipping is used to stabilize training. The (unnormalized) regularization parameter $\lambda_{\mathrm{CInd}}$ is selected from $\{ 0, 0.005, 0.01, 0.03, 0.05 \}$ based on validation performance, with $\lambda_{\mathrm{num}}=0.06$ and $\lambda_{\mathrm{cat}}=0.04$ fixed.

For the PM2.5 dataset, we use latent dimension $d = 8$ and noise dimension $d_{\epsilon}=512$. The encoder consists of 3 hidden layers of width 512, while the split decoder uses deep networks for both numerical and categorical components. DPA baselines and AE/VAE models use a comparable architecture with 4 hidden layers of width 512. Models are trained for up to 200 epochs using Adam with learning rate $10^{-3}$ and batch size 256, with early stopping (patience 30). Gradient clipping is applied to stabilize training. For CI-regularized models, the auditor networks are trained using learning rate $10^{-4}$ with 20 auditor updates per training step. The (unnormalized) conditional independence regularization parameter $\lambda_{\mathrm{CInd}}$ is selected from $\{ 0, 0.005, 0.01, 0.03, 0.05, 0.1 \}$ based on validation performance, while $\lambda_{\mathrm{num}}=0.06$ and $\lambda_{\mathrm{cat}}=0.04$ are fixed.

\section{Additional results} \label{appendix:Additional results}

\subsection{Additional evaluation metrics for the simulation in Section \ref{subsec:Simulations}} \label{appendix:Additional results-Simulations}



Figures \ref{fig:simulation_KS_grid}--\ref{fig:simulation_energy_cat_grid} provide additional simulation results evaluating conditional distribution recovery for both numerical and categorical components across all simulation settings.

Figure \ref{fig:simulation_KS_grid} reports the Kolmogorov--Smirnov (KS) statistic for numerical variables, computed as the average KS statistic across all numerical components. Overall, CI-Reg and CI-Reg-0 consistently achieve strong performance and remain highly competitive with the DPA baseline across different latent dimensions and category configurations. In most settings, all three distributional approaches outperform AE and VAE, indicating the importance of explicitly modeling conditional distributions rather than focusing solely on mean reconstruction.

To assess the overall dependence structure, we follow \citet{pmlr-v202-kotelnikov23a} and construct mixed-type correlation matrices using the Pearson correlation coefficient for numerical--numerical pairs, the correlation ratio for categorical--numerical pairs, and Theil's U statistic for categorical--categorical pairs. To ensure symmetry, Theil’s U is averaged across both directions. Performance is then summarized via relative Frobenius norm errors between the true and reconstructed association matrices.

Figure \ref{fig:simulation_combined_assoc_relative_frobenius_comparison_grid} reports the relative Frobenius norm between the true and learned conditional distributions. CI-Reg consistently ranks among the best-performing methods and typically achieves performance comparable to or better than DPA. In contrast, AE and VAE exhibit substantially larger relative Frobenius norms. These results further demonstrate that the proposed framework effectively recovers the full conditional distribution while handling mixed-type data.

These trends are further clarified in Figures \ref{fig:simulation_energy_cont_grid} and \ref{fig:simulation_energy_cat_grid}, which separately report the energy distances for the numerical and categorical components. Figure \ref{fig:simulation_energy_cont_grid} shows that CI-Reg and CI-Reg-0 remain highly competitive with DPA on the numerical component across nearly all simulation settings, while substantially outperforming AE and VAE. Figure \ref{fig:simulation_energy_cat_grid} demonstrates a much larger advantage for CI-Reg and CI-Reg-0 on the categorical component, particularly as the number of categories increases. In many settings, CI-Reg and CI-Reg-0 achieve dramatically smaller categorical energy distances than AE, VAE, and DPA, indicating that the proposed split-decoder framework together with the likelihood-based categorical objective substantially improves recovery of categorical conditional distributions.

\begin{figure}[htbp]
\centering
\includegraphics[width=\textwidth]{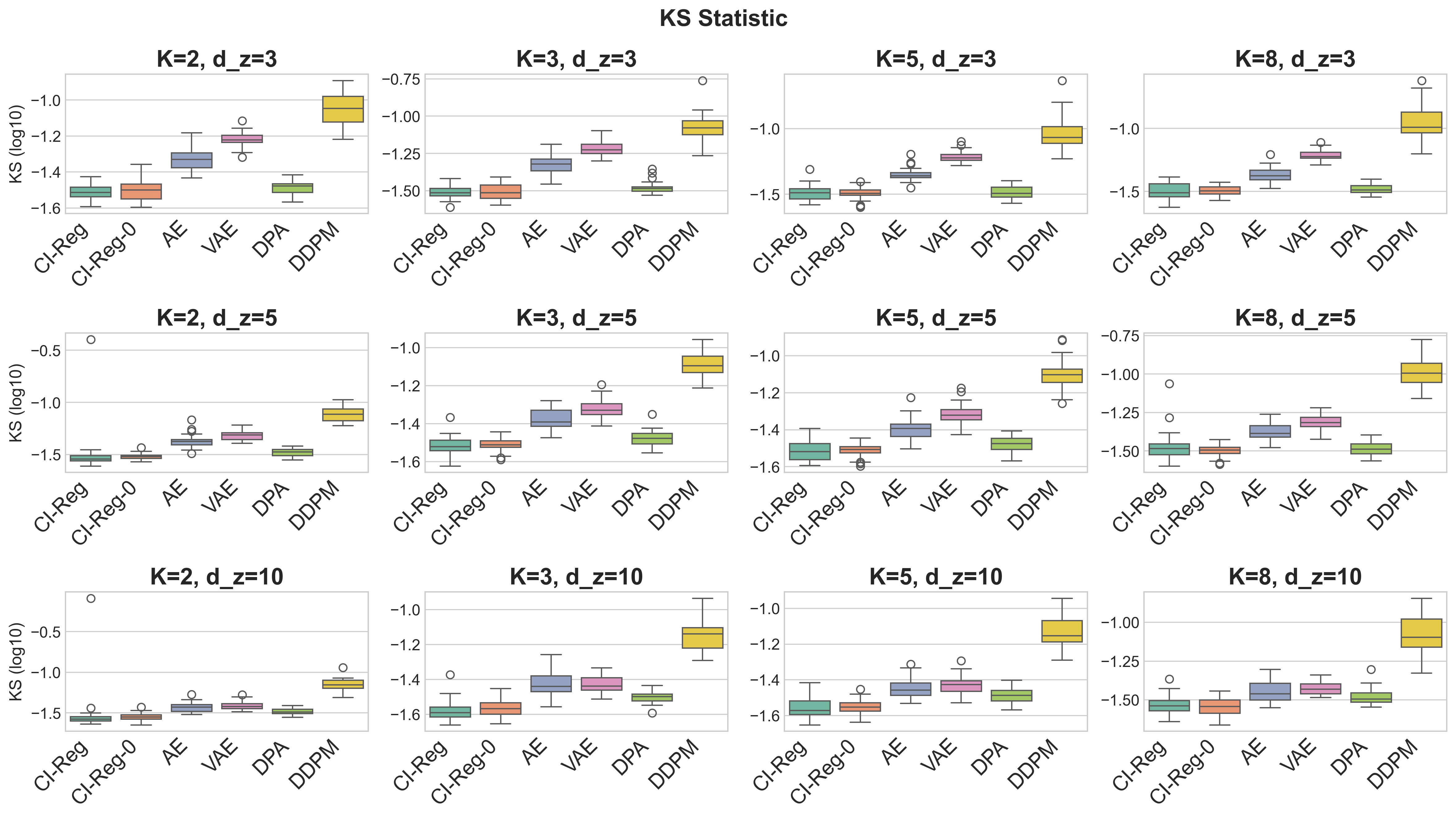}
\caption{
KS statistic for numerical components (column-wise, log-scale) between the true and learned conditional distributions across simulation settings. Lower values indicate better conditional distribution recovery. CI-Reg is consistently among the best-performing methods.
}
\label{fig:simulation_KS_grid}
\end{figure}

\begin{figure}[htbp]
\centering
\includegraphics[width=\textwidth]{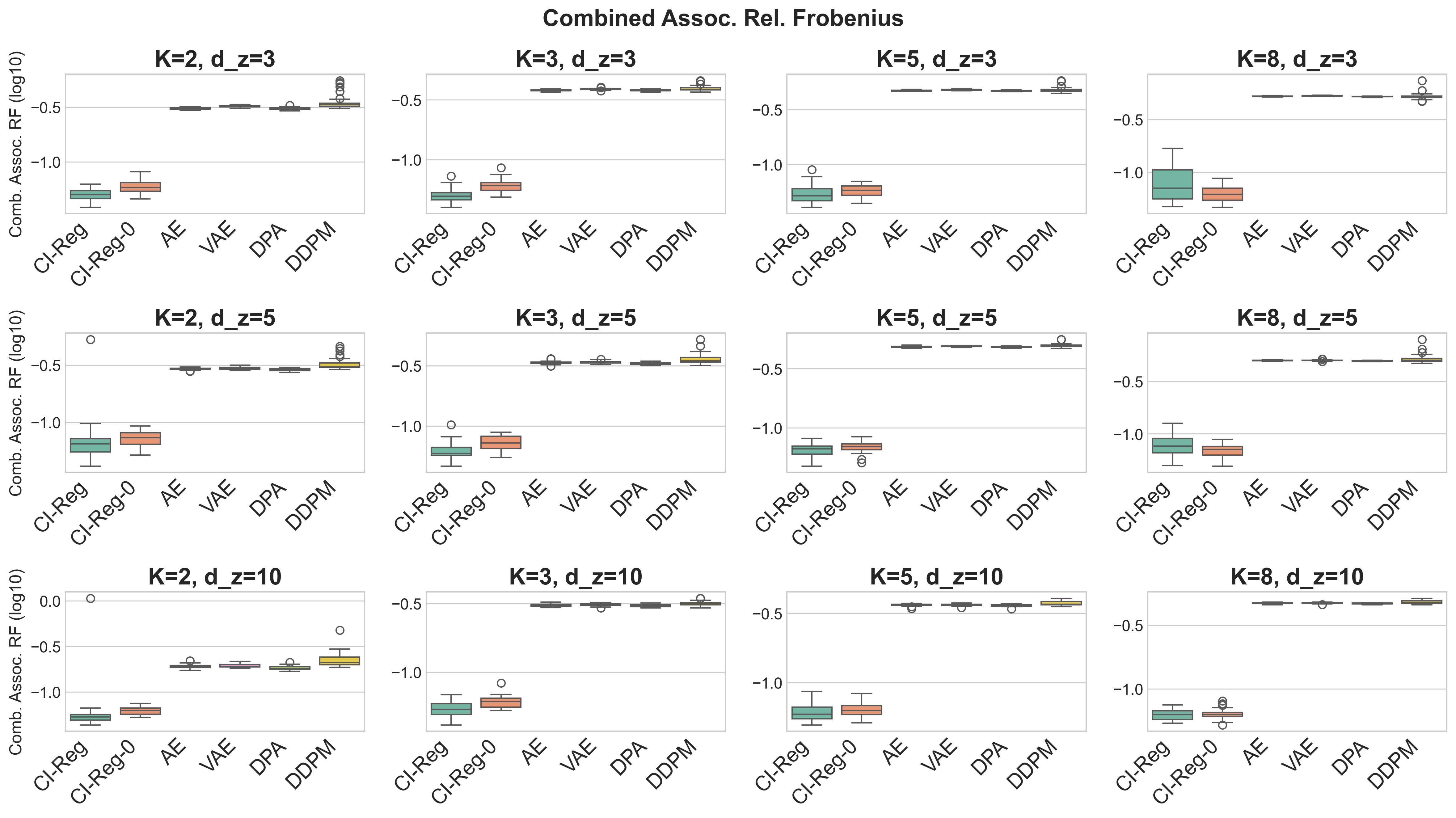}
\caption{
Relative Frobenius norm of the correlation matrix (log-scale) comparing covariance structure between the true and learned conditional distributions. Lower values indicate better recovery of second-order structure. CI-Reg performs competitively across all settings.
}
\label{fig:simulation_combined_assoc_relative_frobenius_comparison_grid}
\end{figure}

\begin{figure}[htbp]
\centering
\includegraphics[width=\textwidth]{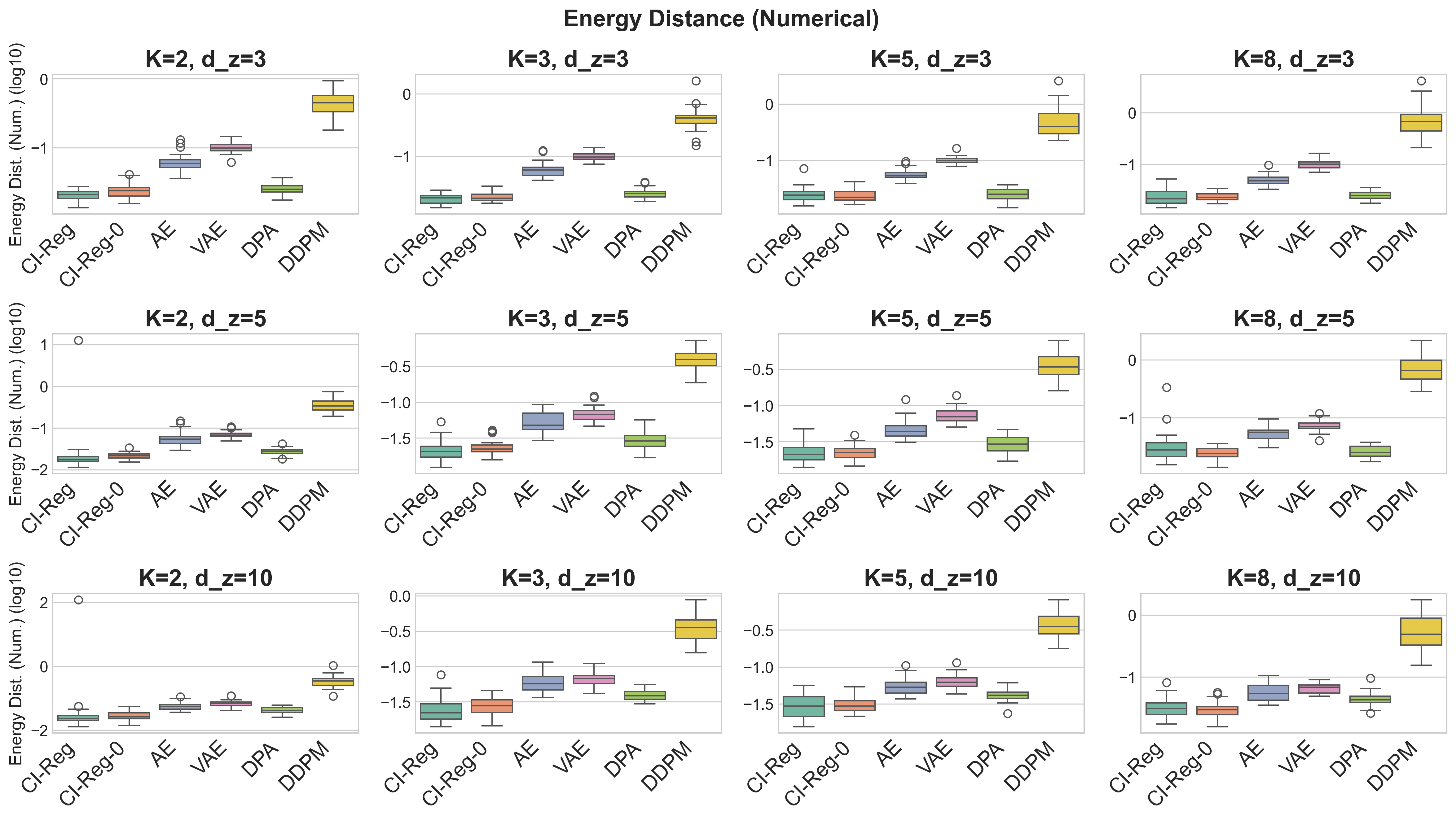}
\caption{
Energy distance for the numerical part (log-scale) between the true and learned conditional distributions across simulation settings. Lower values indicate better conditional distribution recovery.
}
\label{fig:simulation_energy_cont_grid}
\end{figure}

\begin{figure}[htbp]
\centering
\includegraphics[width=\textwidth]{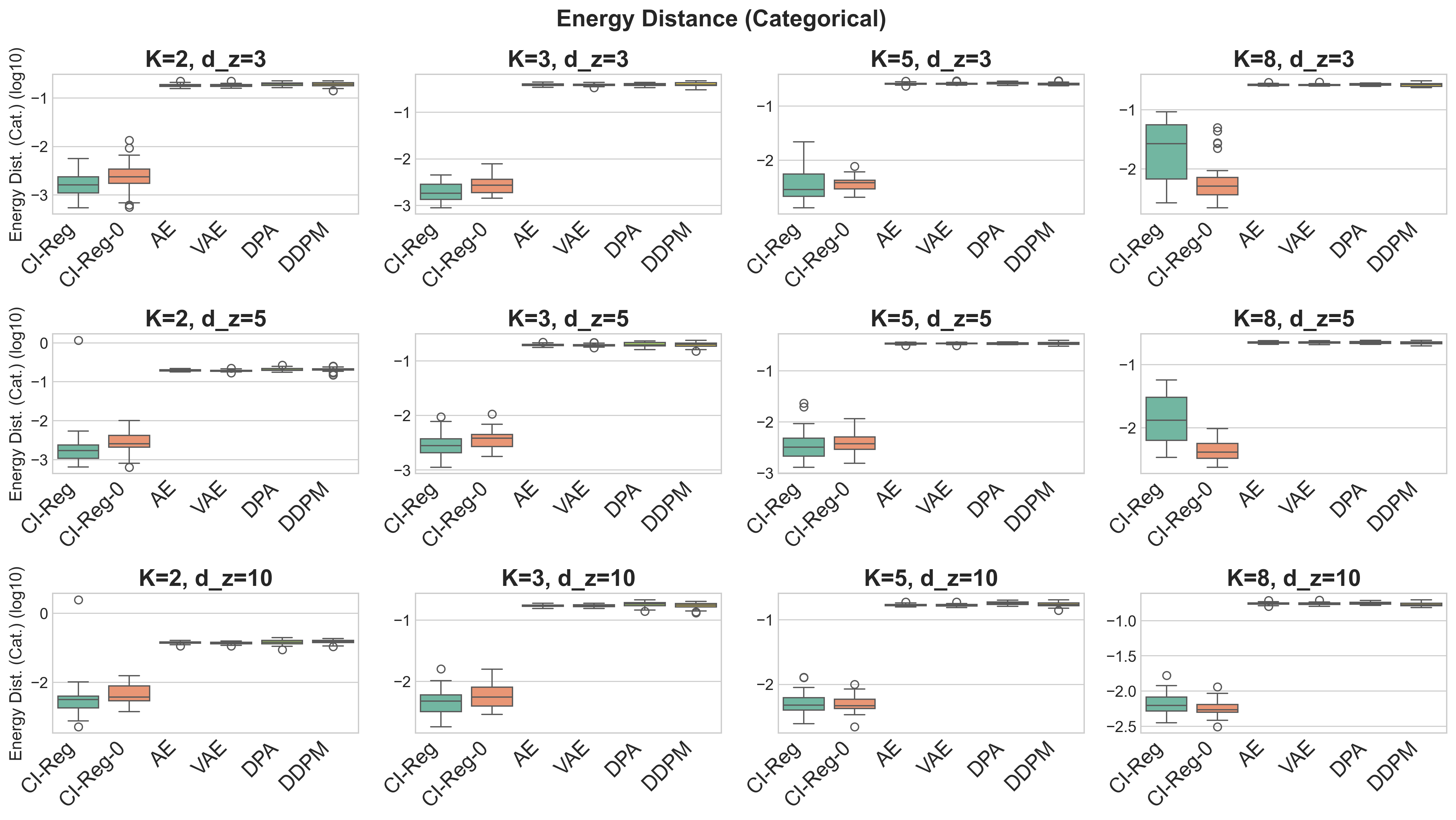}
\caption{
Energy distance for the categorical part between the true and learned conditional distributions across simulation settings. Lower values indicate better conditional distribution recovery.
}
\label{fig:simulation_energy_cat_grid}
\end{figure}

\subsection{Sensitivity analysis of $\lambda_{\mathrm{CInd}}$}
\label{appendix:sensitivity_original_simulation}

To better understand the effect of the conditional independence regularizer, we conduct a sensitivity analysis with respect to the regularization parameter $\lambda_{\mathrm{CInd}}$. The motivation is that conditional independence should be viewed as a soft structural inductive bias rather than a hard modeling assumption. The sensitivity analysis is performed on the simulation settings described in Section \ref{subsec:Simulations} and Appendix \ref{appendix:Experimental details-Simulations}. To provide a focused sensitivity analysis, we consider two representative scenarios: $K=2$, $d_z=3$ and $K=8$, $d_z=10$, corresponding to the simplest and most challenging settings considered in the original simulation study, respectively. For each scenario, we train the proposed method using $\lambda_{\mathrm{CInd}} \in \{0, 0.01, 0.03, 0.1, 0.3, 0.9\}$, while keeping all other hyperparameters fixed. Please note that the constraint $\lambda_{\mathrm{num}}+\lambda_{\mathrm{cat}}+\lambda_{\mathrm{CInd}}=1$ is imposed only for normalization and is not required for model training, since multiplying all three weights by a common positive constant leaves the optimizer unchanged. The number of independent repetitions is 10.

\begin{figure}[p]
\centering

\begin{subfigure}{0.48\textwidth}
\centering
\includegraphics[width=\textwidth]{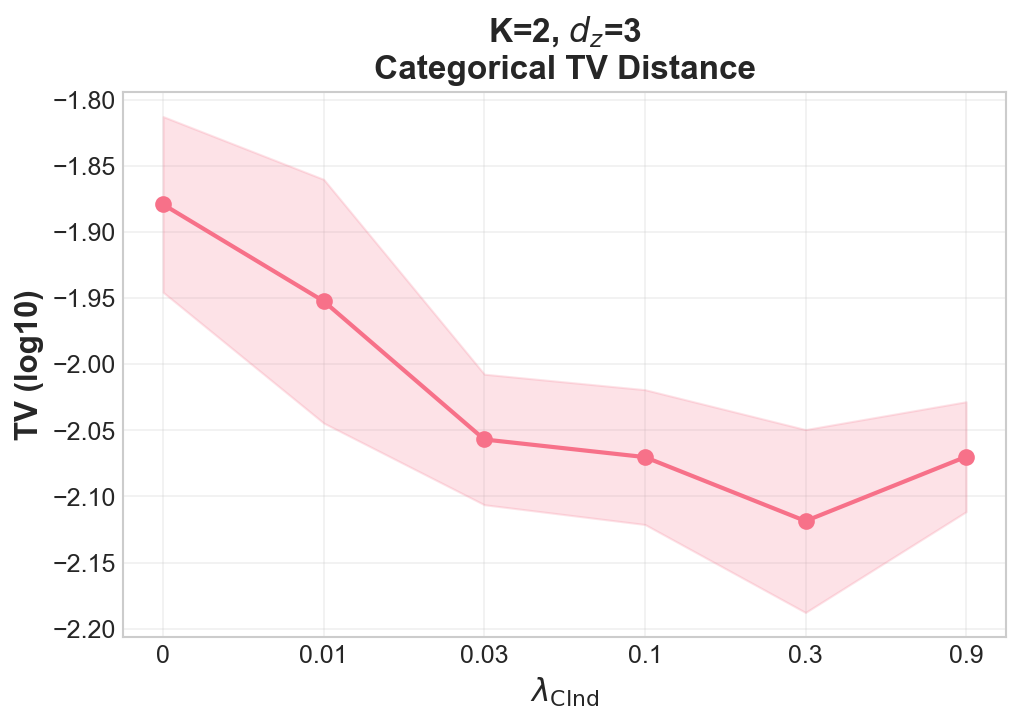}
\caption{Categorical TV distance}
\end{subfigure}
\hfill
\begin{subfigure}{0.48\textwidth}
\centering
\includegraphics[width=\textwidth]{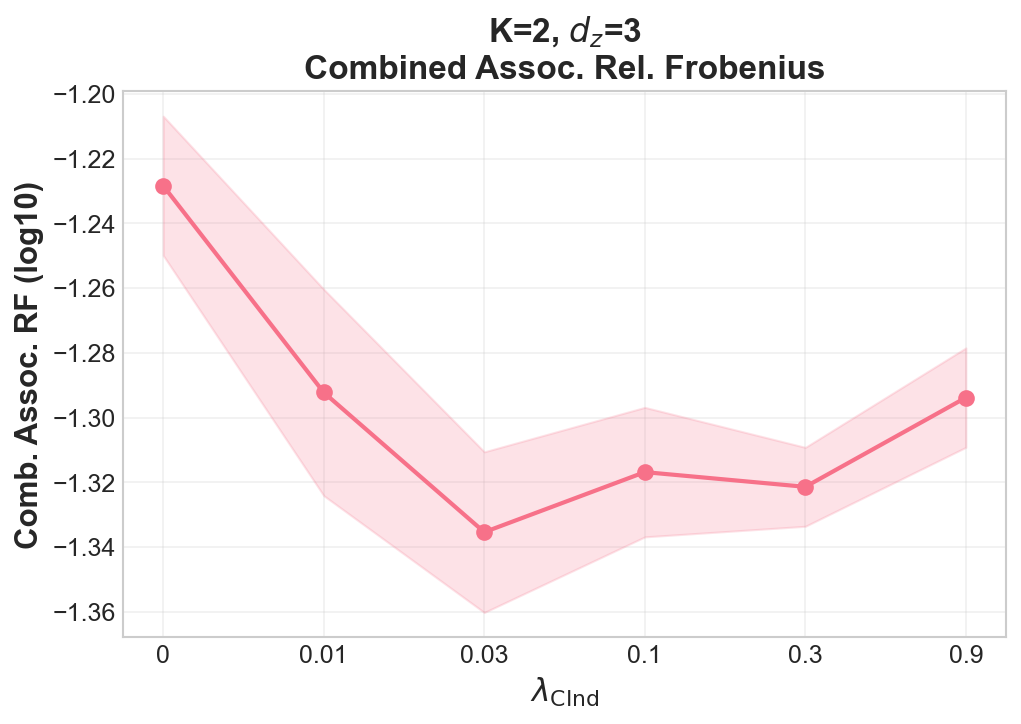}
\caption{Combined association relative Frobenius error}
\end{subfigure}

\vspace{0.5em}

\begin{subfigure}{0.48\textwidth}
\centering
\includegraphics[width=\textwidth]{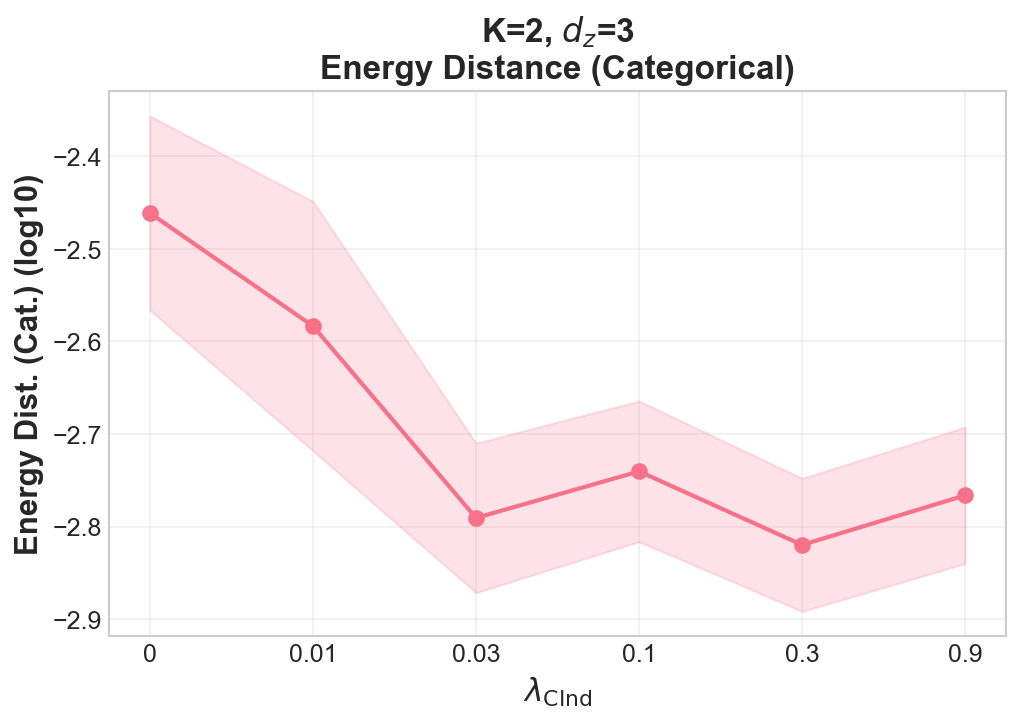}
\caption{Categorical energy distance}
\end{subfigure}
\hfill
\begin{subfigure}{0.48\textwidth}
\centering
\includegraphics[width=\textwidth]{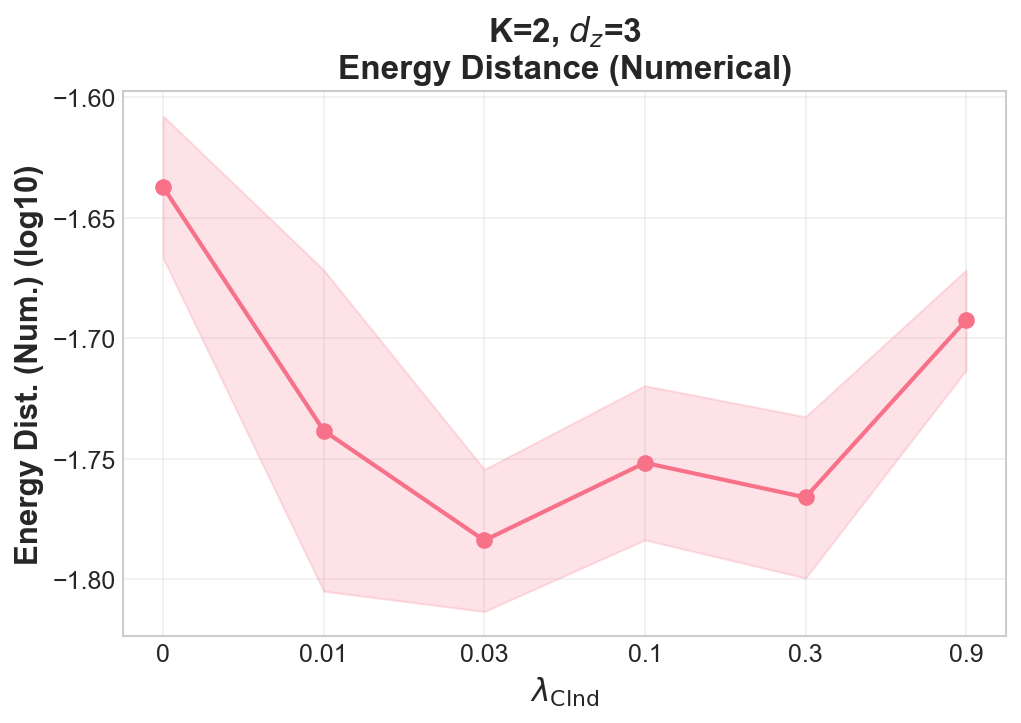}
\caption{Numerical energy distance}
\end{subfigure}

\vspace{0.5em}

\begin{subfigure}{0.48\textwidth}
\centering
\includegraphics[width=\textwidth]{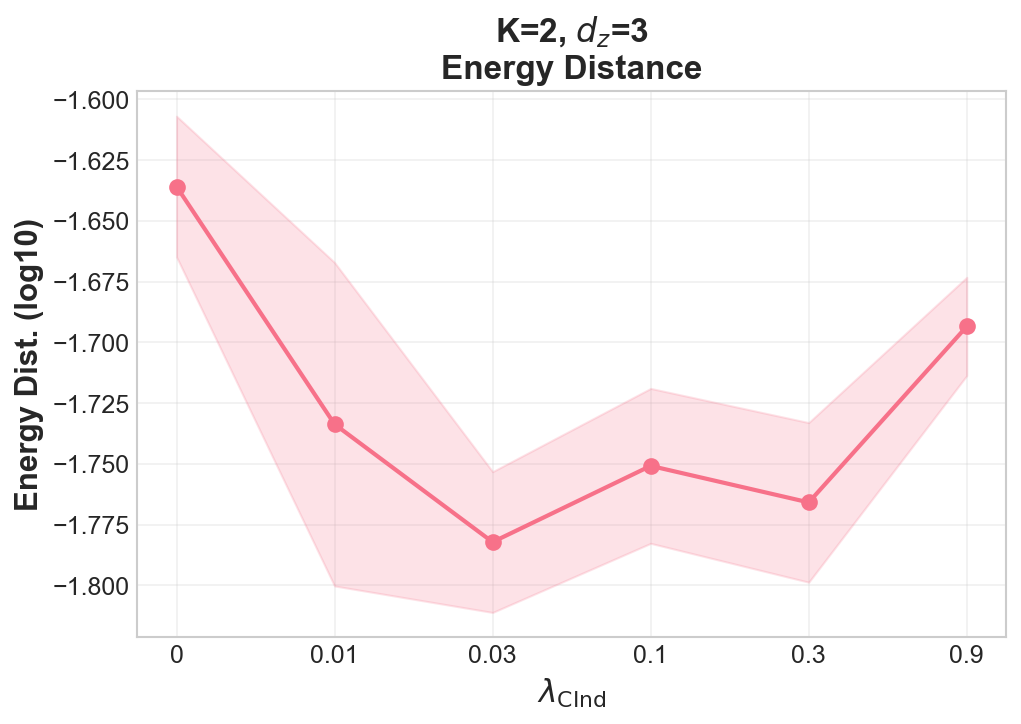}
\caption{Overall energy distance}
\end{subfigure}
\hfill
\begin{subfigure}{0.48\textwidth}
\centering
\includegraphics[width=\textwidth]{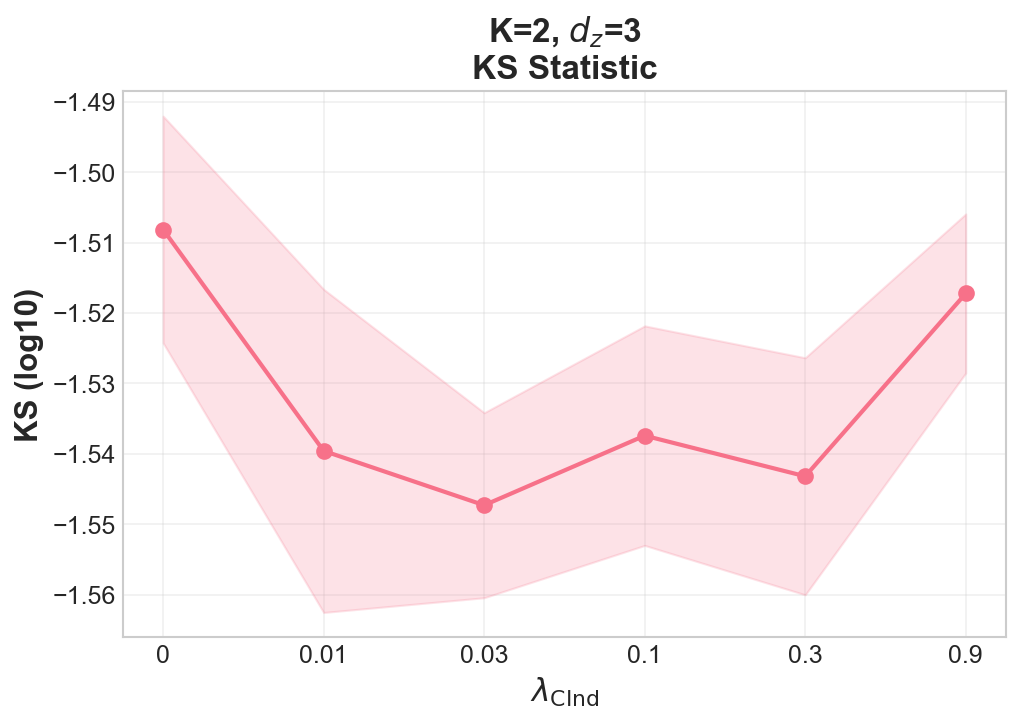}
\caption{Numerical KS statistic}
\end{subfigure}

\caption{
Sensitivity analysis with respect to $\lambda_{\mathrm{CInd}}$ for the setting $K=2$, $d_z=3$. Lower values indicate better performance. Solid lines show the mean across 10 independent repetitions and shaded regions indicate $\pm$ one standard deviation.
}
\label{fig:sensitivity_lambda_cind_K2}
\end{figure}

\begin{figure}[p]
\centering

\begin{subfigure}{0.48\textwidth}
\centering
\includegraphics[width=\textwidth]{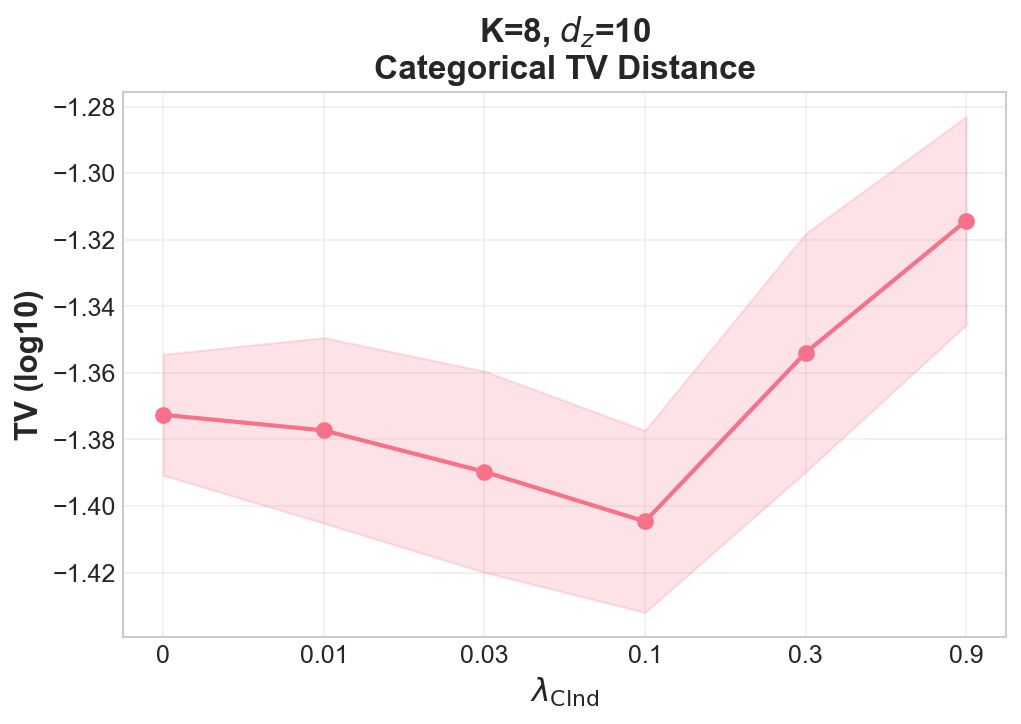}
\caption{Categorical TV distance}
\end{subfigure}
\hfill
\begin{subfigure}{0.48\textwidth}
\centering
\includegraphics[width=\textwidth]{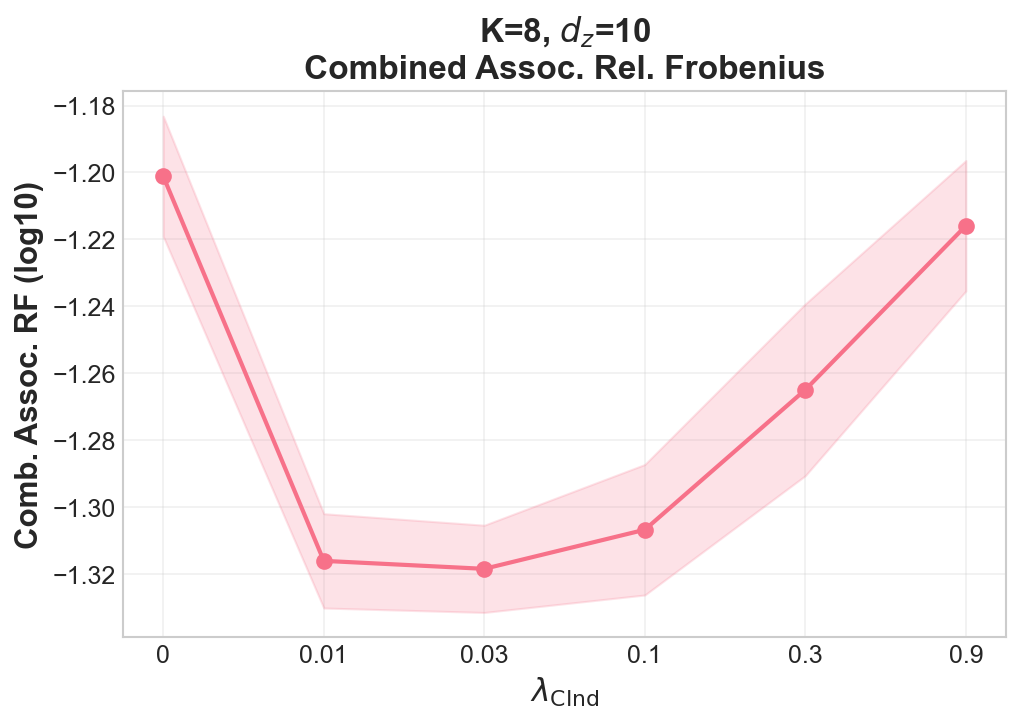}
\caption{Combined association relative Frobenius error}
\end{subfigure}

\vspace{0.5em}

\begin{subfigure}{0.48\textwidth}
\centering
\includegraphics[width=\textwidth]{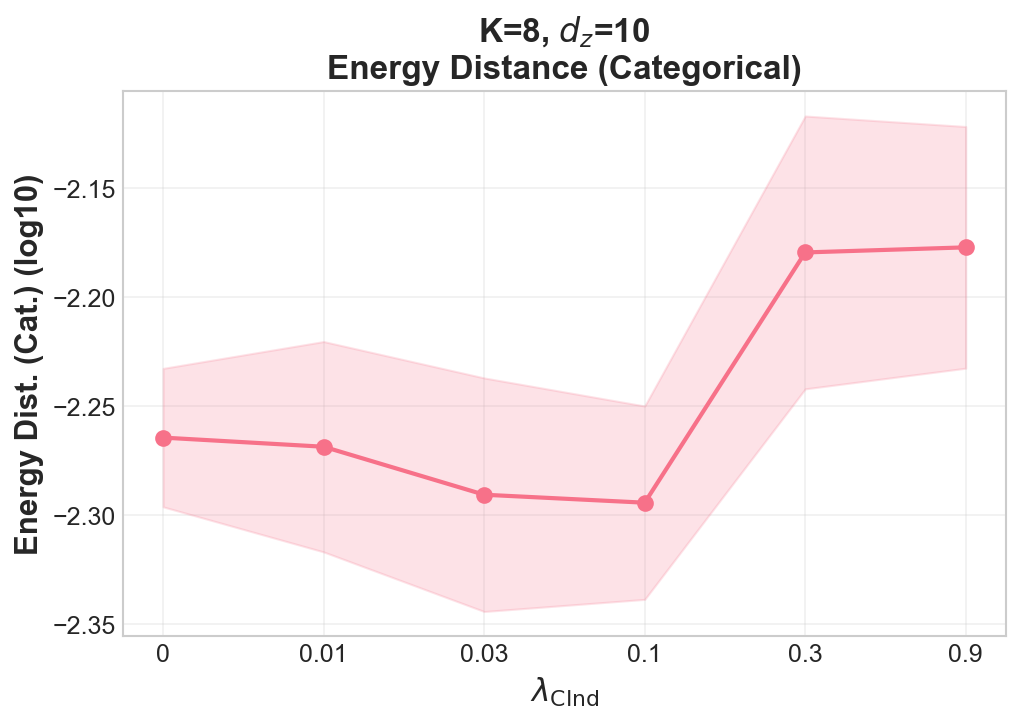}
\caption{Categorical energy distance}
\end{subfigure}
\hfill
\begin{subfigure}{0.48\textwidth}
\centering
\includegraphics[width=\textwidth]{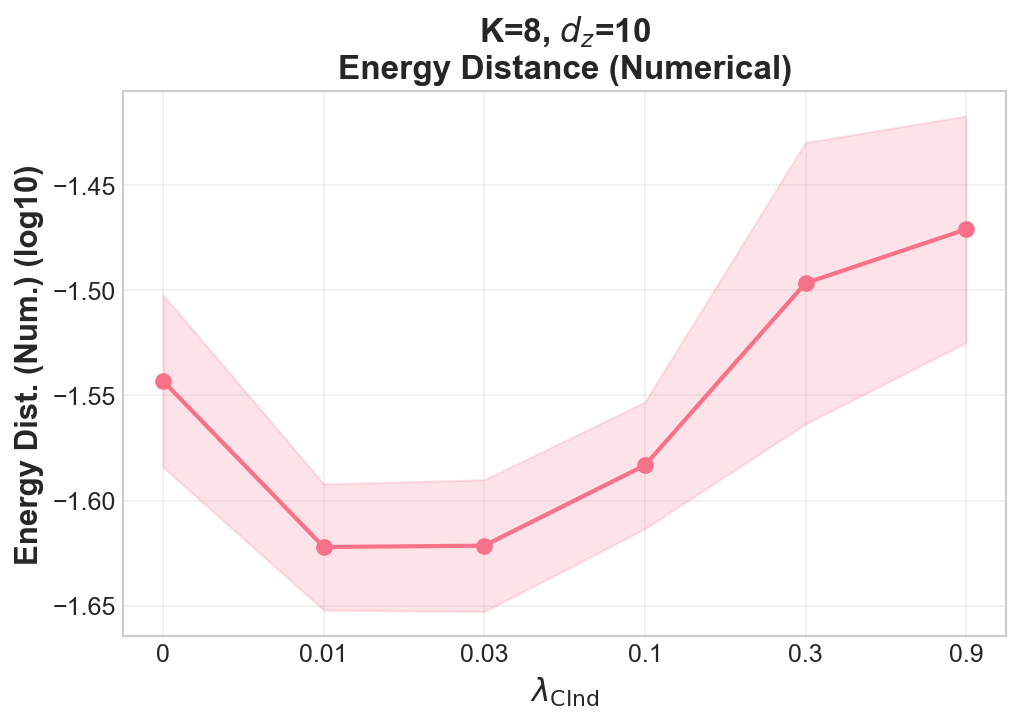}
\caption{Numerical energy distance}
\end{subfigure}

\vspace{0.5em}

\begin{subfigure}{0.48\textwidth}
\centering
\includegraphics[width=\textwidth]{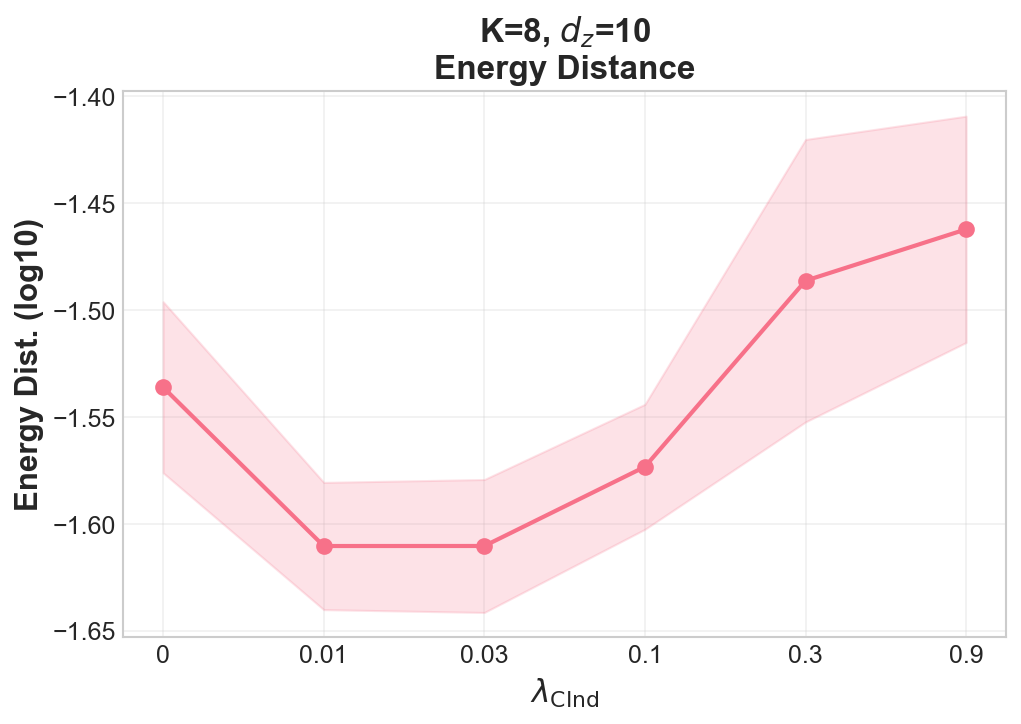}
\caption{Overall energy distance}
\end{subfigure}
\hfill
\begin{subfigure}{0.48\textwidth}
\centering
\includegraphics[width=\textwidth]{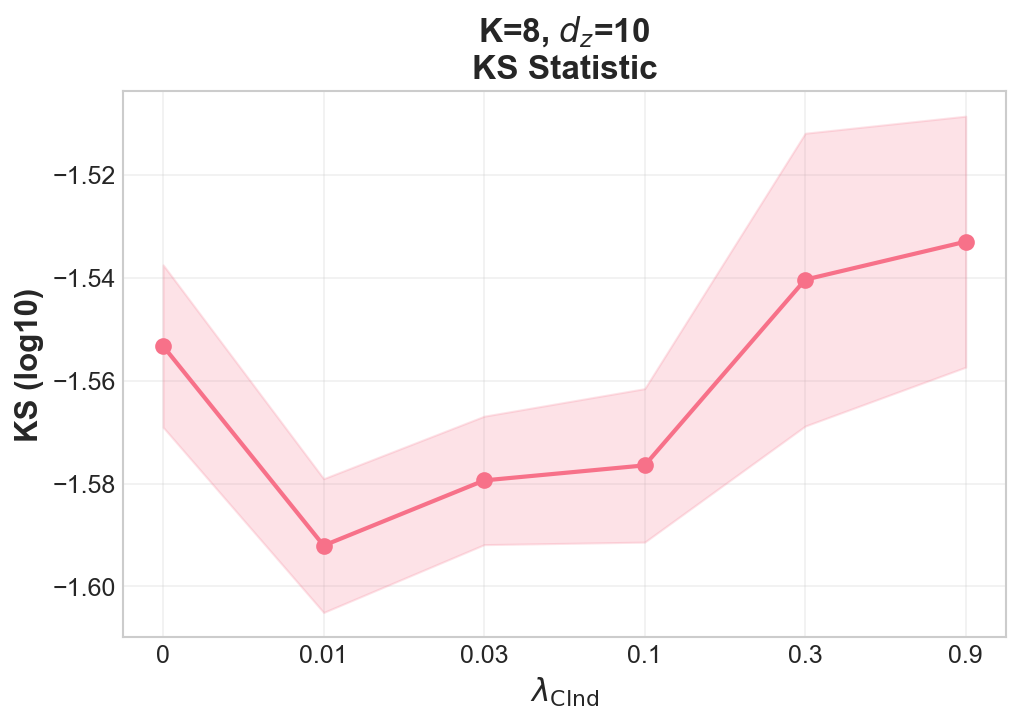}
\caption{Numerical KS statistic}
\end{subfigure}

\caption{
Sensitivity analysis with respect to $\lambda_{\mathrm{CInd}}$ for the setting $K=8$, $d_z=10$. Lower values indicate better performance. Solid lines show the mean across 10 independent repetitions and shaded regions indicate $\pm$ one standard deviation.
}
\label{fig:sensitivity_lambda_cind_K8}
\end{figure}

Figures~\ref{fig:sensitivity_lambda_cind_K2} and \ref{fig:sensitivity_lambda_cind_K8} summarize the effect of $\lambda_{\mathrm{CInd}}$ across several evaluation metrics. Overall, the results indicate that the proposed framework is relatively robust to the choice of $\lambda_{\mathrm{CInd}}$. Across both simulation settings, performance varies smoothly as the regularization strength changes, and a broad range of moderate values yields similar performance.

Several consistent patterns emerge. Introducing a moderate amount of conditional independence regularization generally improves categorical distribution recovery, as measured by both the categorical TV distance and the categorical energy distance. Moderate values of $\lambda_{\mathrm{CInd}}$ also tend to improve preservation of the mixed-type dependence structure, as measured by the combined association relative Frobenius error. At the same time, numerical reconstruction metrics, including the numerical energy distance and KS statistic, remain relatively stable over a wide range of regularization strengths.

Overall, these results suggest that the conditional independence regularizer acts as a useful structural inductive bias when applied moderately, while the proposed method is not overly sensitive to the precise choice of $\lambda_{\mathrm{CInd}}$.

\subsection{Additional structured-factor simulation study}
\label{appendix:structured_factor_simulation}

\paragraph{Simulation design.} We consider a more challenging simulation study where the synthetic mixed-type latent variable model has a structured factor partition. Let $Z^* \in \mathbb{R}^{d_z}$, $d_z = 15$, be the true latent factor, sampled from a 3-component Gaussian mixture model with mixing proportions $(0.1, 0.2, 0.7)$, cluster separation $5.0$, and unit within-cluster variance.

The 15 latent factors are partitioned into three groups of five:
\begin{itemize}
    \item $Z^*_{\mathrm{num}} = Z^*_{1:5}$: factors that affect \emph{only}
          the numerical variables;
    \item $Z^*_{\mathrm{cat}} = Z^*_{6:10}$: factors that affect \emph{only}
          the categorical variables;
    \item $Z^*_{\mathrm{both}} = Z^*_{11:15}$: factors that affect
          \emph{both} numerical and categorical variables.
\end{itemize}
The numerical vector is generated as
\[
    X^{\mathrm{num}} = b + A_1 Z^* + \sin(A_2 Z^*) + \varepsilon,
    \quad \varepsilon \sim \mathcal{N}(0, \sigma^2 I),
\]
where the loading matrices $A_1, A_2 \in \mathbb{R}^{d_{\mathrm{num}} \times d_z}$
have their columns corresponding to $Z^*_{\mathrm{cat}}$ (indices 6--10) set
to zero, so that the categorical-only factors do not influence the numerical
features. For each categorical variable $j = 1, \ldots, d_{\mathrm{cat}}$,
\[
    X^{\mathrm{cat}}_j \sim \mathrm{Categorical}
    \bigl(\mathrm{softmax}(W_j Z^* + b_j)\bigr),
\]
where the weight matrices $W_j \in \mathbb{R}^{d_z \times K_j}$ have their
rows corresponding to $Z^*_{\mathrm{num}}$ (indices 1--5) set to zero, so
that the numerical-only factors do not influence the categorical features.
The observed sample is
$X = (X^{\mathrm{num}}, X^{\mathrm{cat}}_1, \ldots, X^{\mathrm{cat}}_{d_{\mathrm{cat}}})$.

We consider encoder dimensions $d \in \{10, 15, 20\}$ to investigate the effect of latent dimension misspecification. Recall that the true latent dimension in the data-generating process is $d_z = 15$. Consequently, $d=10$ corresponds to an underestimated representation dimension, forcing the encoder to compress information into a lower-dimensional space than the true latent structure. In contrast, $d=15$ matches the true latent dimension and therefore represents the correctly specified setting. Finally, $d=20$ corresponds to an overestimated representation dimension, providing the encoder with additional capacity beyond that required by the data-generating process. Evaluating performance across these three settings allows us to assess the robustness of the competing methods to both under- and over-specification of the representation dimension, a situation commonly encountered in practice where the true latent dimension is unknown.

The DGP parameters $b$, $A_1$, $A_2$, $\{W_j\}$, and $\{b_j\}$ are generated once and held fixed across all 10 repetitions and all three encoder dimensions $d \in \{10, 15, 20\}$. Specifically, the bias vector $b \in \mathbb{R}^{d_{\mathrm{num}}}$ is drawn from $\mathcal{N}(0, 4I)$; the non-zero columns of the linear loading matrix
$A_1 \in \mathbb{R}^{d_{\mathrm{num}} \times d_z}$ are drawn from $\mathcal{N}(0, I)$; and the non-zero columns of the nonlinear loading matrix $A_2 \in \mathbb{R}^{d_{\mathrm{num}} \times d_z}$ are drawn from $\mathcal{N}(0, 0.25I)$. For each categorical feature $j$, the non-zero rows of the weight matrix $W_j \in \mathbb{R}^{d_z \times K_j}$ are drawn from $\mathcal{N}(0, 0.25I)$, and the bias vector $b_j \in \mathbb{R}^{K_j}$ is drawn from $\mathcal{N}(0, I)$. The noise standard deviation is fixed at $\sigma = 1.0$. Each dataset consists of $n = 2{,}000$ samples with $d_{\mathrm{num}} = 20$ numerical and $d_{\mathrm{cat}} = 10$ categorical variables, each taking $K = 3$ values. Within each repetition, the latent
variables $Z^*$ are resampled from the GMM and fresh noise is drawn for both numerical and categorical feature generation.

For CI-Reg models, we consider $\lambda_{\mathrm{num}} = 0.06$, $\lambda_{\mathrm{cat}} = 0.04$, and $\lambda_{\mathrm{CInd}} \in \left\{0, 0.01, 0.03, 0.1, 0.3, 0.9 \right\}$. Please note again that the constraint $\lambda_{\mathrm{num}}+\lambda_{\mathrm{cat}}+\lambda_{\mathrm{CInd}}=1$ is imposed only for normalization but is not necessary for model training since multiplying all three weights by a common positive constant leaves the optimizer unchanged. For each considered $\lambda_{\mathrm{CInd}}$ value, we denote our proposed method as CI-Reg-$\lambda_{\mathrm{CInd}}$ (e.g., CI-Reg-0, CI-Reg-0.01, CI-Reg-0.03, etc.).

\begin{figure}[p]
\centering

\begin{subfigure}{\textwidth}
    \centering
    \includegraphics[width=0.8\textwidth]{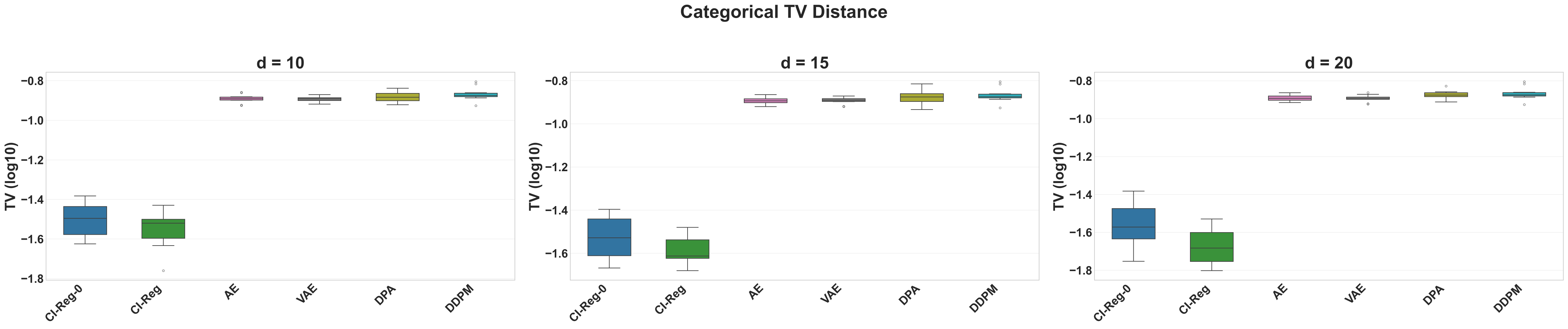}
    \caption{Categorical TV distance. Lower is better.}
\end{subfigure}

\vspace{0.5em}

\begin{subfigure}{\textwidth}
    \centering
    \includegraphics[width=0.8\textwidth]{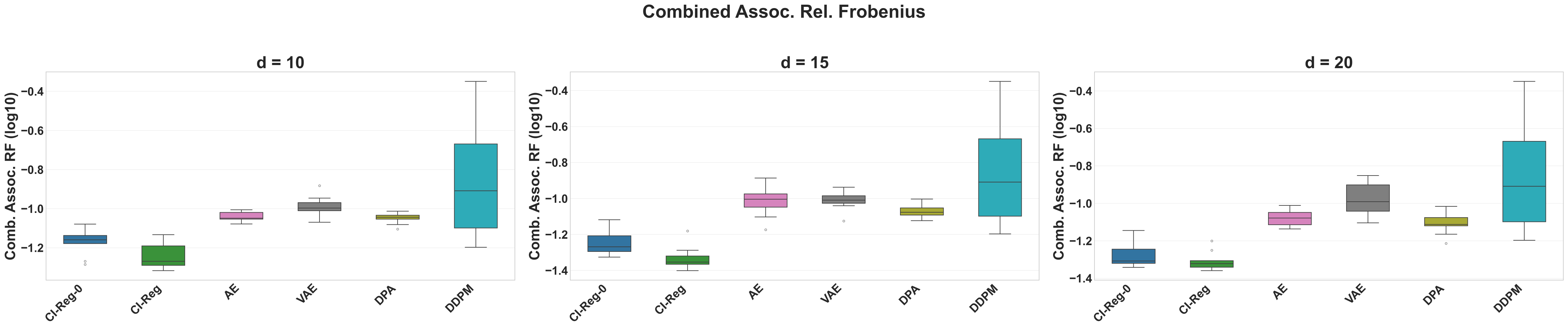}
    \caption{Combined association relative Frobenius error. Lower is better.}
\end{subfigure}

\vspace{0.5em}

\begin{subfigure}{\textwidth}
    \centering
    \includegraphics[width=0.8\textwidth]{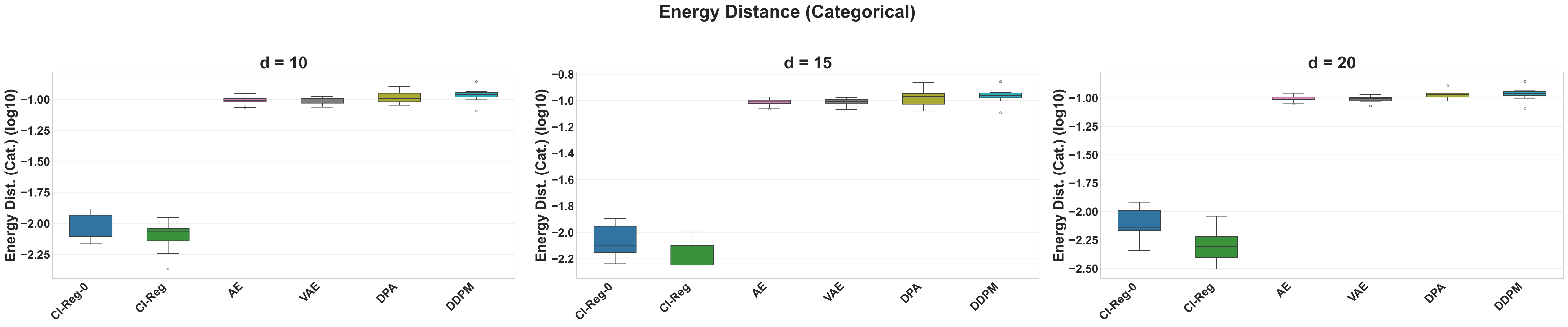}
    \caption{Categorical energy distance. Lower is better.}
\end{subfigure}

\vspace{0.5em}

\begin{subfigure}{\textwidth}
    \centering
    \includegraphics[width=0.8\textwidth]{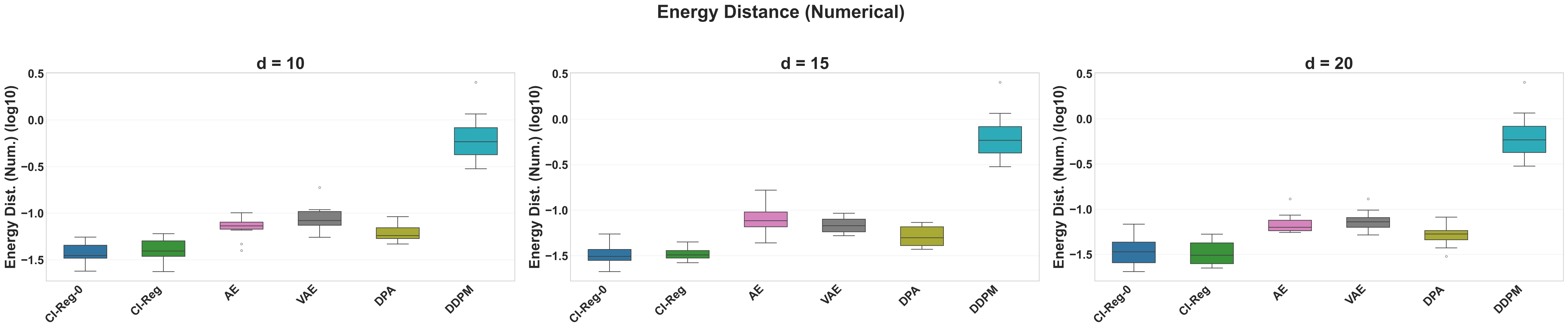}
    \caption{Numerical energy distance. Lower is better.}
\end{subfigure}

\vspace{0.5em}

\begin{subfigure}{\textwidth}
    \centering
    \includegraphics[width=0.8\textwidth]{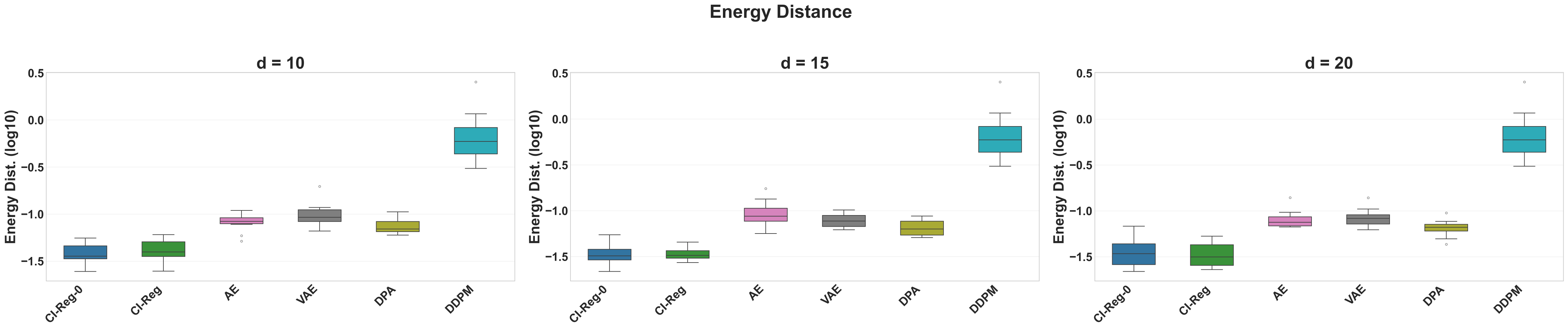}
    \caption{Overall energy distance. Lower is better.}
\end{subfigure}

\vspace{0.5em}

\begin{subfigure}{\textwidth}
    \centering
    \includegraphics[width=0.8\textwidth]{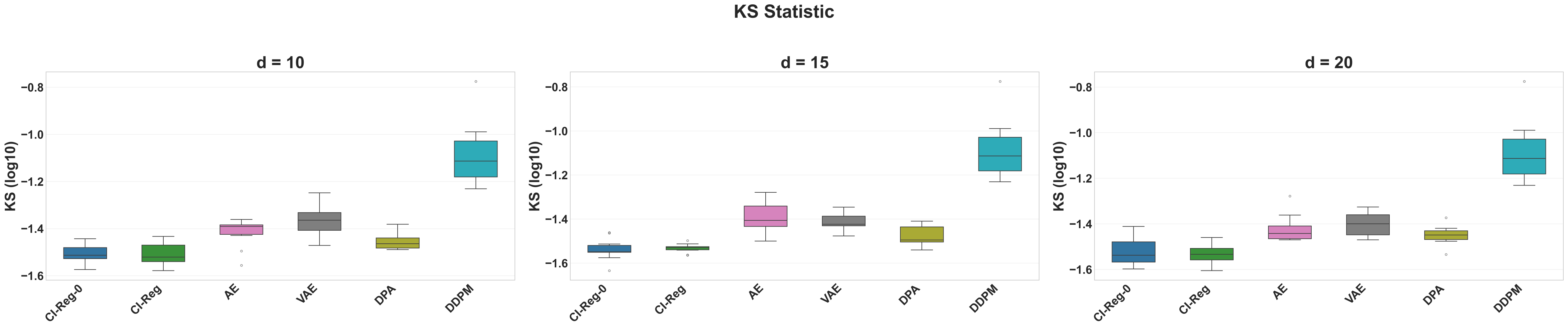}
    \caption{Numerical KS statistic. Lower is better.}
\end{subfigure}

\caption{
Comparison of CI-Reg, CI-Reg-0, AE, VAE, DPA, and DDPM on the structured-factor simulation study. Here CI-Reg denotes CI-Reg-0.03. Results are based on 10 independent repetitions.
}
\label{fig:rebuttal_main_comparison_new_simulation}
\end{figure}

Overall, the results in Figure \ref{fig:rebuttal_main_comparison_new_simulation} demonstrate that the proposed framework performs competitively across all evaluation metrics and encoder dimensions. Both CI-Reg and CI-Reg-0 substantially outperform AE, VAE, DPA, and DDPM on categorical distribution recovery, as measured by the categorical TV distance and categorical energy distance. At the same time, the proposed methods remain competitive on numerical reconstruction, achieving favorable numerical energy distances and KS statistics across all considered representation dimensions.

For preservation of mixed-type dependence structure, CI-Reg generally improves upon CI-Reg-0 in terms of the combined association relative Frobenius error, indicating the contribution of the conditional independence regularizer beyond the split-decoder architecture. Importantly, these qualitative conclusions remain consistent under under-specified ($d=10$), correctly specified ($d=15$), and over-specified ($d=20$) representation dimensions, suggesting that the proposed framework is robust to latent dimension misspecification.

\paragraph{Sensitivity analysis.}

To better understand the effect of the conditional independence regularizer, we further examine performance across a range of regularization strengths $\lambda_{\mathrm{CInd}} \in \{ 0, 0.01, 0.03, 0.1, 0.3, 0.9 \}$. Recall that $\lambda_{\mathrm{CInd}}=0$ corresponds exactly to CI-Reg-0.

\begin{figure}[htbp]
\centering

\begin{subfigure}{0.48\textwidth}
\centering
\includegraphics[width=\textwidth]{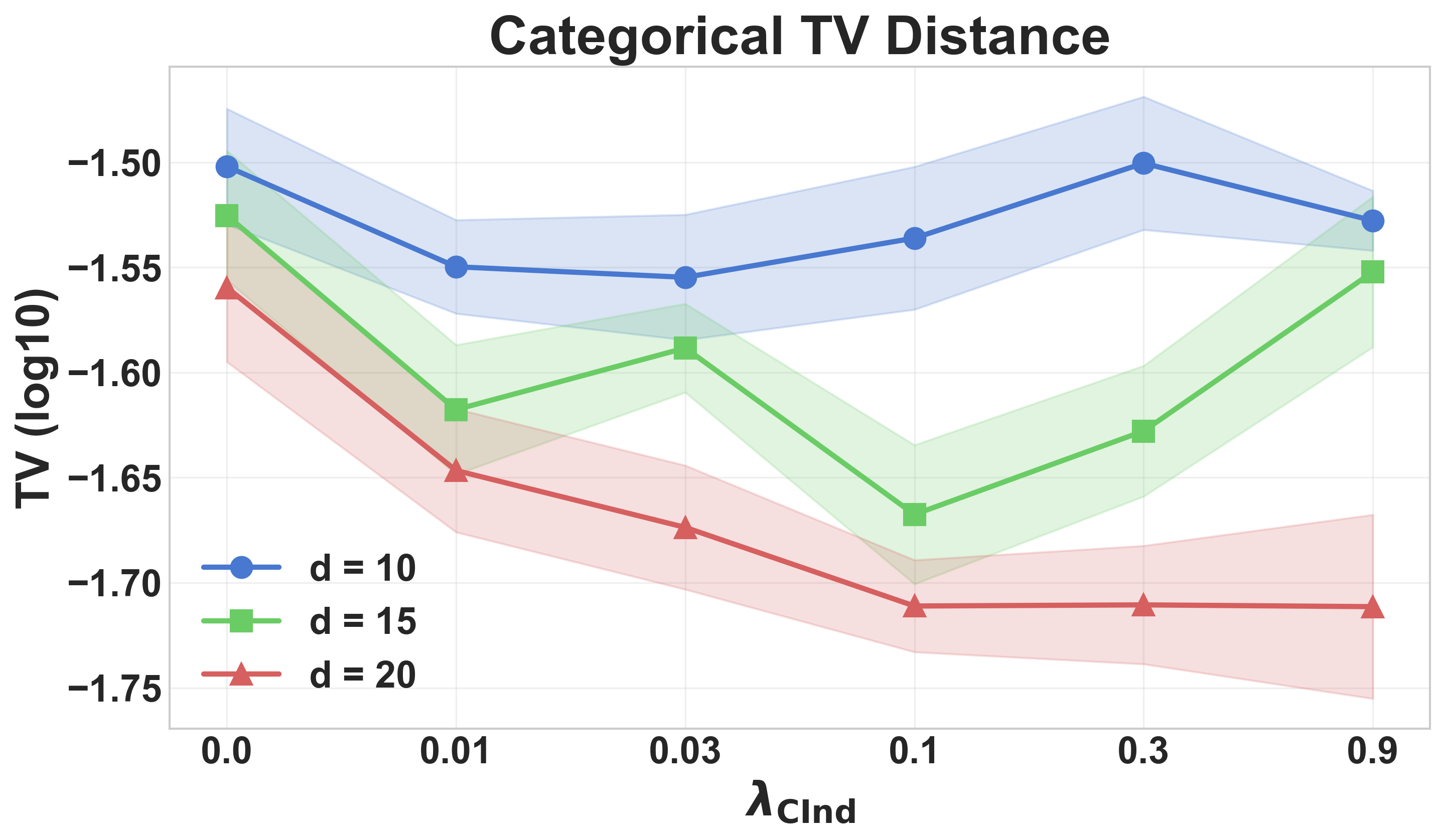}
\caption{Categorical TV distance.}
\end{subfigure}
\hfill
\begin{subfigure}{0.48\textwidth}
\centering
\includegraphics[width=\textwidth]{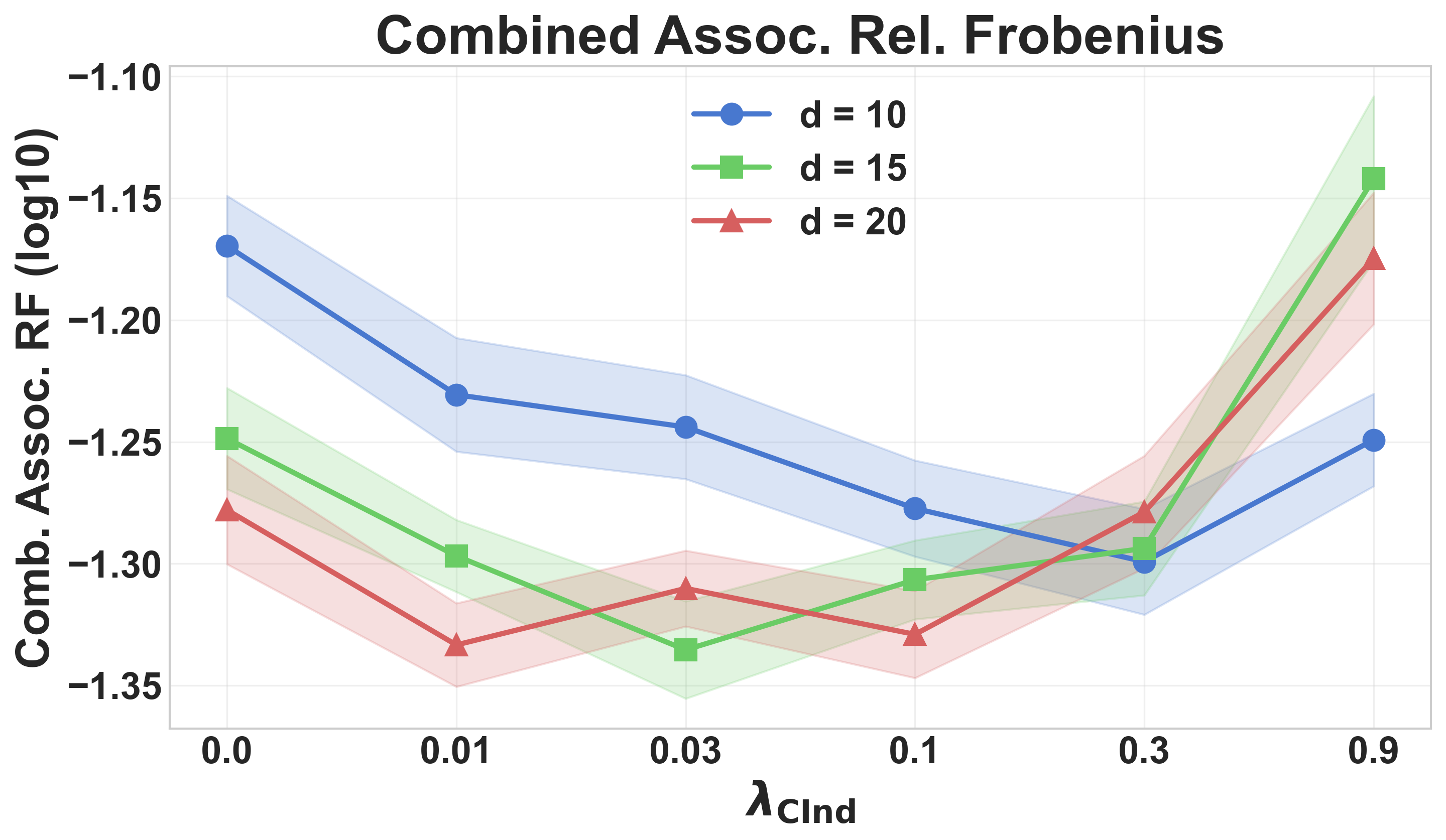}
\caption{Combined association relative Frobenius error.}
\end{subfigure}

\vspace{0.5em}

\begin{subfigure}{0.48\textwidth}
\centering
\includegraphics[width=\textwidth]{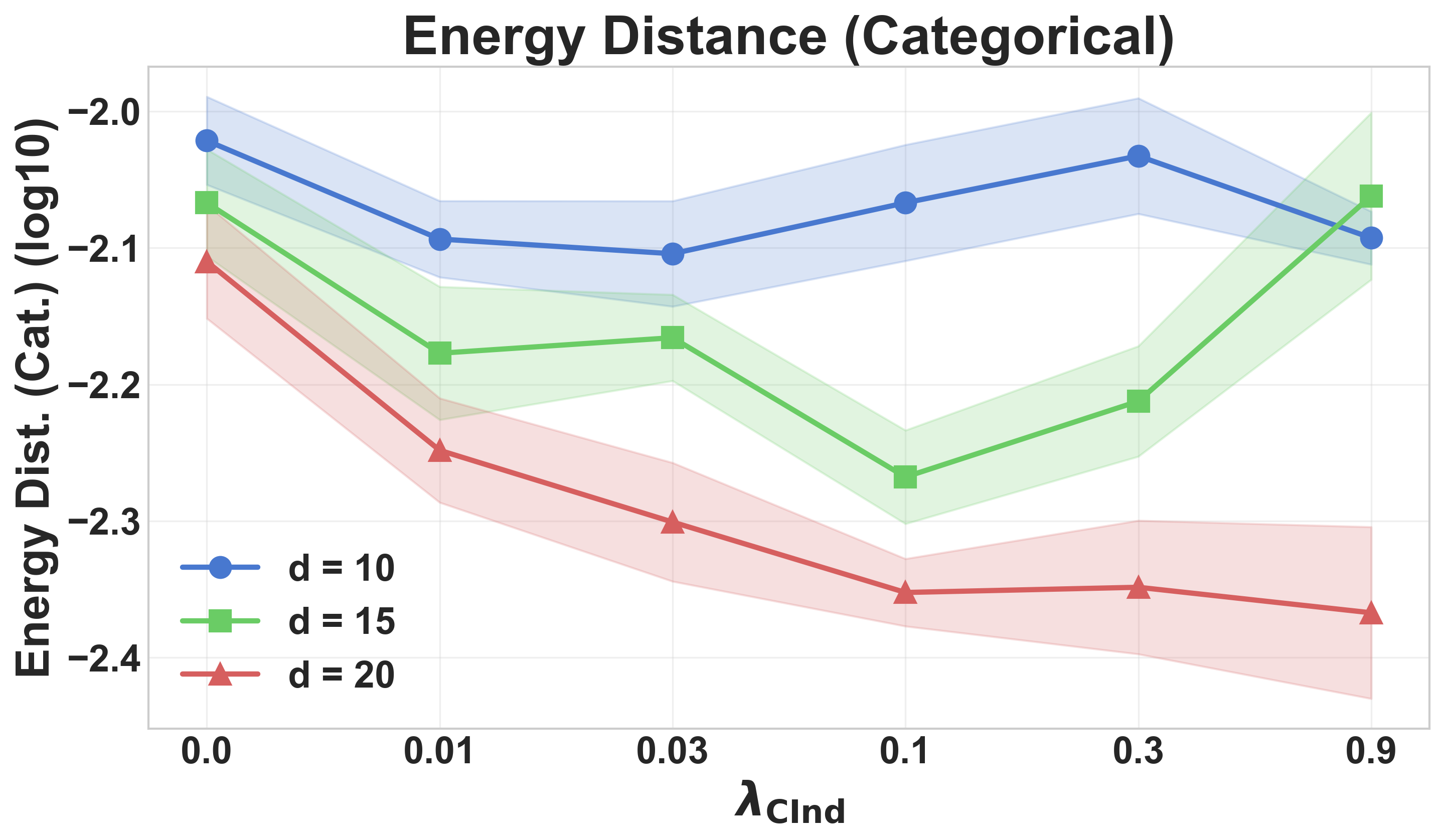}
\caption{Categorical energy distance.}
\end{subfigure}
\hfill
\begin{subfigure}{0.48\textwidth}
\centering
\includegraphics[width=\textwidth]{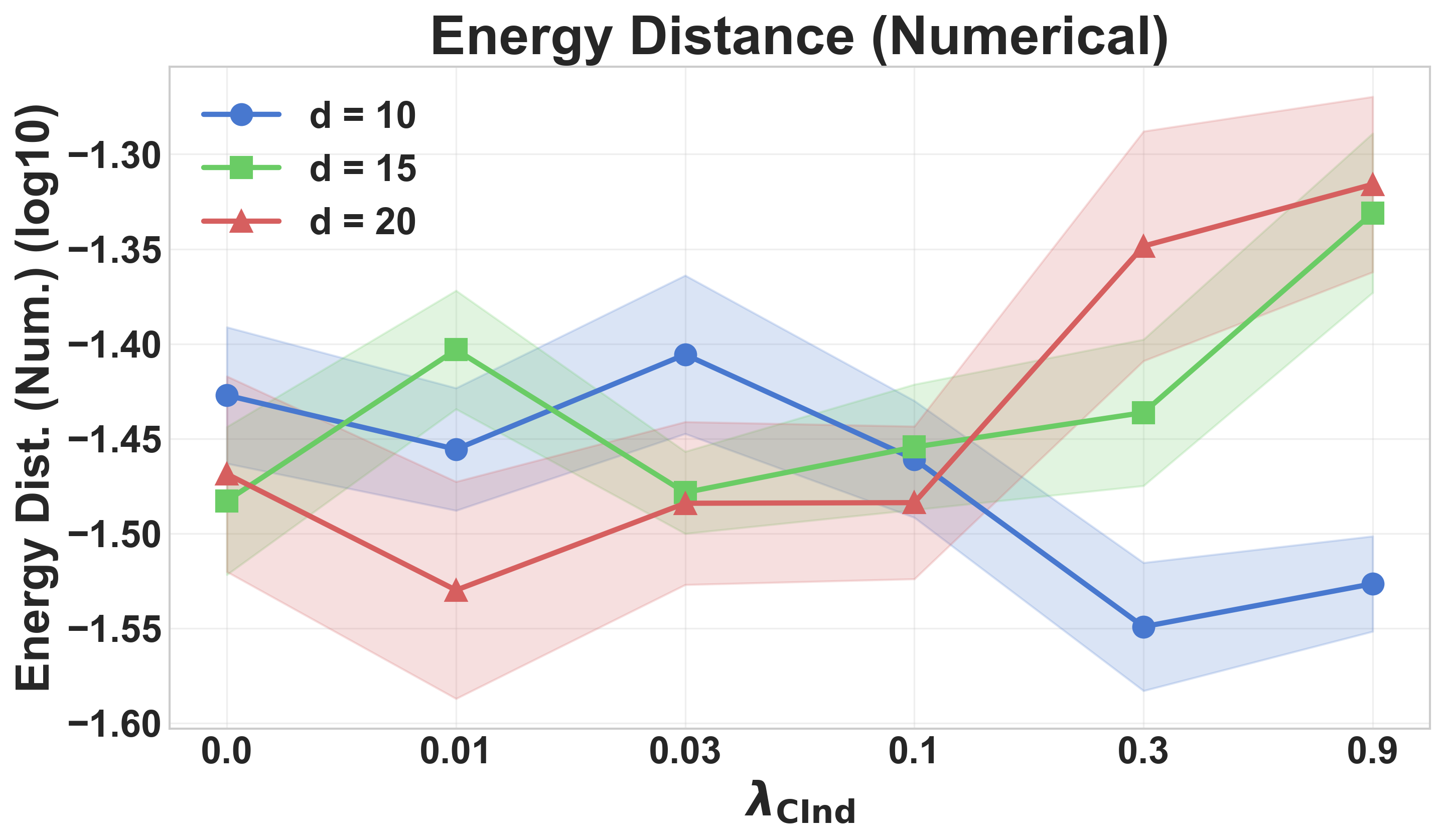}
\caption{Numerical energy distance.}
\end{subfigure}

\vspace{0.5em}

\begin{subfigure}{0.48\textwidth}
\centering
\includegraphics[width=\textwidth]{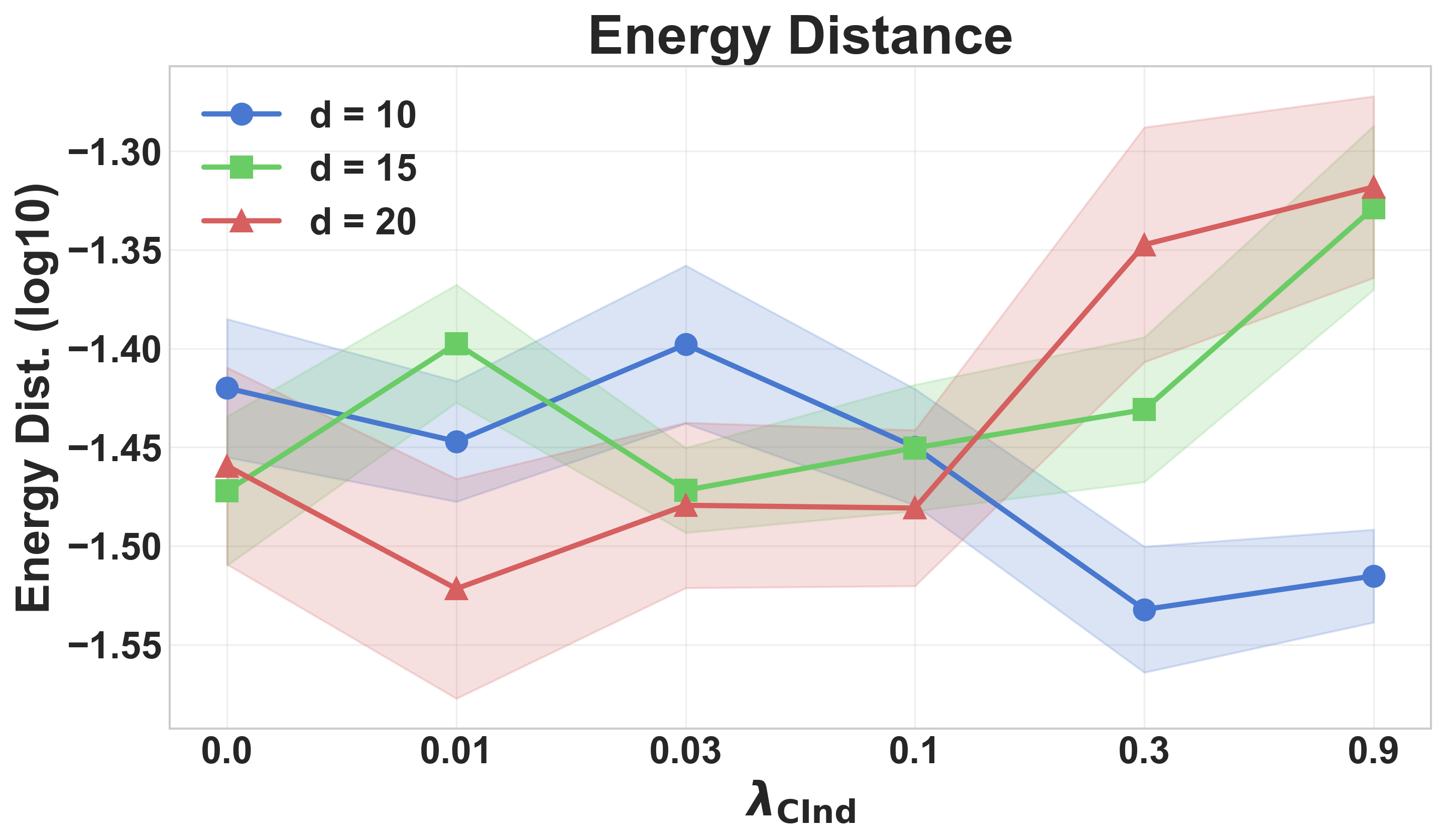}
\caption{Overall energy distance.}
\end{subfigure}
\hfill
\begin{subfigure}{0.48\textwidth}
\centering
\includegraphics[width=\textwidth]{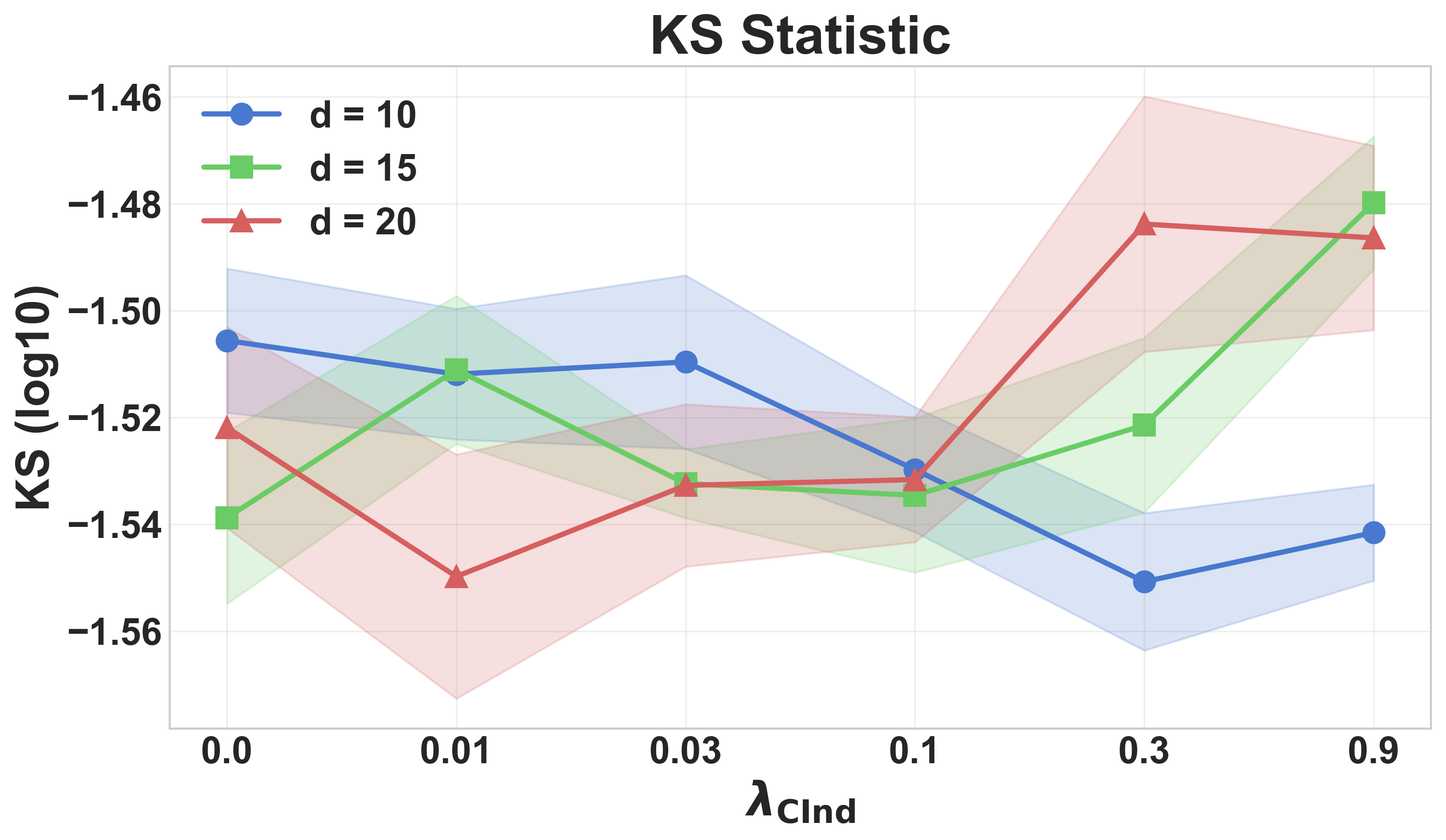}
\caption{Numerical KS statistic.}
\end{subfigure}

\caption{
Sensitivity analysis with respect to $\lambda_{\mathrm{CInd}}$ on the structured-factor simulation study. Lower values indicate better performance. Curves show the mean over 10 independent repetitions, with shaded regions indicating $\pm$ one standard deviation.
}
\label{fig:rebuttal_lambda_sensitivity_all_metrics}
\end{figure}

Figure~\ref{fig:rebuttal_lambda_sensitivity_all_metrics} further examines the effect of $\lambda_{\mathrm{CInd}}$ across all evaluation metrics. Overall, the results indicate that the proposed framework is relatively robust to the choice of $\lambda_{\mathrm{CInd}}$. Across all three encoder dimensions, performance varies smoothly as the regularization strength changes, and a broad range of moderate values ($0.01$--$0.1$, and in some cases $0.3$) yields similar performance.

Several consistent patterns emerge. First, introducing a moderate amount of conditional independence regularization improves categorical distribution recovery, as evidenced by both categorical TV distance and categorical energy distance. Second, moderate values of $\lambda_{\mathrm{CInd}}$ also improve preservation of the mixed-type dependence structure, as measured by the combined association relative Frobenius error. Third, numerical reconstruction metrics, including the numerical energy distance and KS statistic, remain relatively stable over a wide range of regularization strengths.

Overall, these results suggest that the conditional independence regularizer acts as a useful structural inductive bias when applied moderately, while the proposed method is not overly sensitive to the precise choice of $\lambda_{\mathrm{CInd}}$.

\begin{figure}[htbp]
\centering

\scriptsize
\begin{tabular}{cccccc}
CI-Reg-0 &
CI-Reg-0.01 &
CI-Reg-0.03 &
CI-Reg-0.10 &
CI-Reg-0.30 &
CI-Reg-0.90 \\[0.3em]

\includegraphics[width=0.14\textwidth]{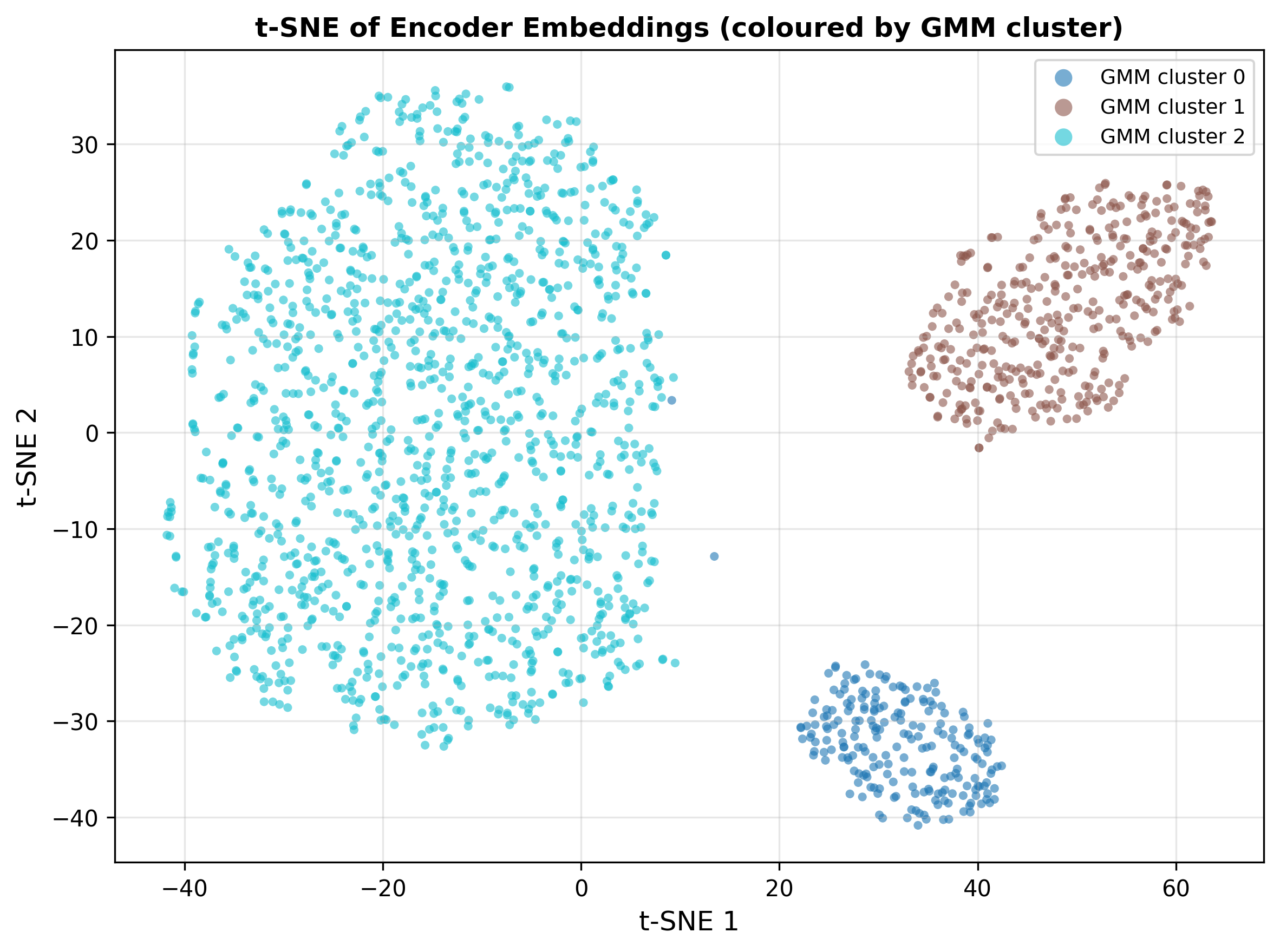} &
\includegraphics[width=0.14\textwidth]{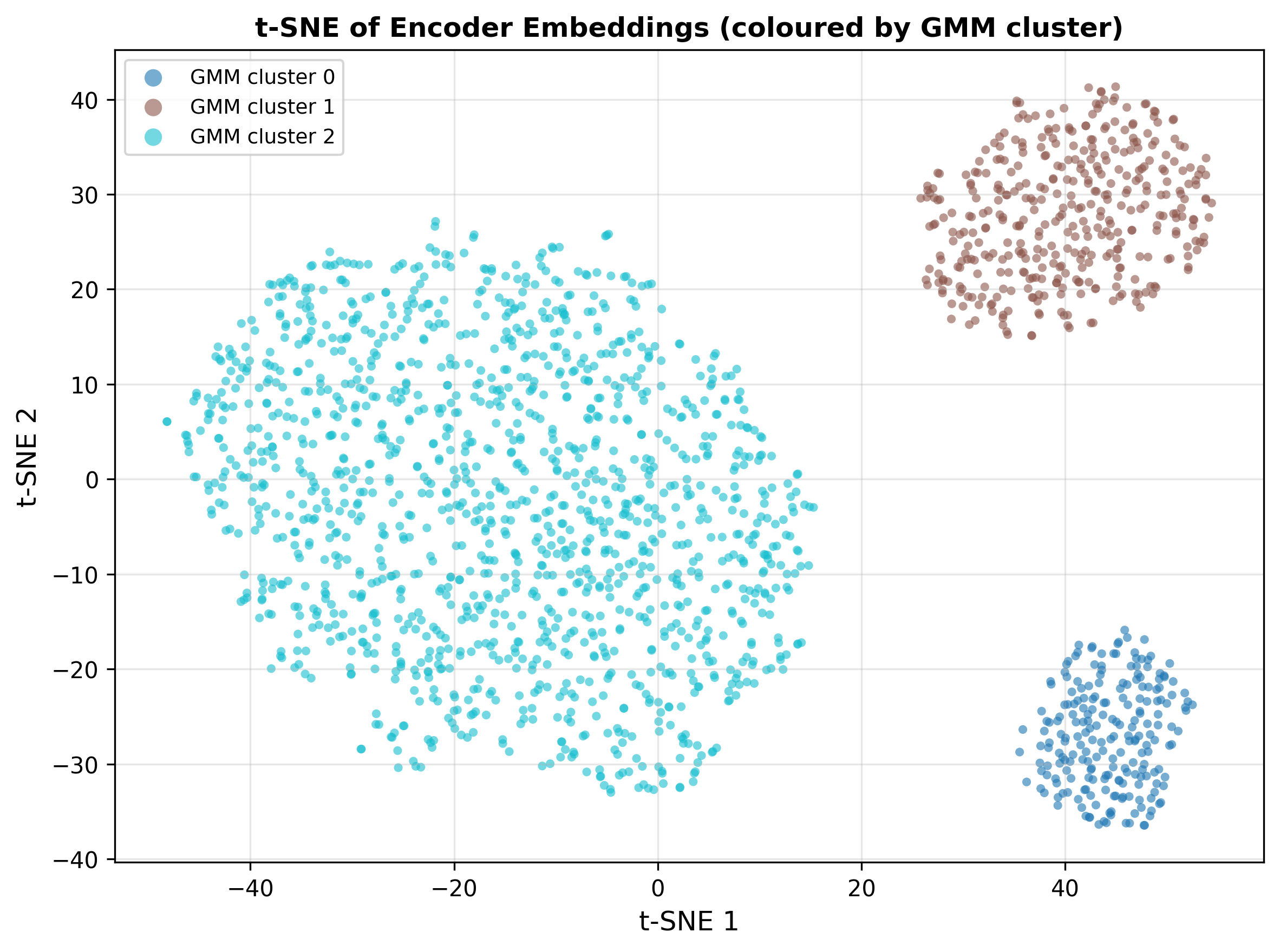} &
\includegraphics[width=0.14\textwidth]{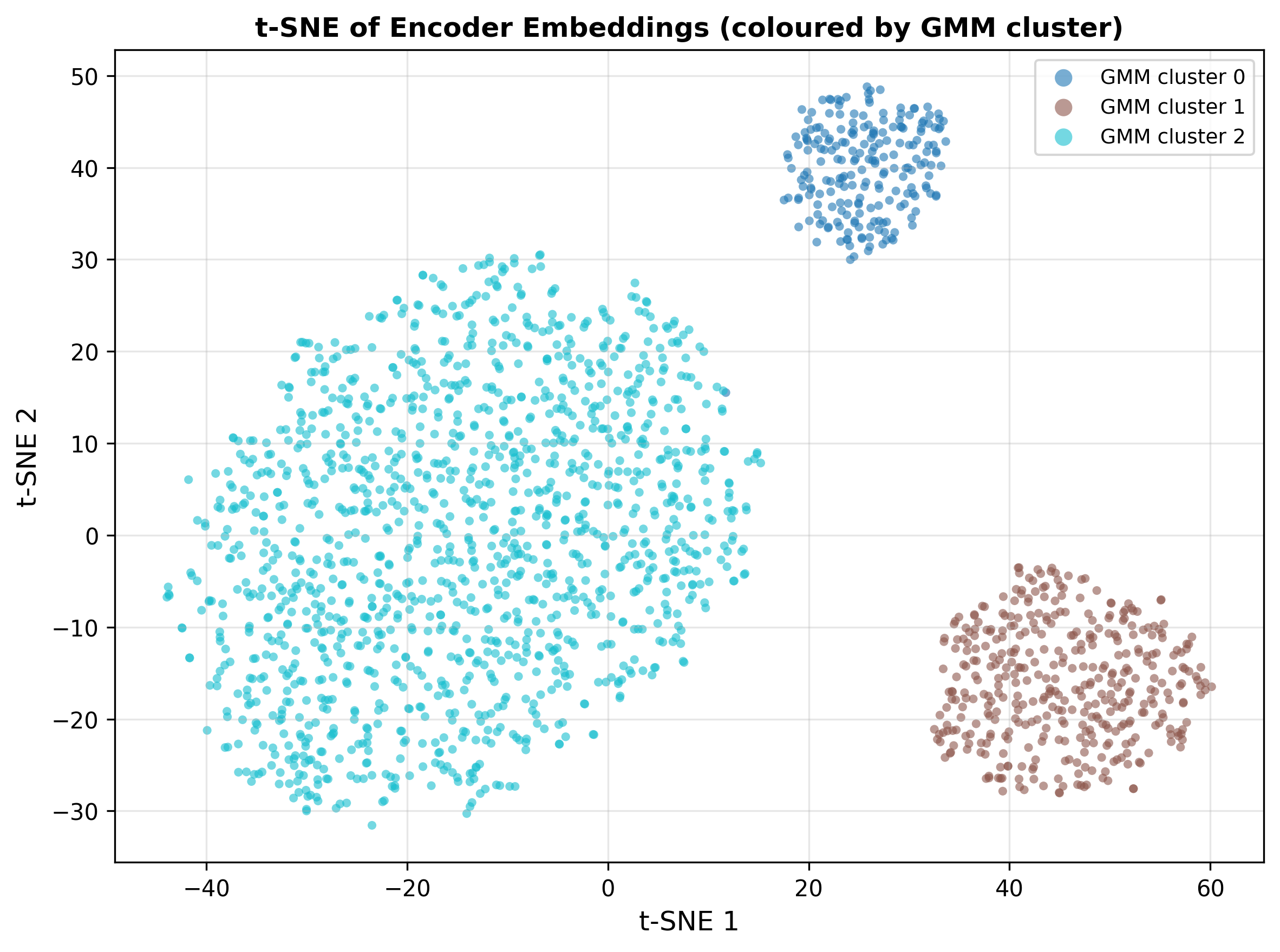} &
\includegraphics[width=0.14\textwidth]{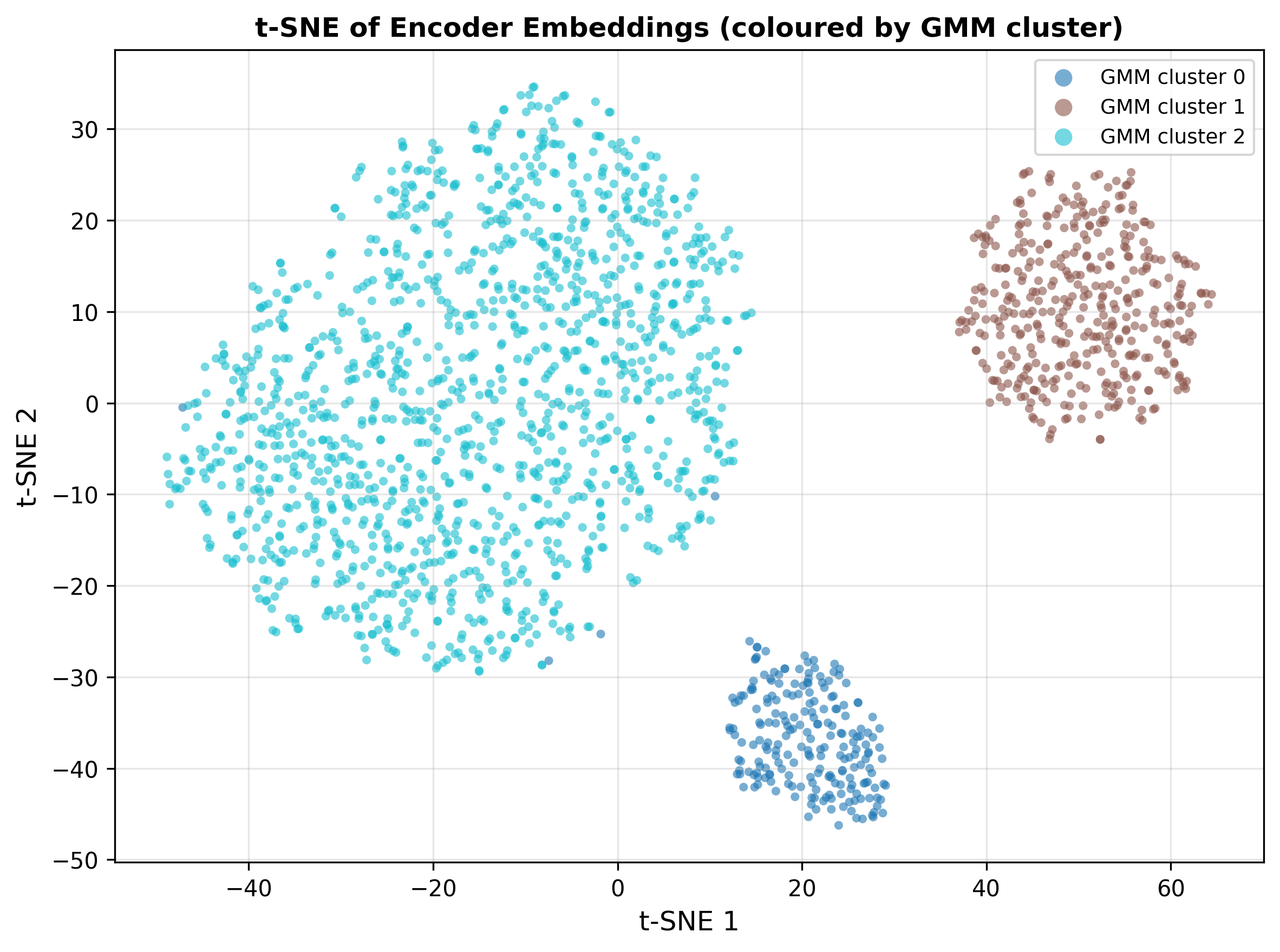} &
\includegraphics[width=0.14\textwidth]{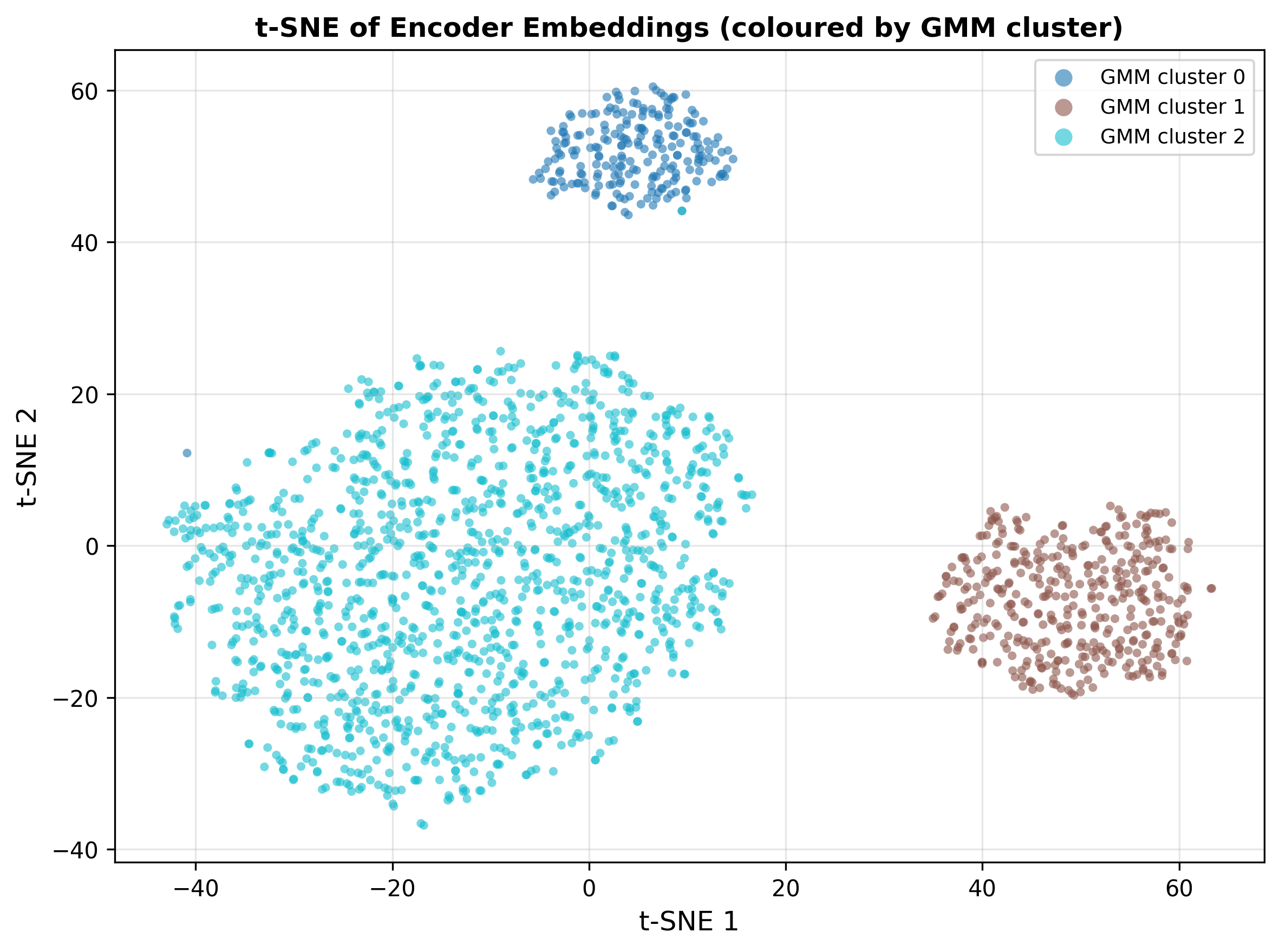} &
\includegraphics[width=0.14\textwidth]{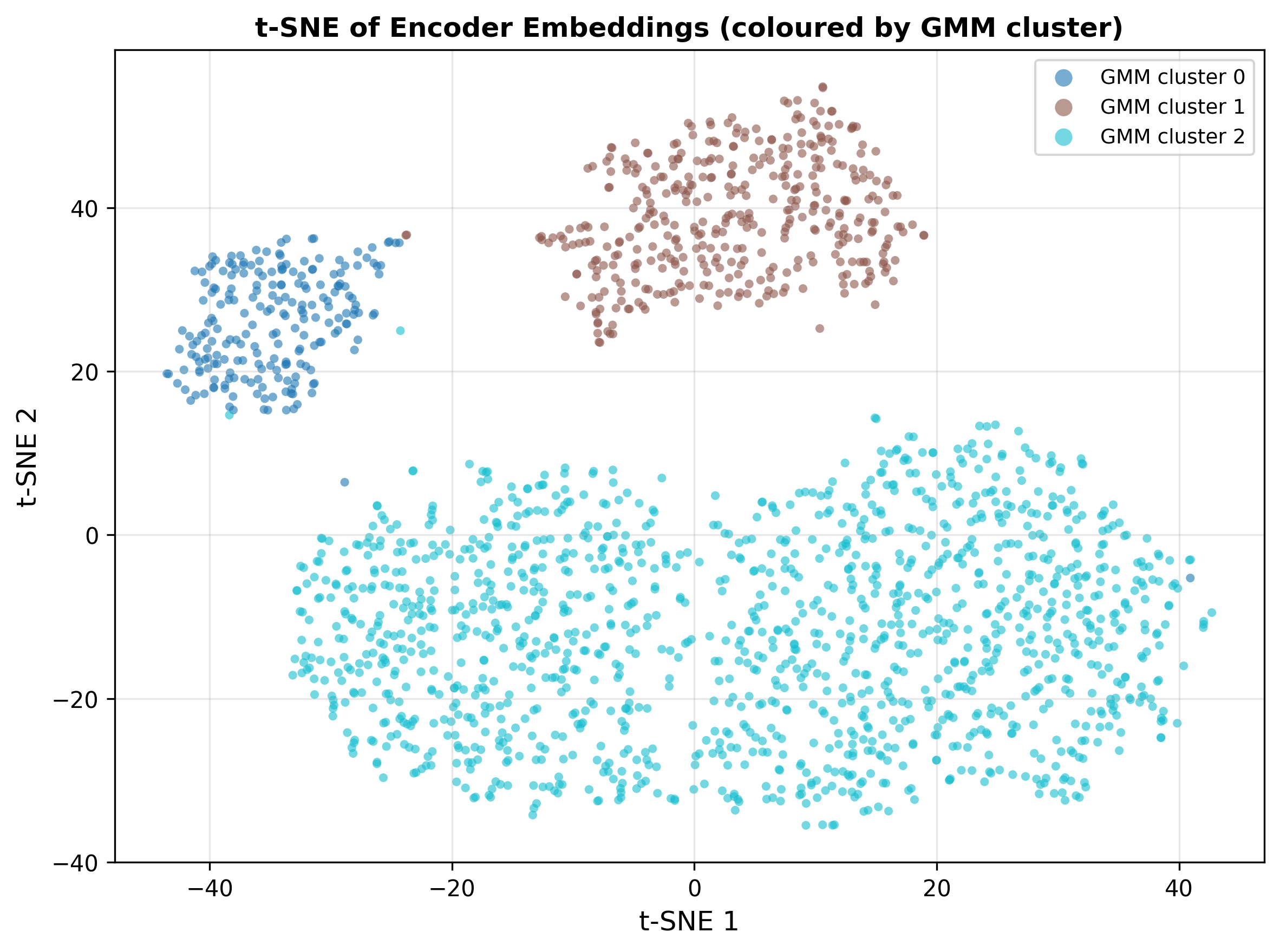}
\\
\multicolumn{6}{c}{$d=10$}
\\[0.5em]

\includegraphics[width=0.14\textwidth]{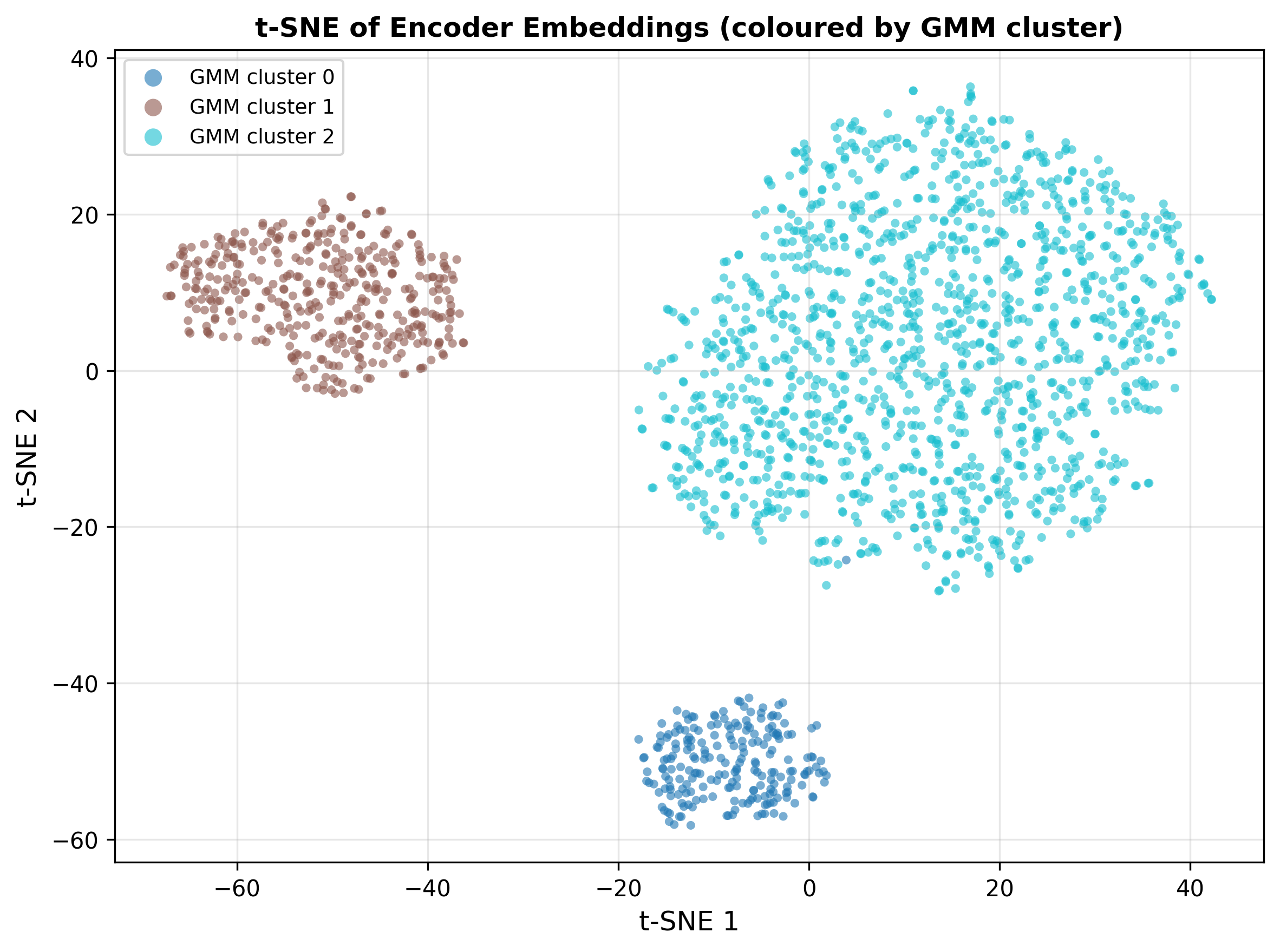} &
\includegraphics[width=0.14\textwidth]{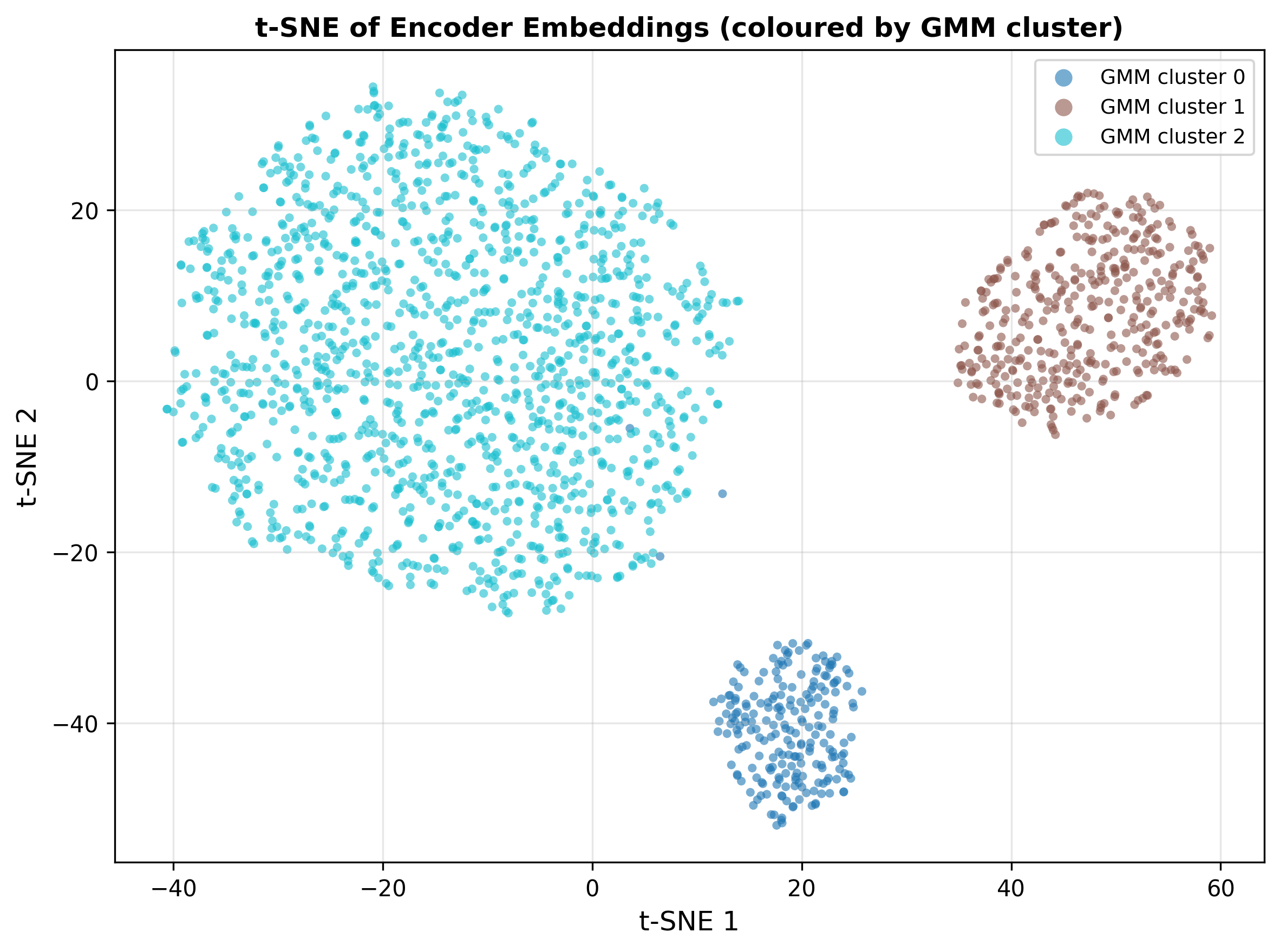} &
\includegraphics[width=0.14\textwidth]{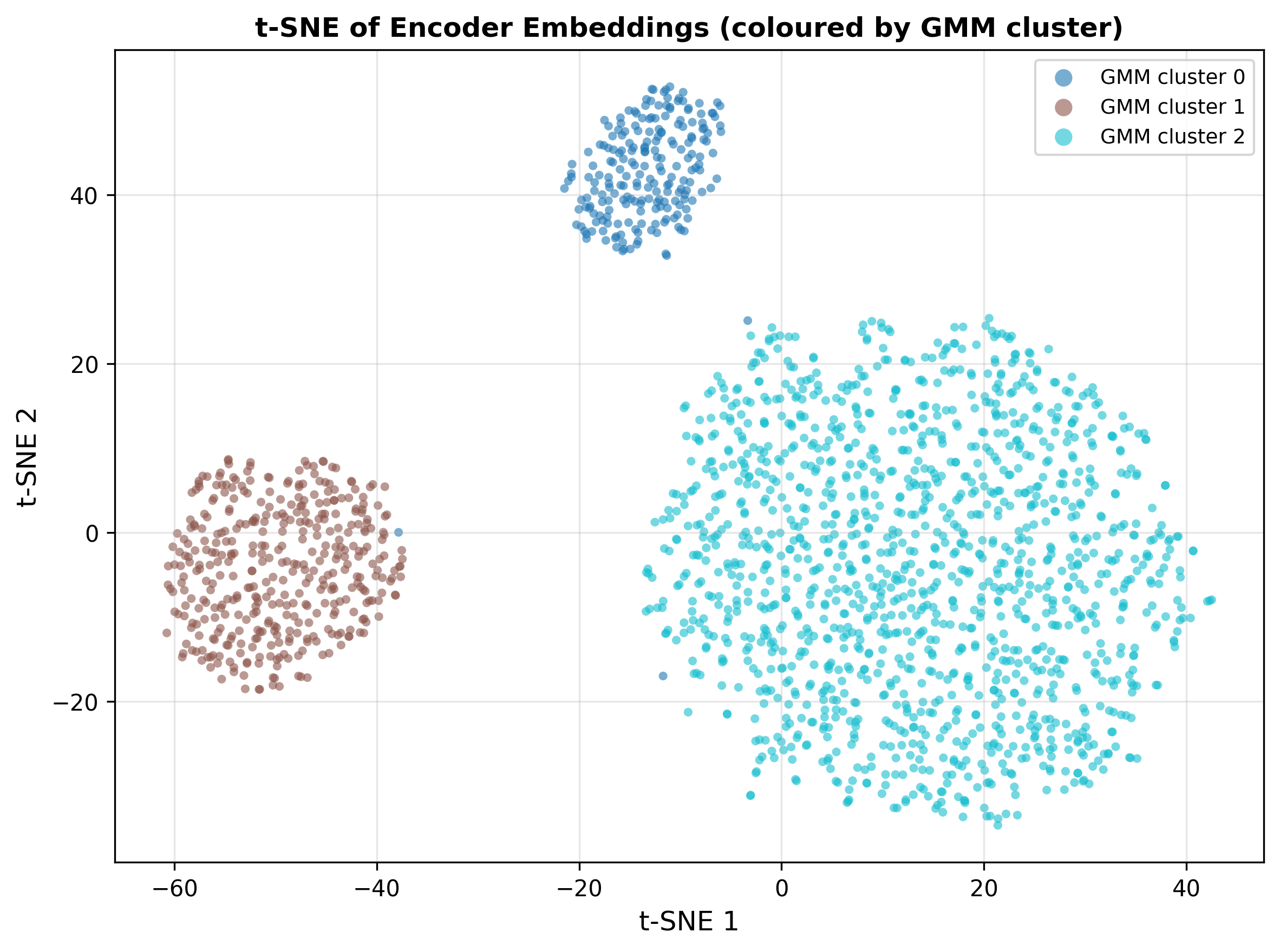} &
\includegraphics[width=0.14\textwidth]{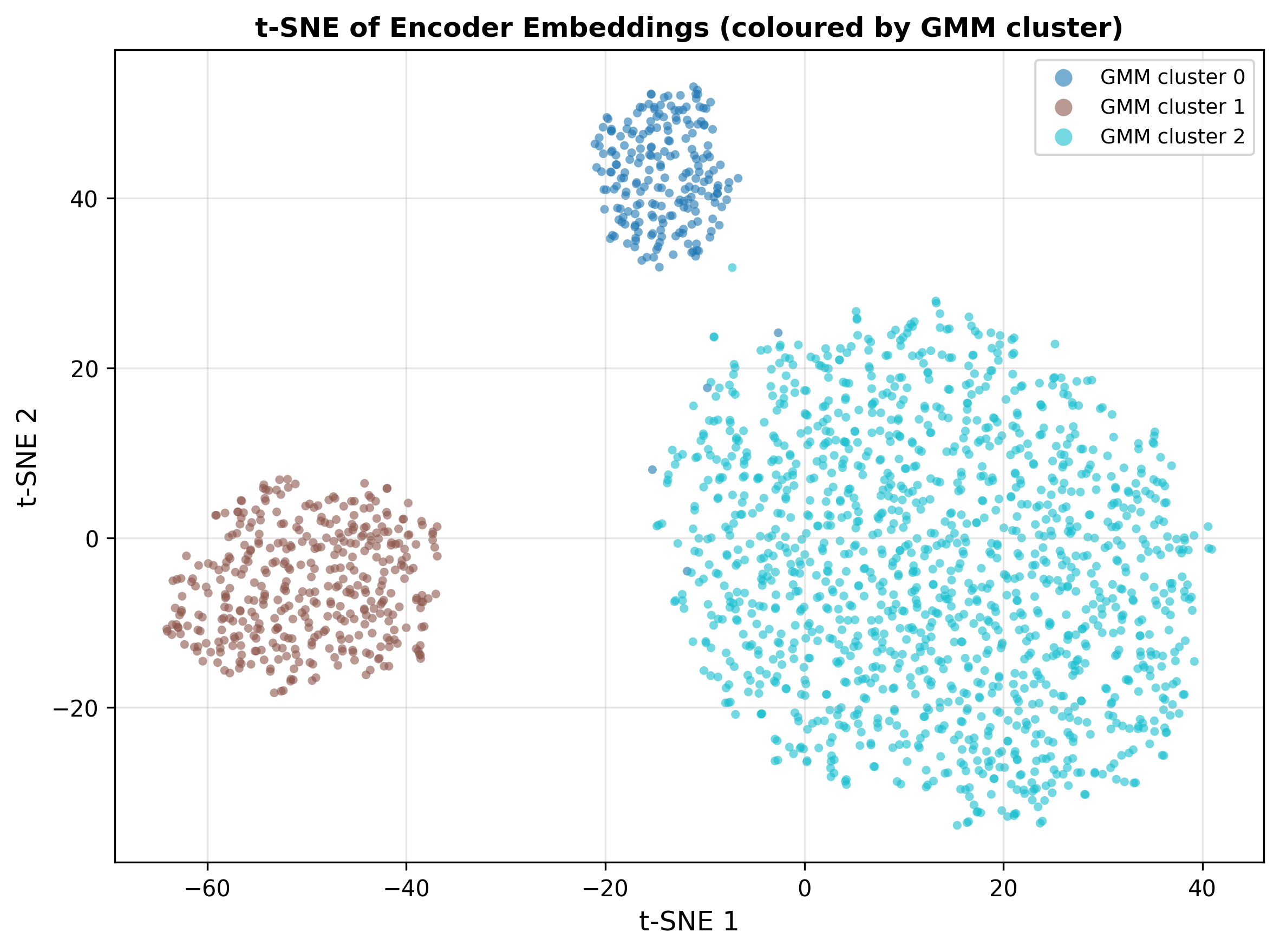} &
\includegraphics[width=0.14\textwidth]{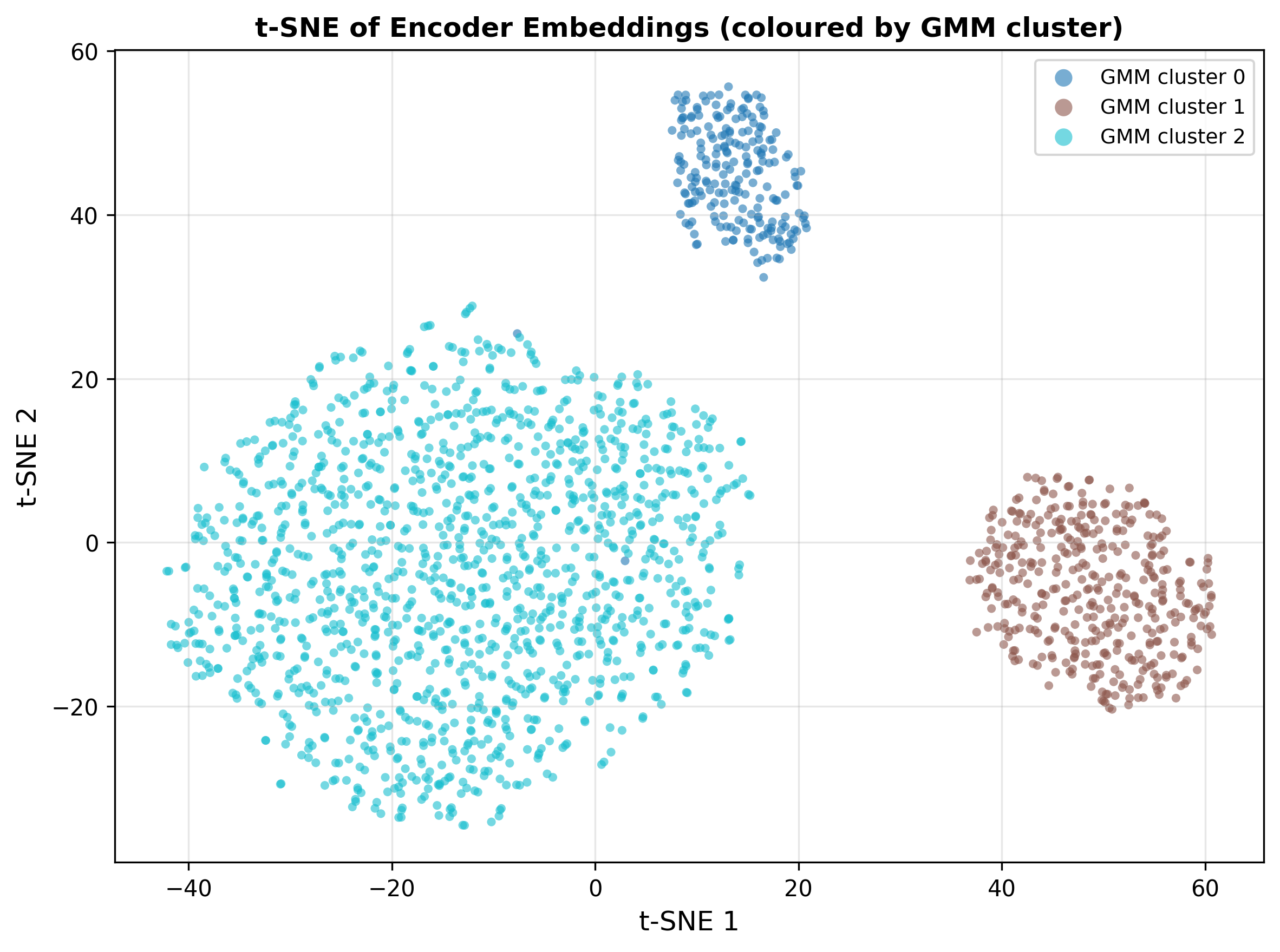} &
\includegraphics[width=0.14\textwidth]{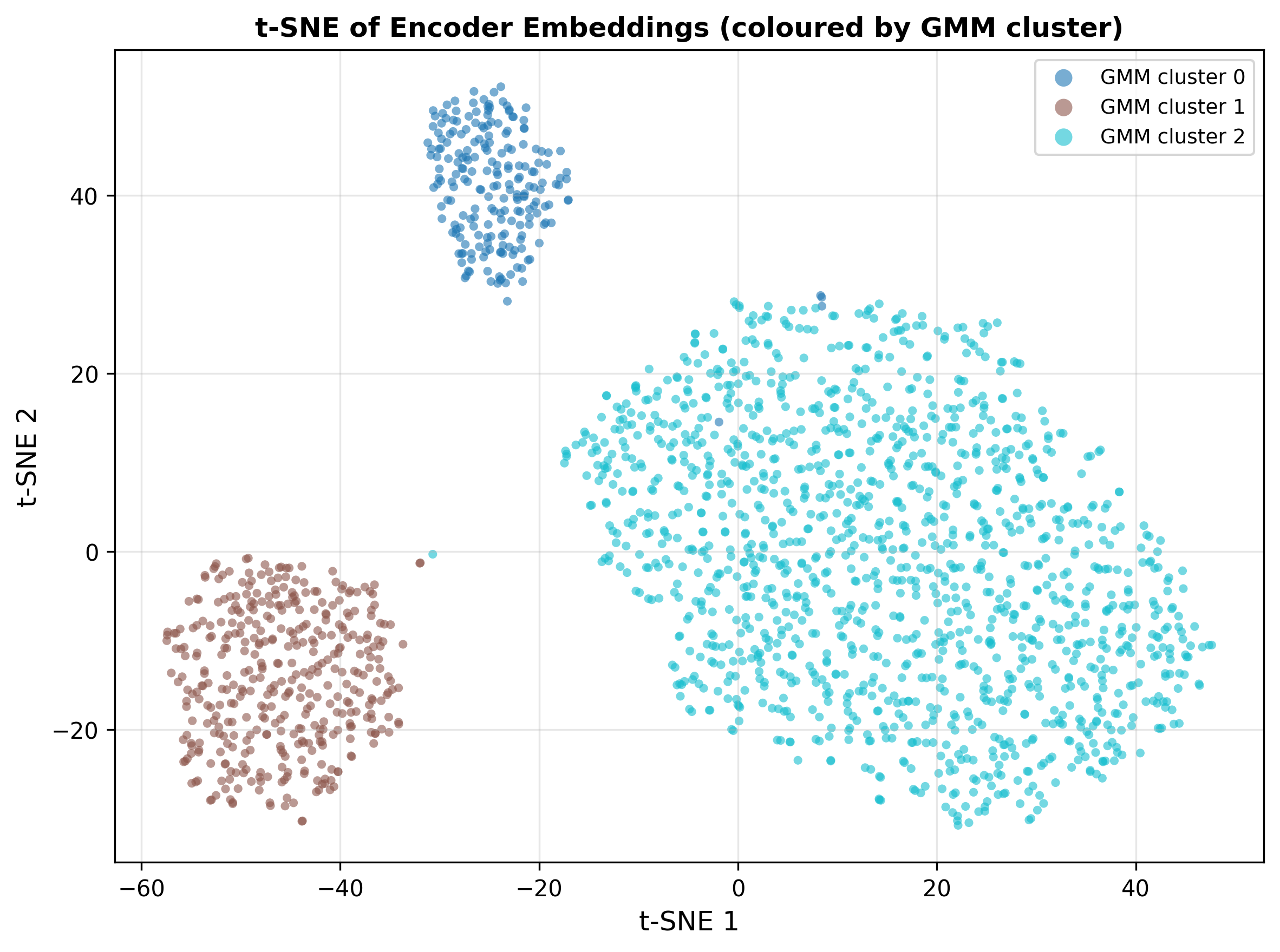}
\\
\multicolumn{6}{c}{$d=15$}
\\[0.5em]

\includegraphics[width=0.14\textwidth]{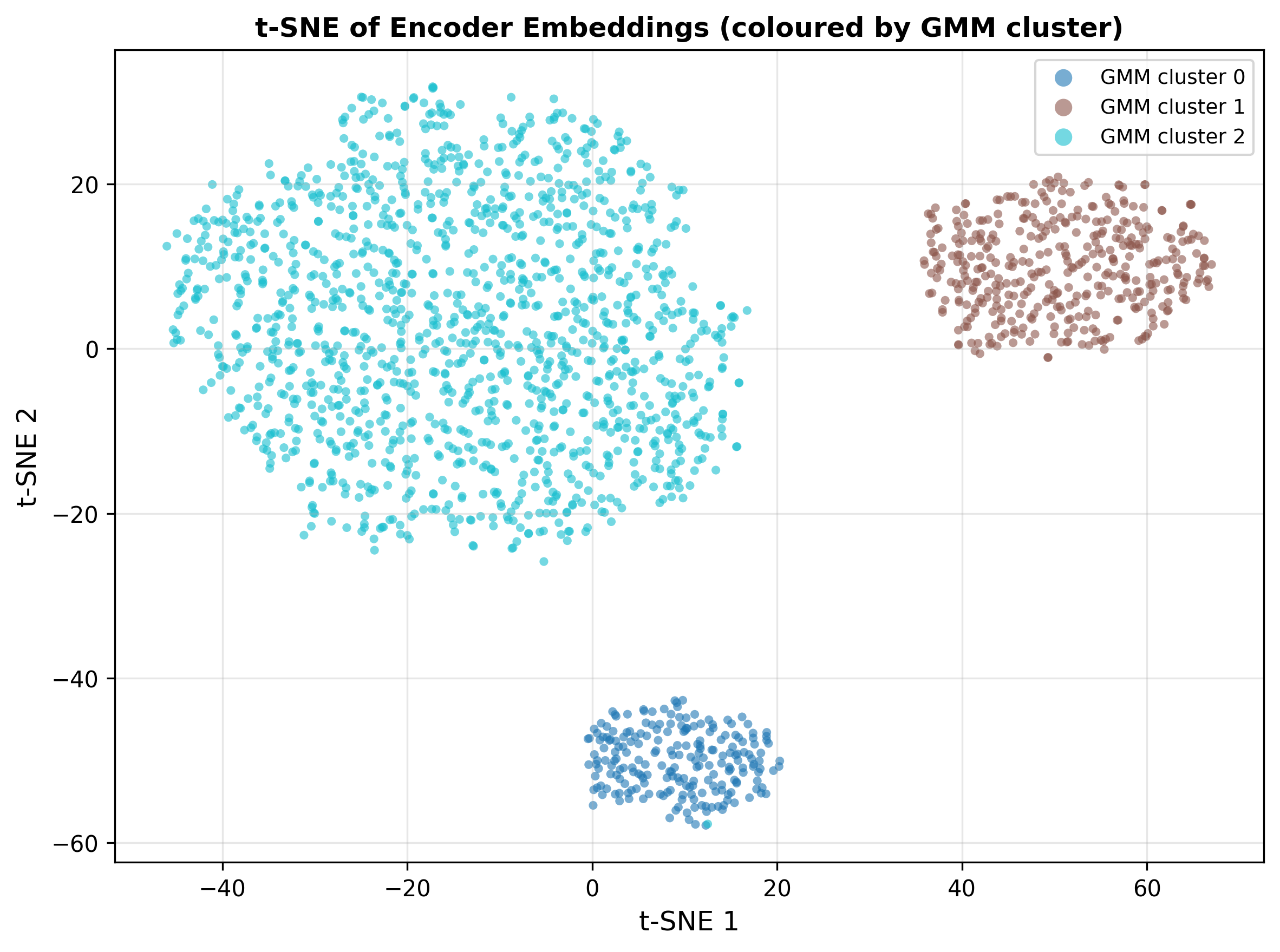} &
\includegraphics[width=0.14\textwidth]{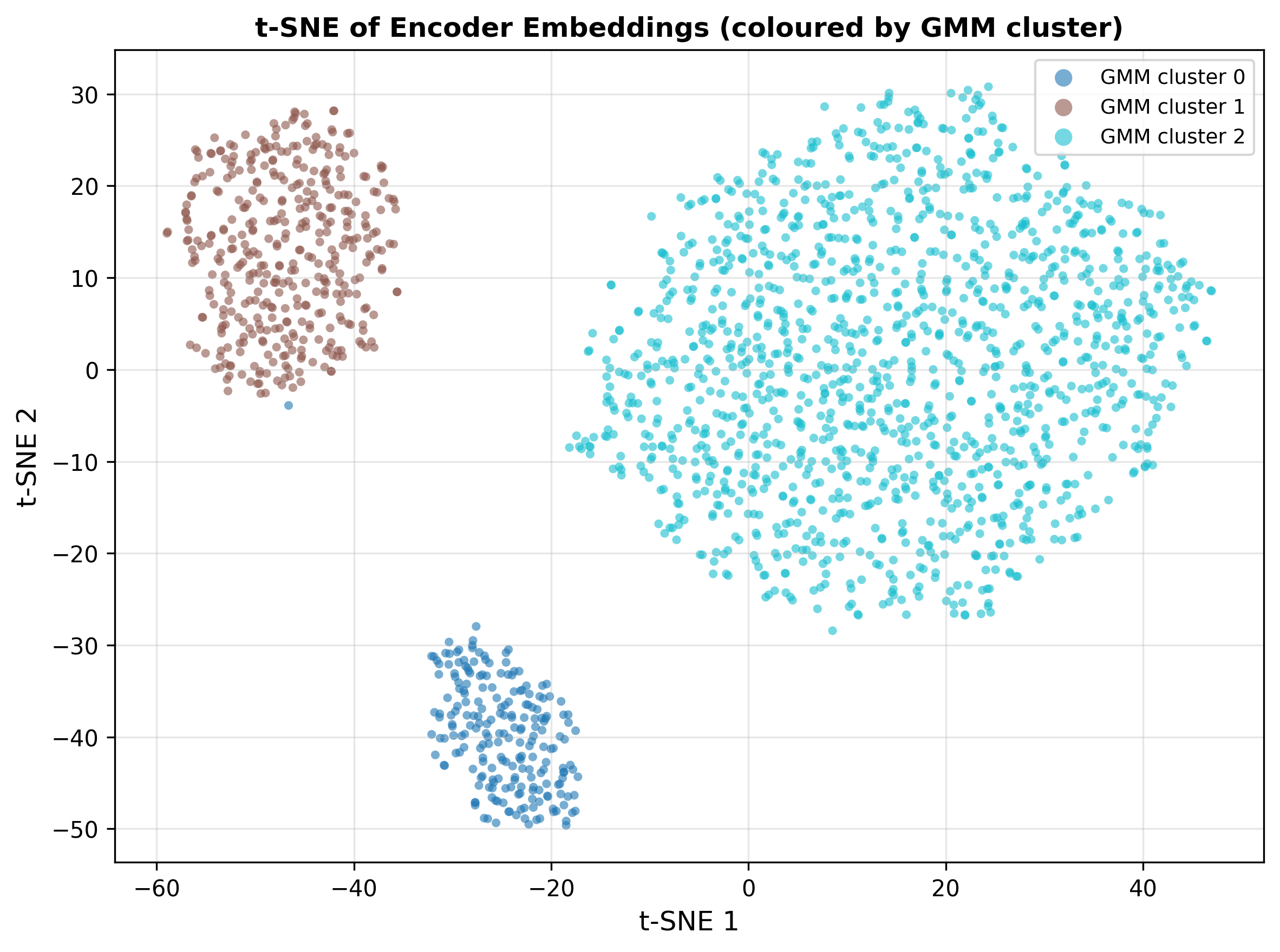} &
\includegraphics[width=0.14\textwidth]{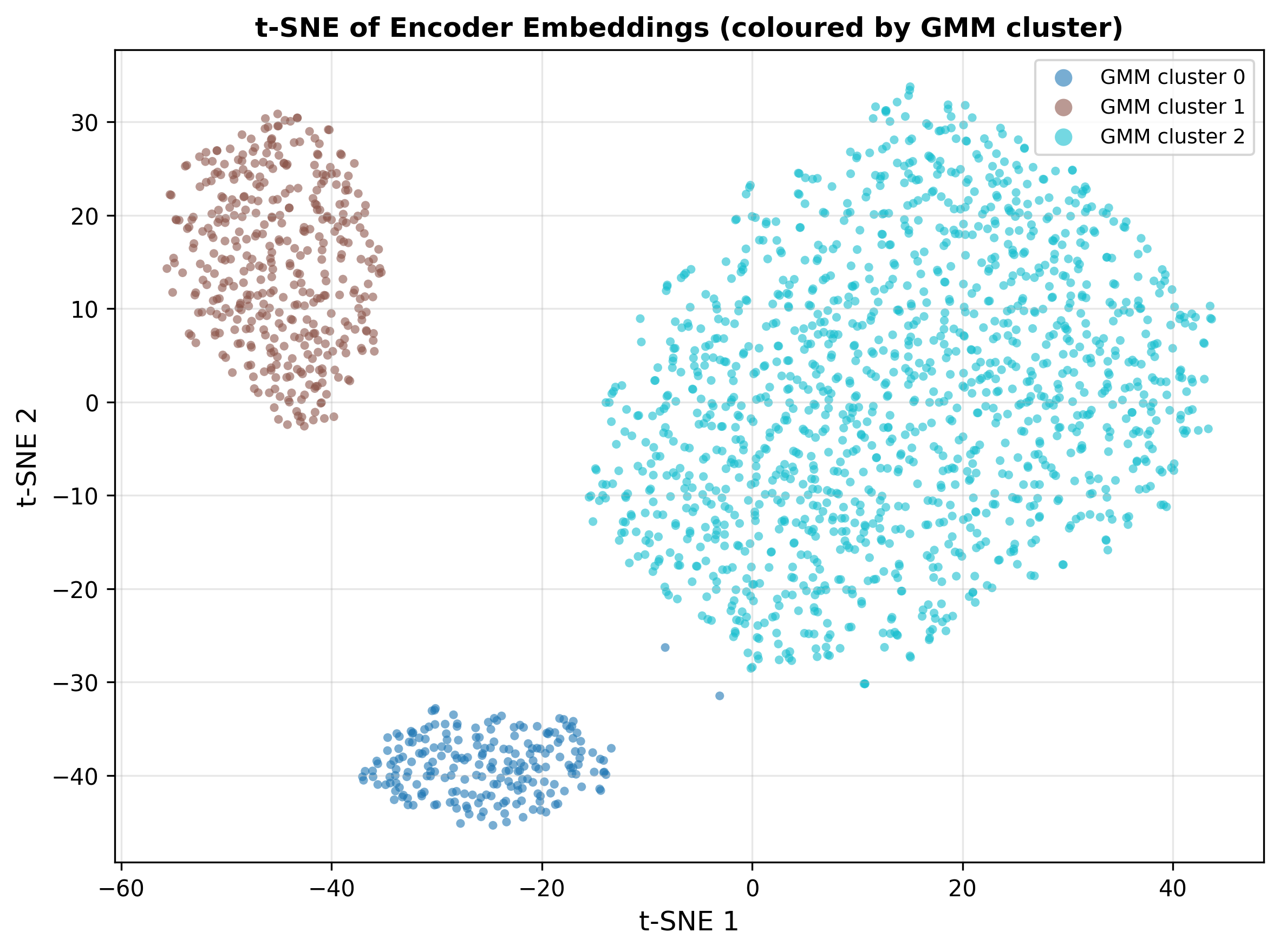} &
\includegraphics[width=0.14\textwidth]{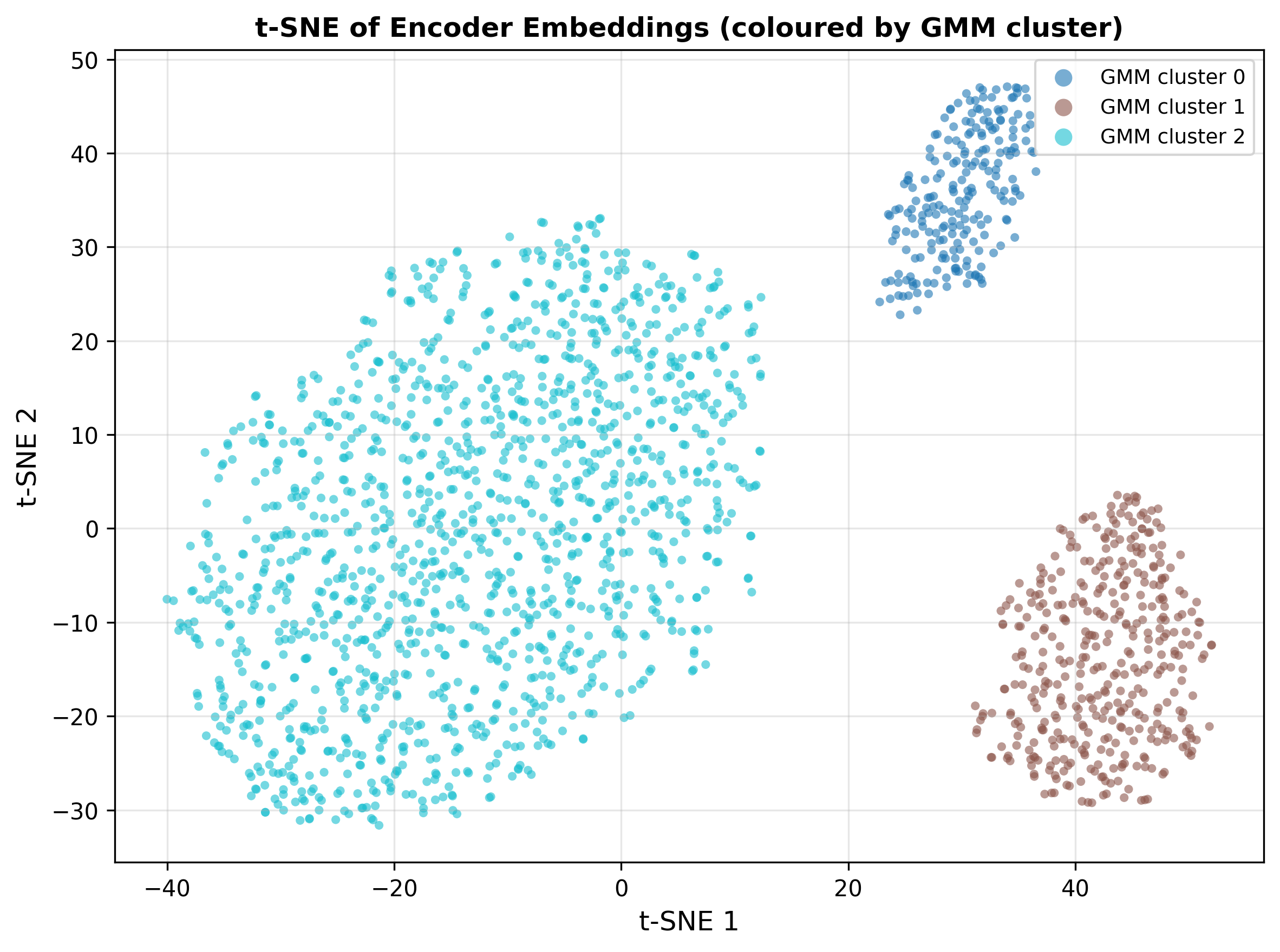} &
\includegraphics[width=0.14\textwidth]{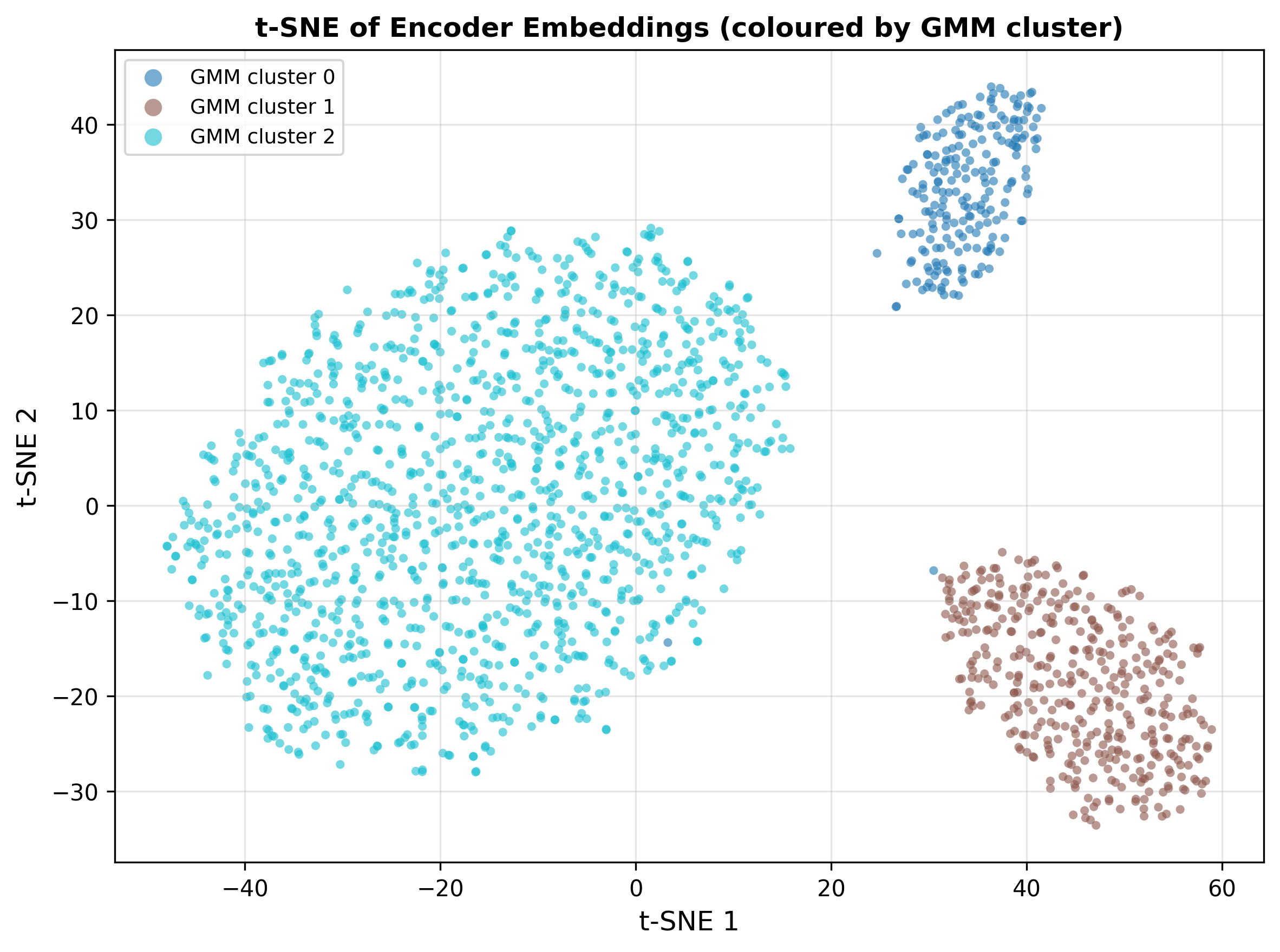} &
\includegraphics[width=0.14\textwidth]{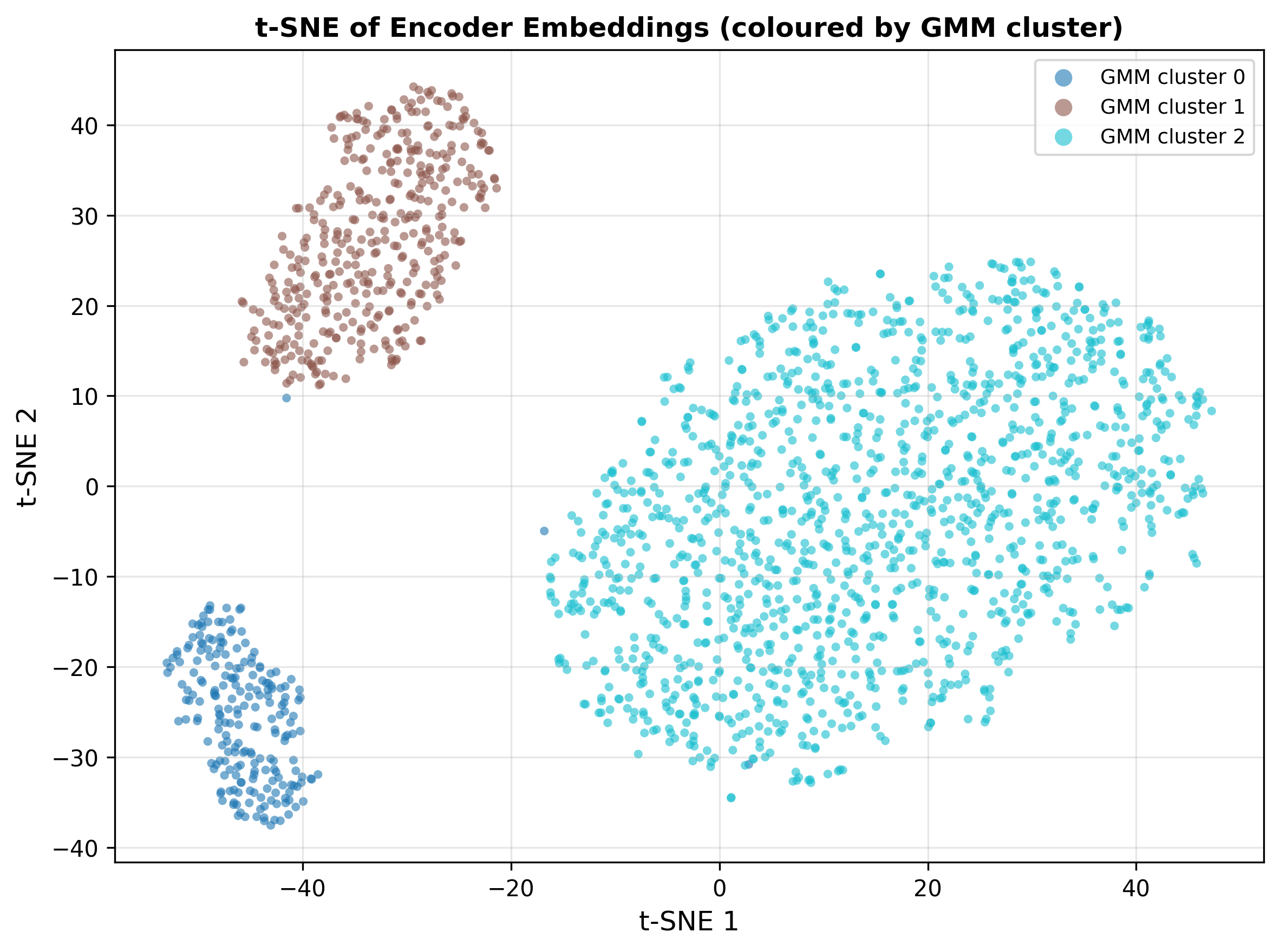}
\\
\multicolumn{6}{c}{$d=20$}
\end{tabular}

\caption{
t-SNE visualization of the learned representations for CI-Reg variants in the structured-factor simulation study with cluster separation 5.0. Rows correspond to encoder dimensions and columns correspond to different values of $\lambda_{\mathrm{CInd}}$. Points are colored according to the true latent Gaussian mixture component. Across all encoder dimensions and regularization strengths, the learned representations preserve the underlying latent cluster structure.
}
\label{fig:structured_factor_tsne_of_embeddings}
\end{figure}

\begin{figure}[htbp]
\centering

\scriptsize
\begin{tabular}{cccccc}
CI-Reg-0 &
CI-Reg-0.01 &
CI-Reg-0.03 &
CI-Reg-0.10 &
CI-Reg-0.30 &
CI-Reg-0.90 \\[0.3em]

\includegraphics[width=0.14\textwidth]{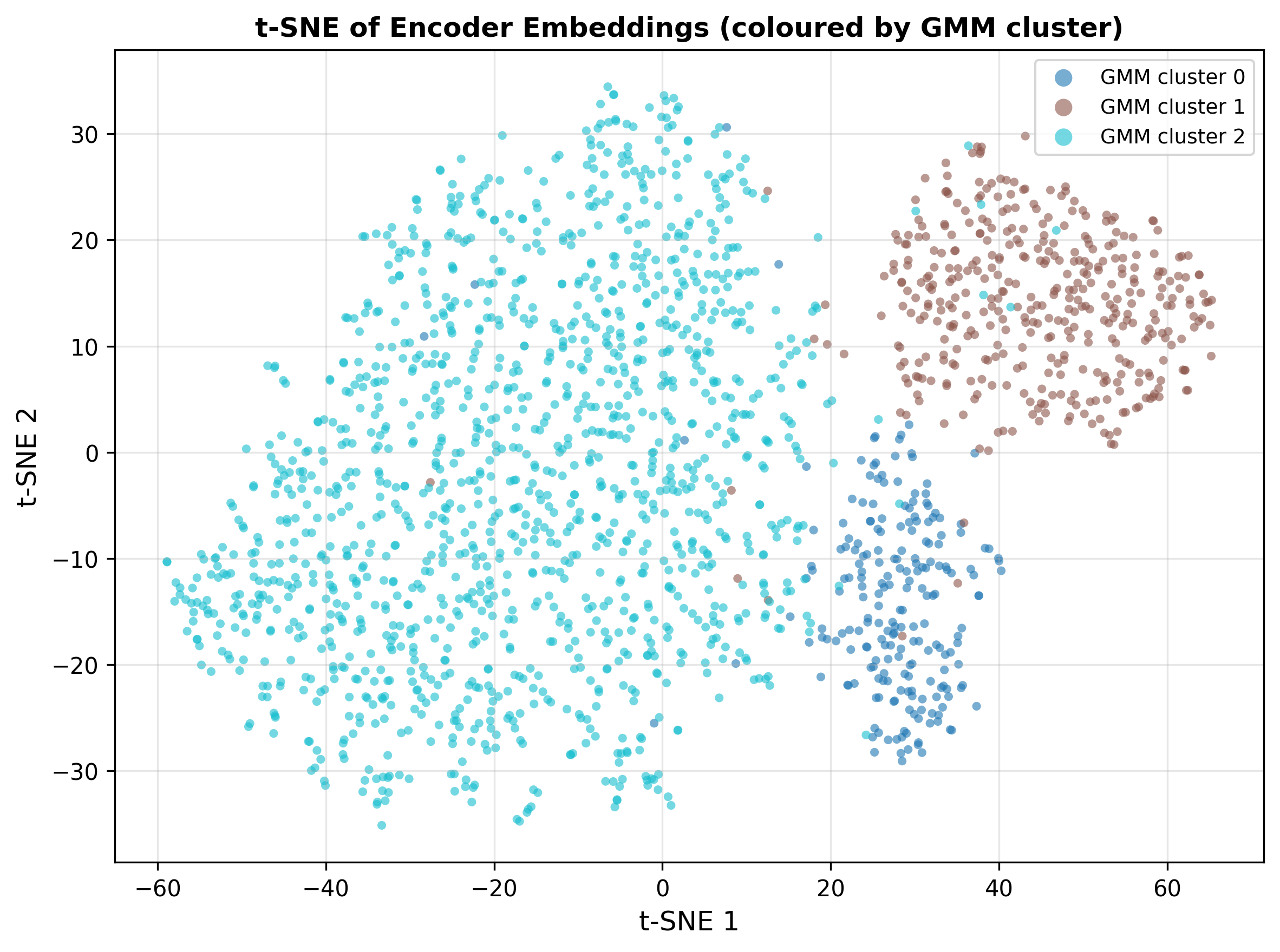} &
\includegraphics[width=0.14\textwidth]{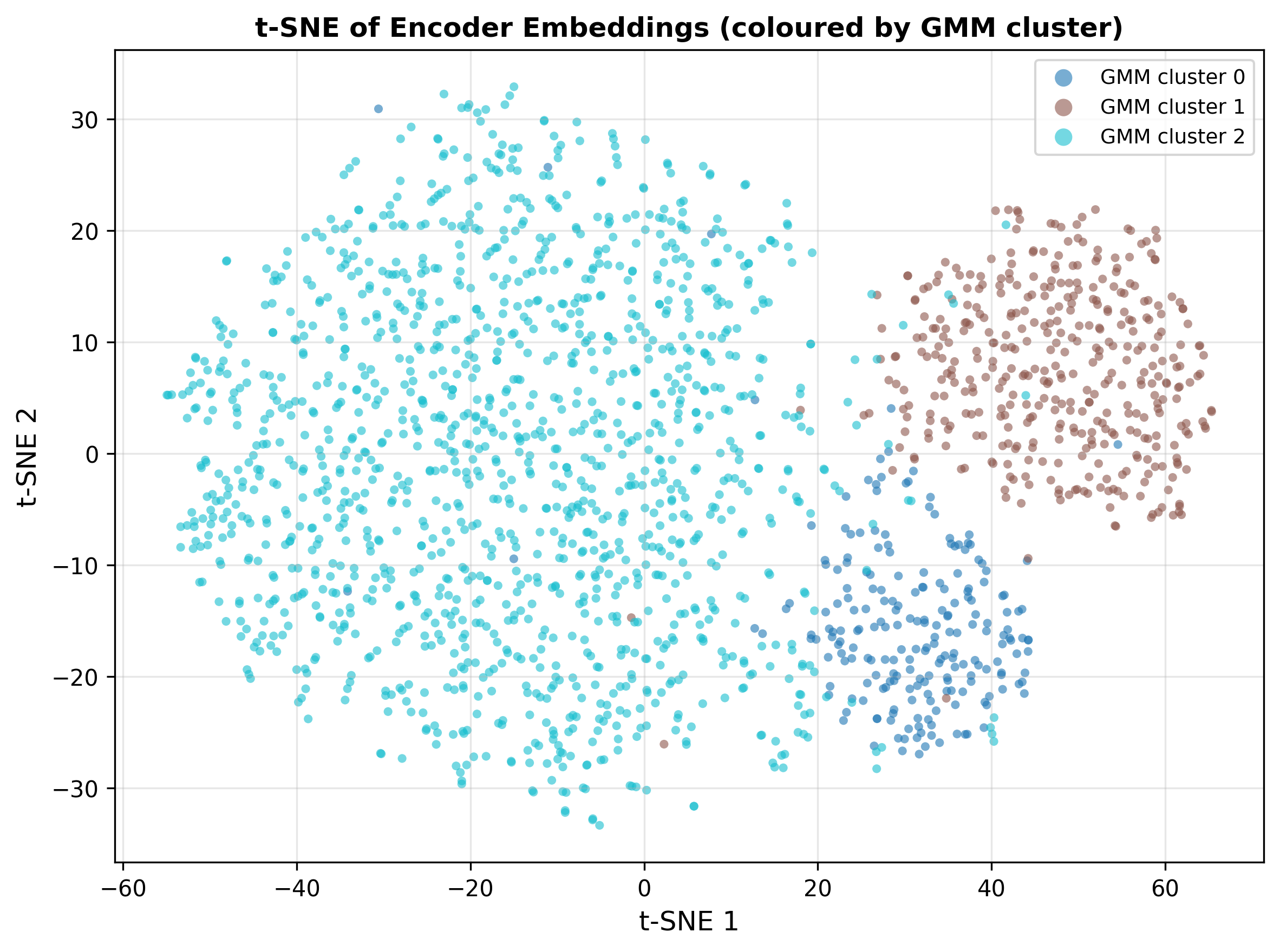} &
\includegraphics[width=0.14\textwidth]{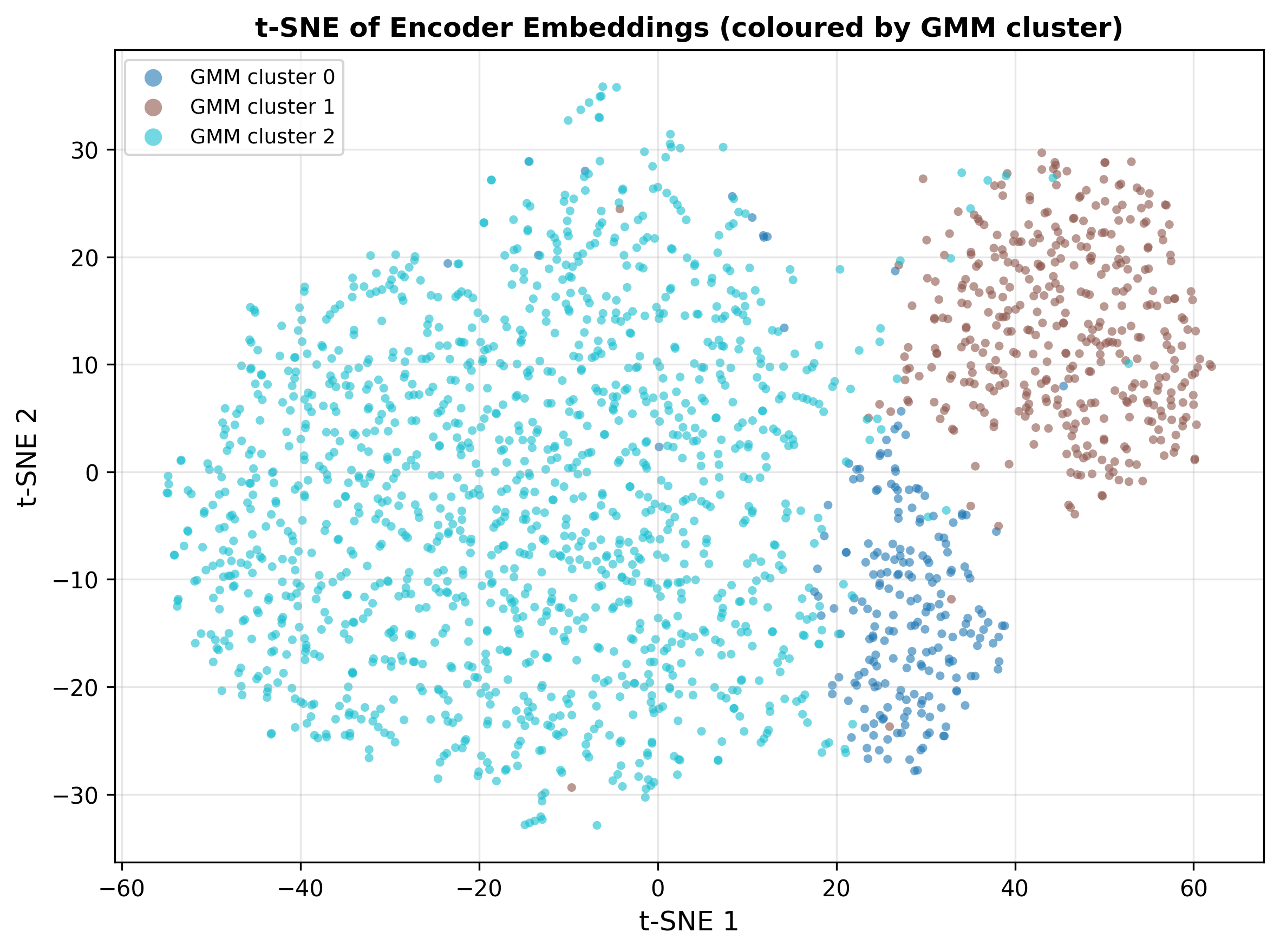} &
\includegraphics[width=0.14\textwidth]{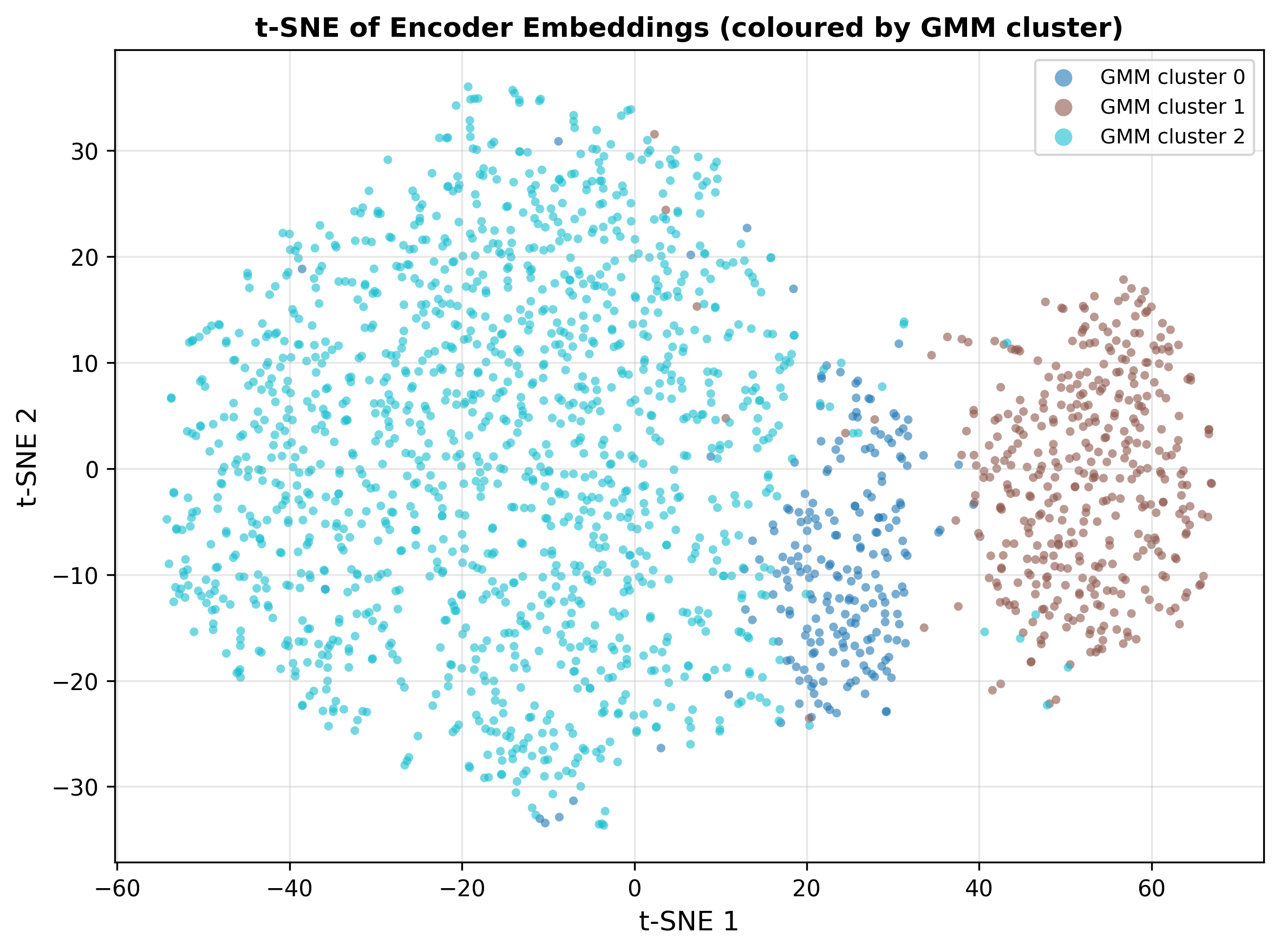} &
\includegraphics[width=0.14\textwidth]{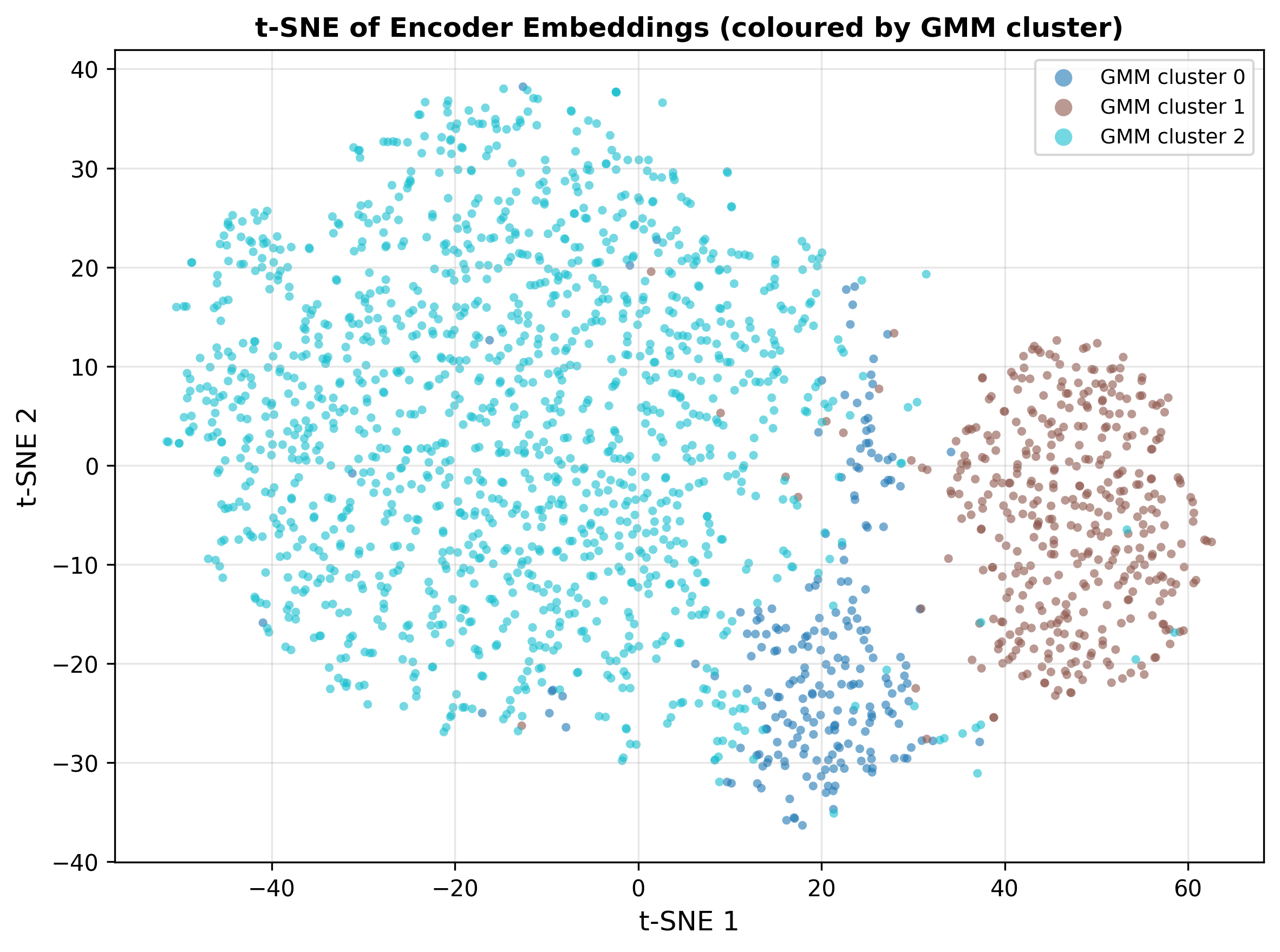} &
\includegraphics[width=0.14\textwidth]{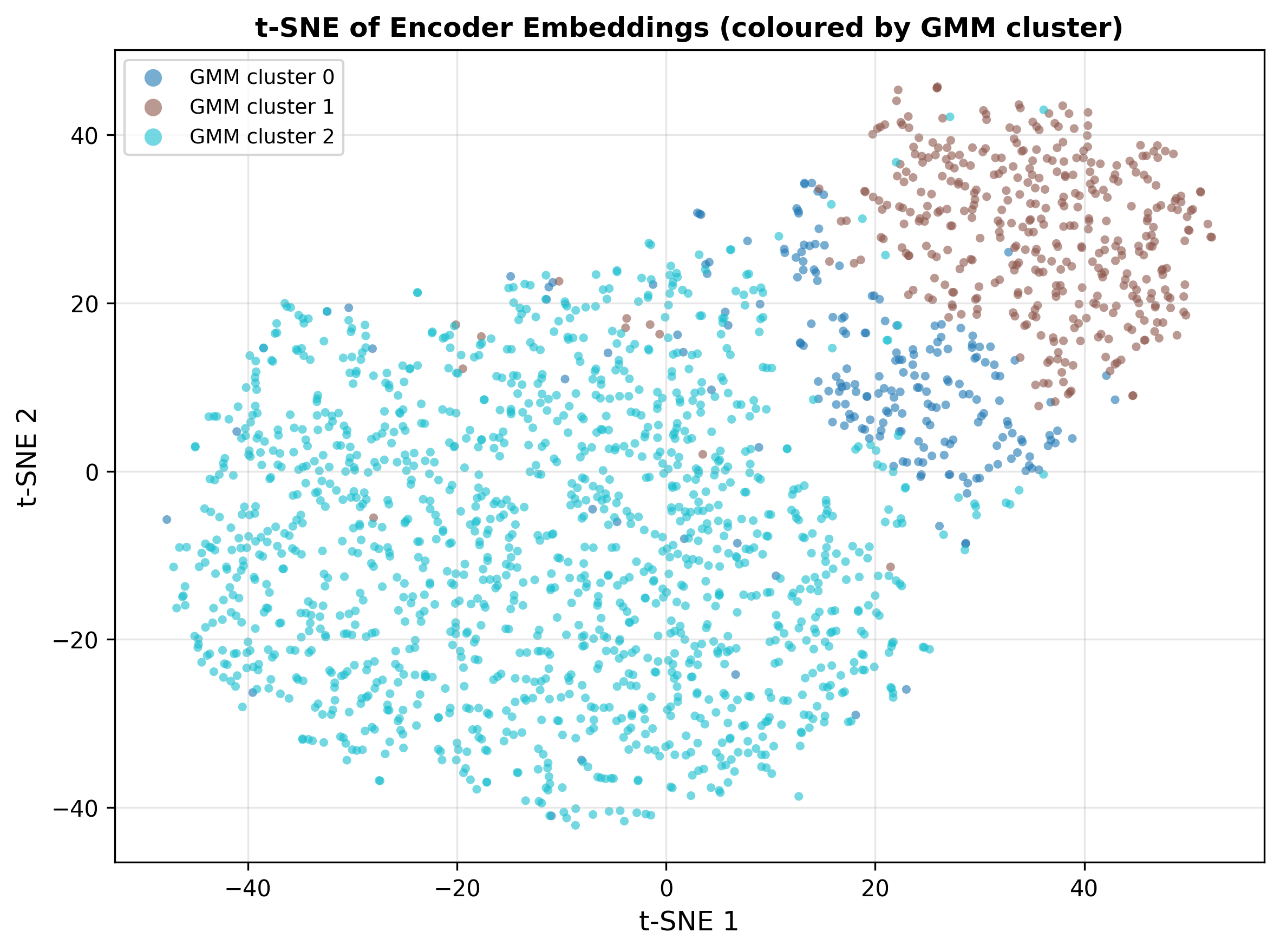}
\\
\multicolumn{6}{c}{$d=10$}
\\[0.5em]

\includegraphics[width=0.14\textwidth]{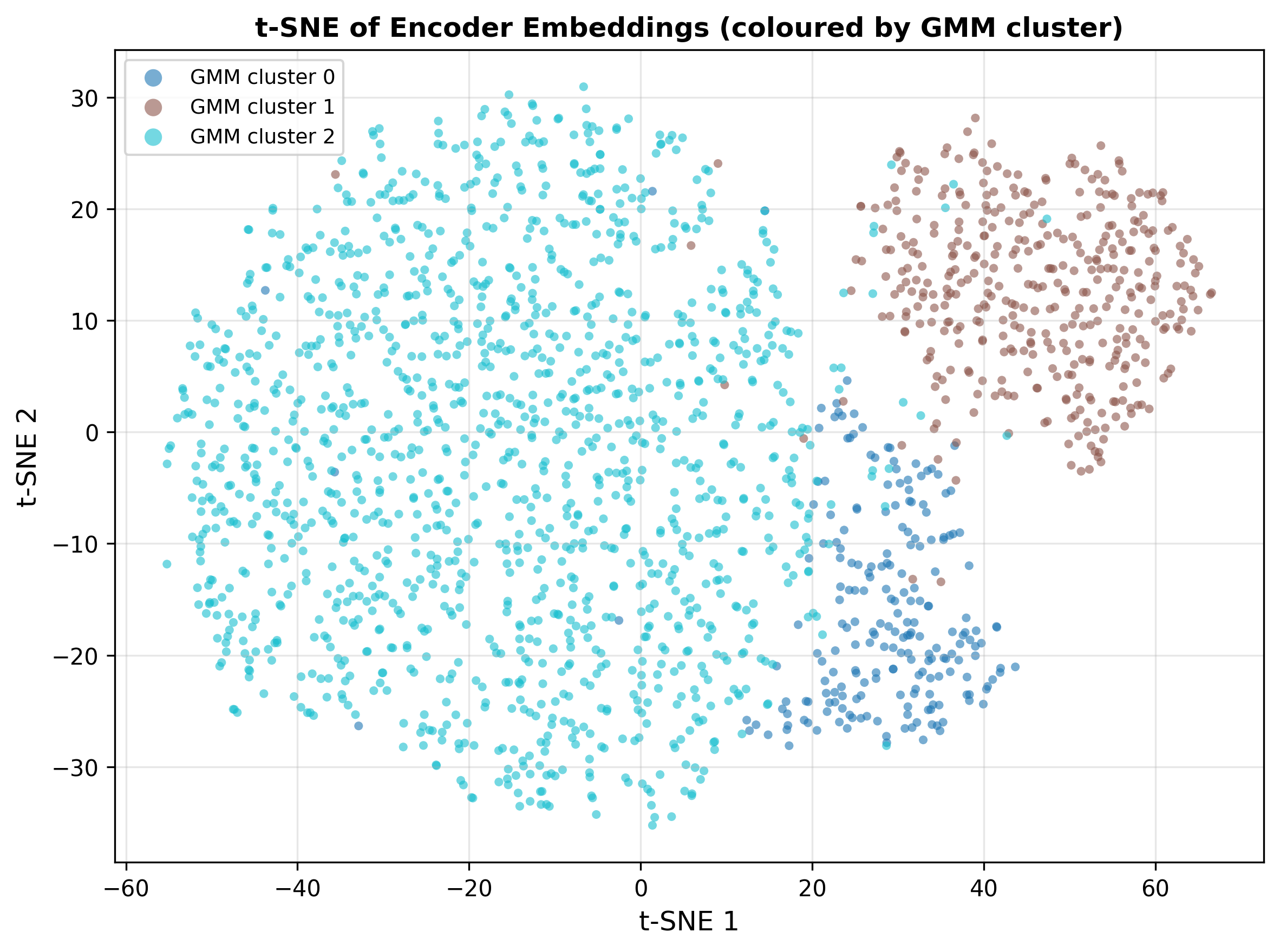} &
\includegraphics[width=0.14\textwidth]{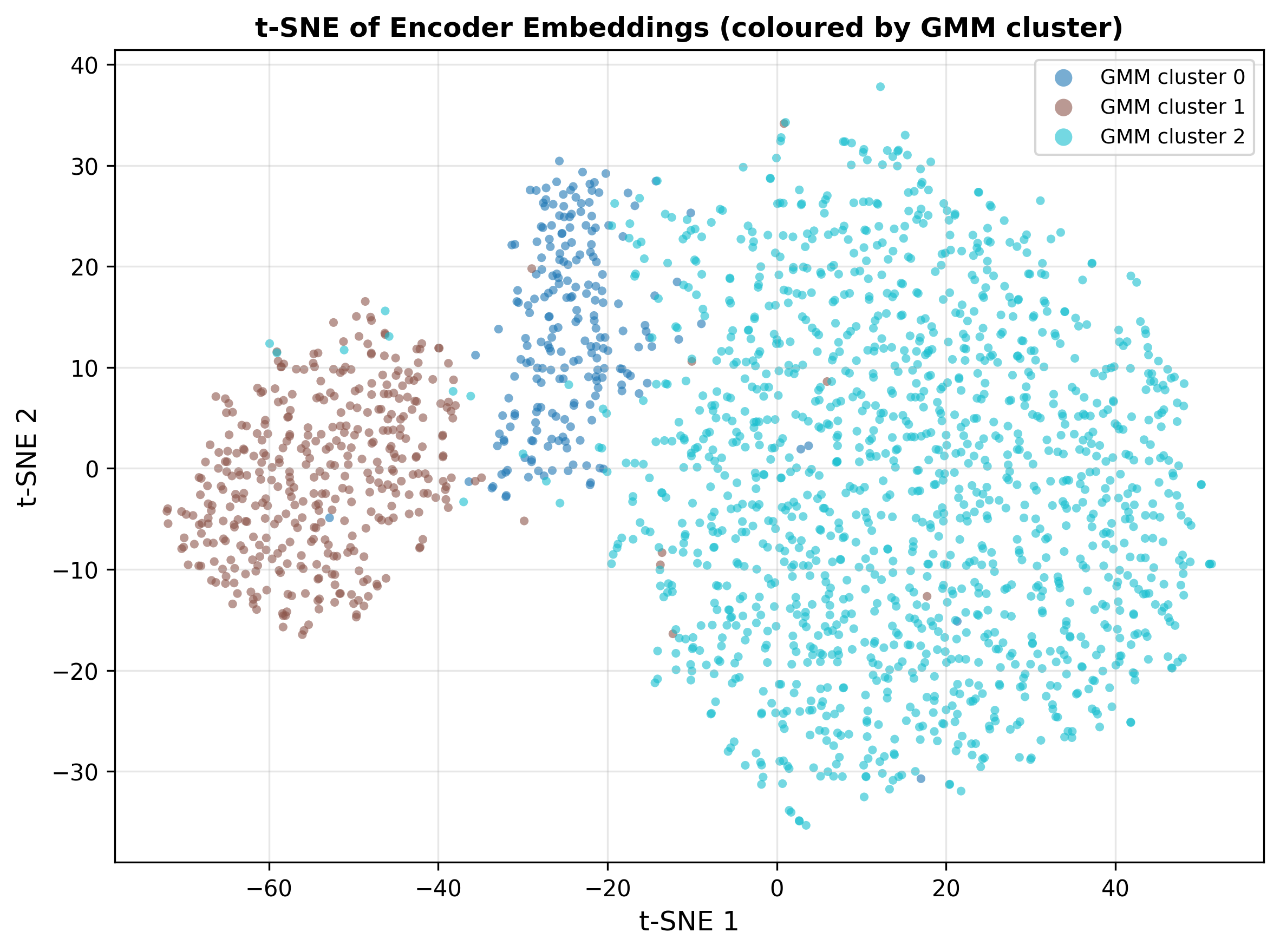} &
\includegraphics[width=0.14\textwidth]{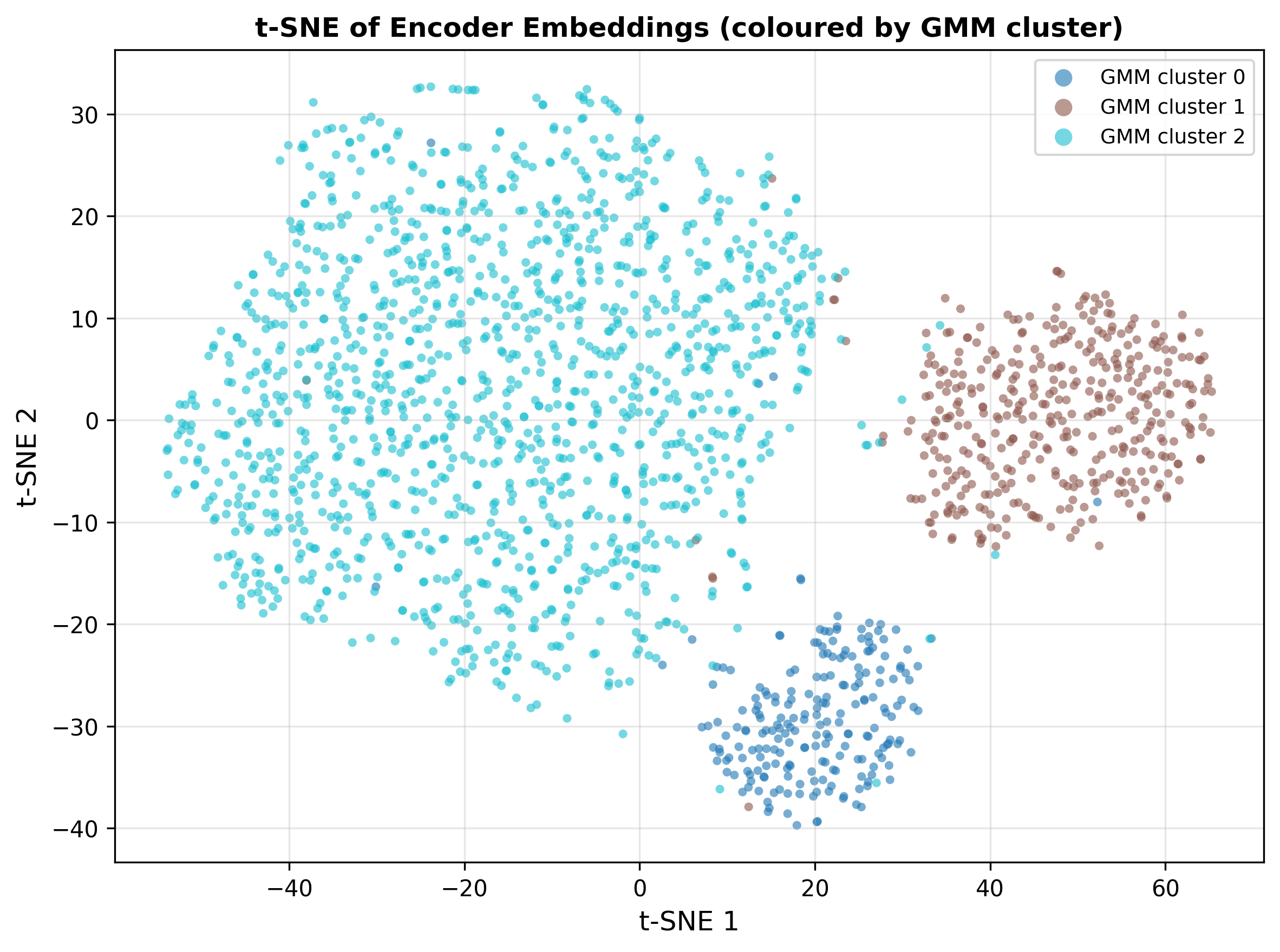} &
\includegraphics[width=0.14\textwidth]{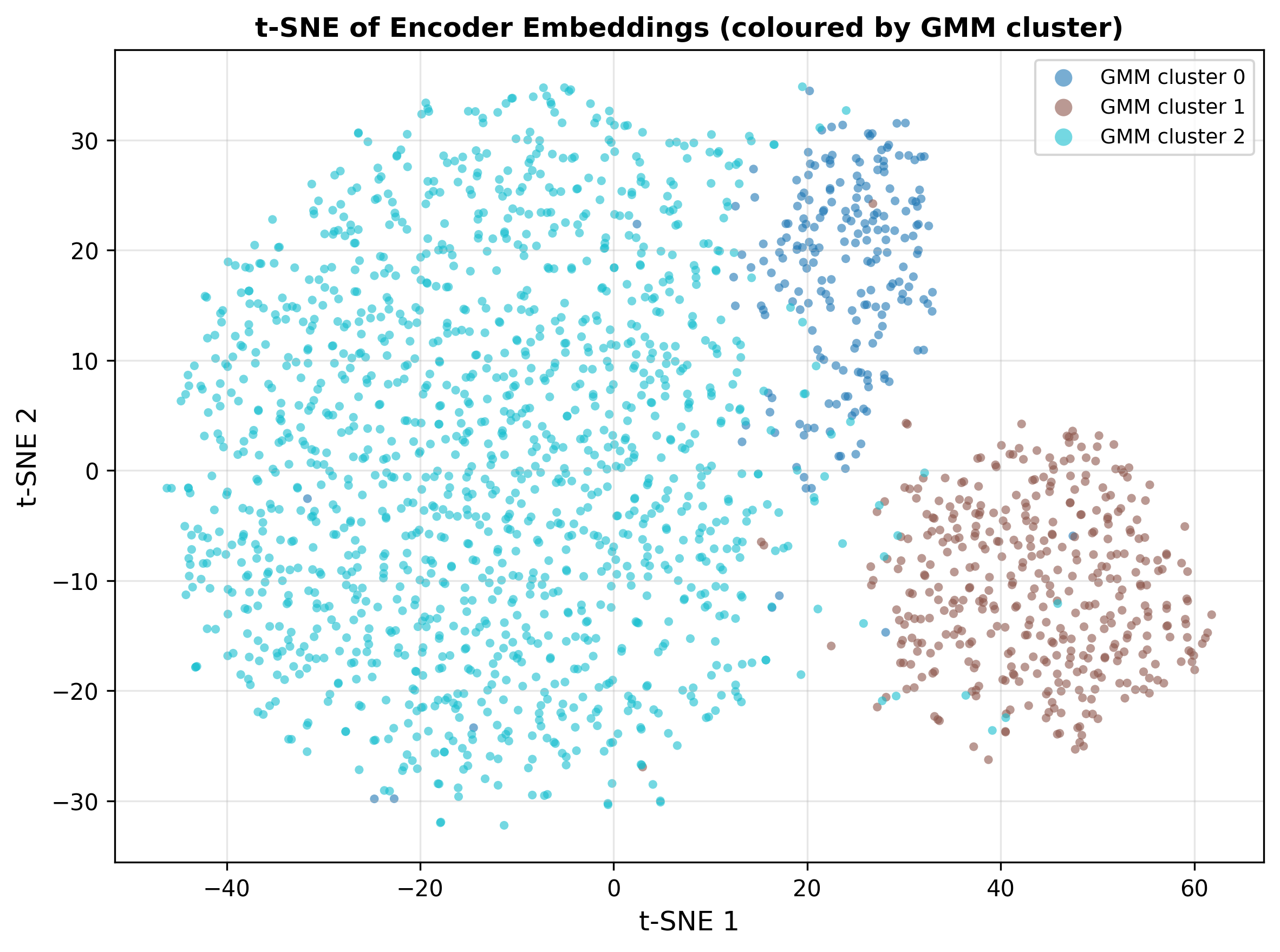} &
\includegraphics[width=0.14\textwidth]{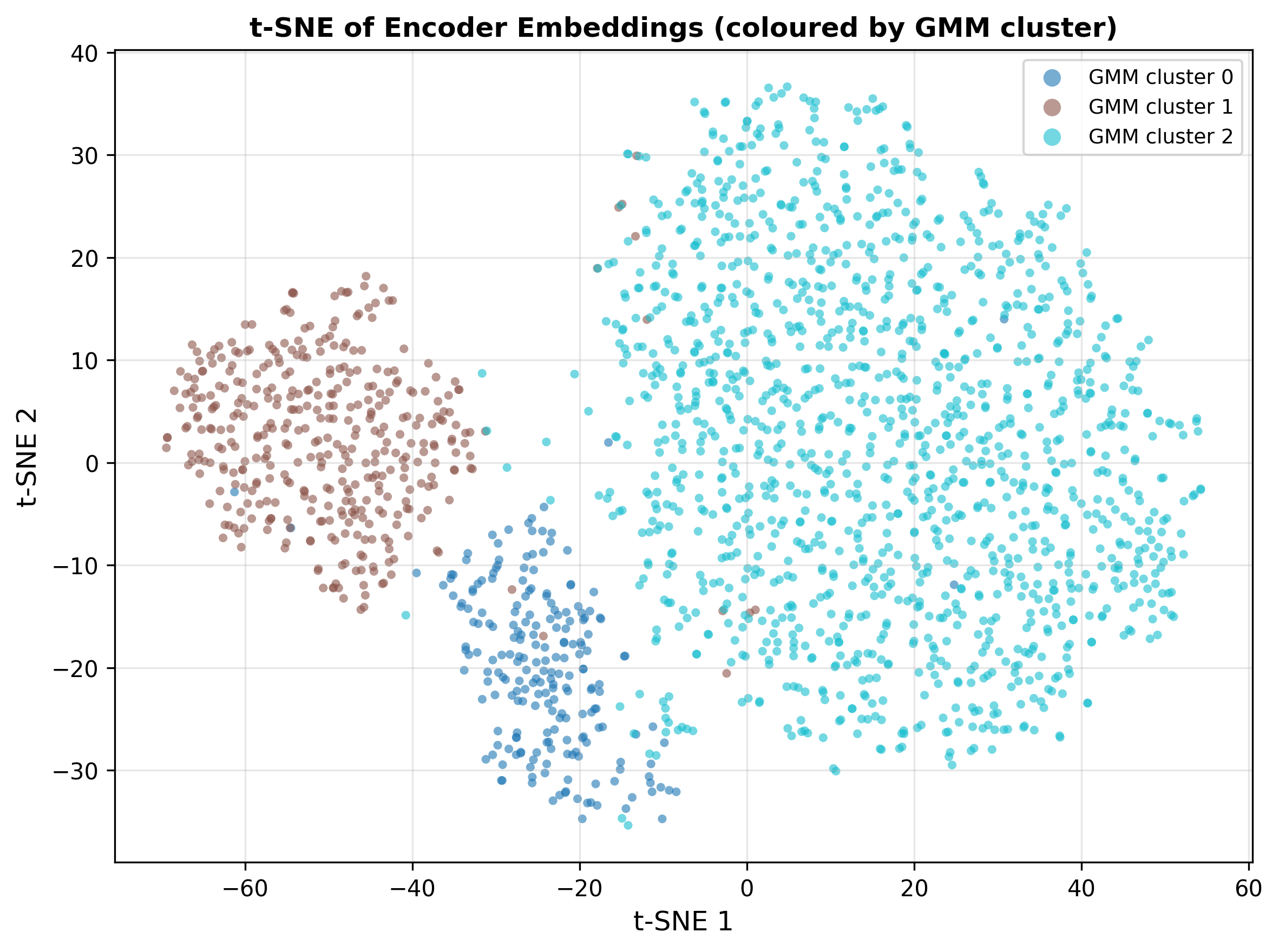} &
\includegraphics[width=0.14\textwidth]{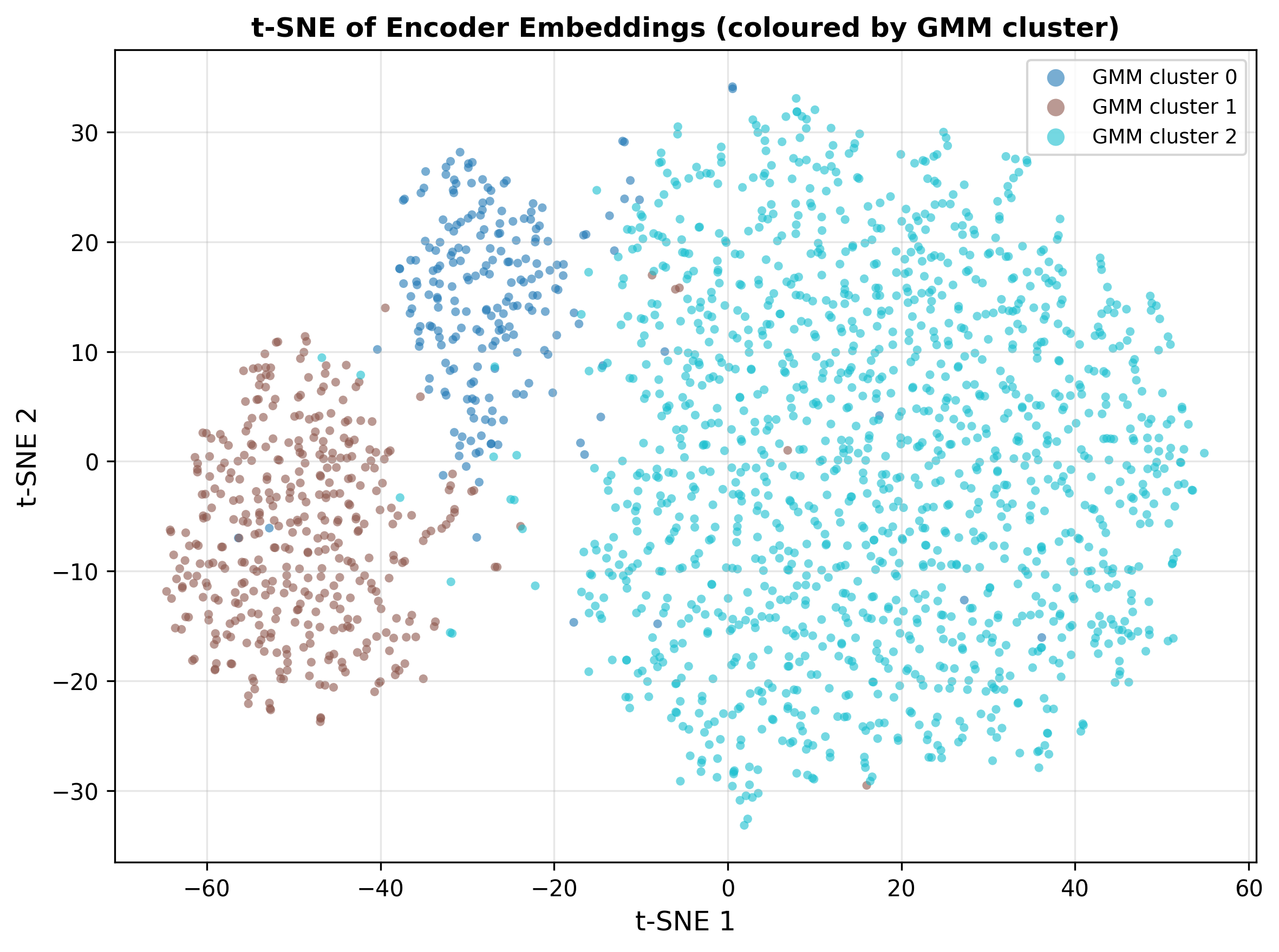}
\\
\multicolumn{6}{c}{$d=15$}
\\[0.5em]

\includegraphics[width=0.14\textwidth]{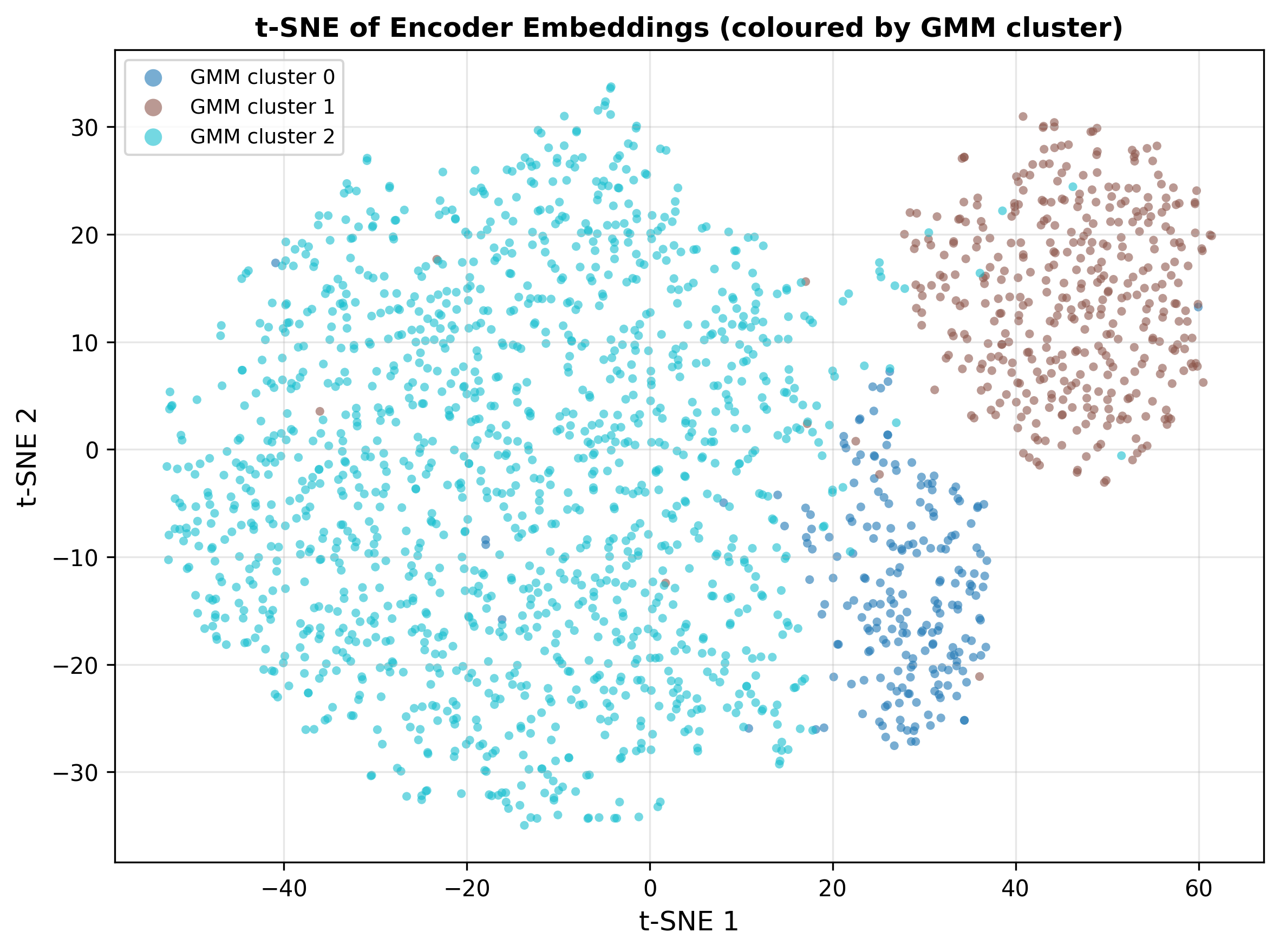} &
\includegraphics[width=0.14\textwidth]{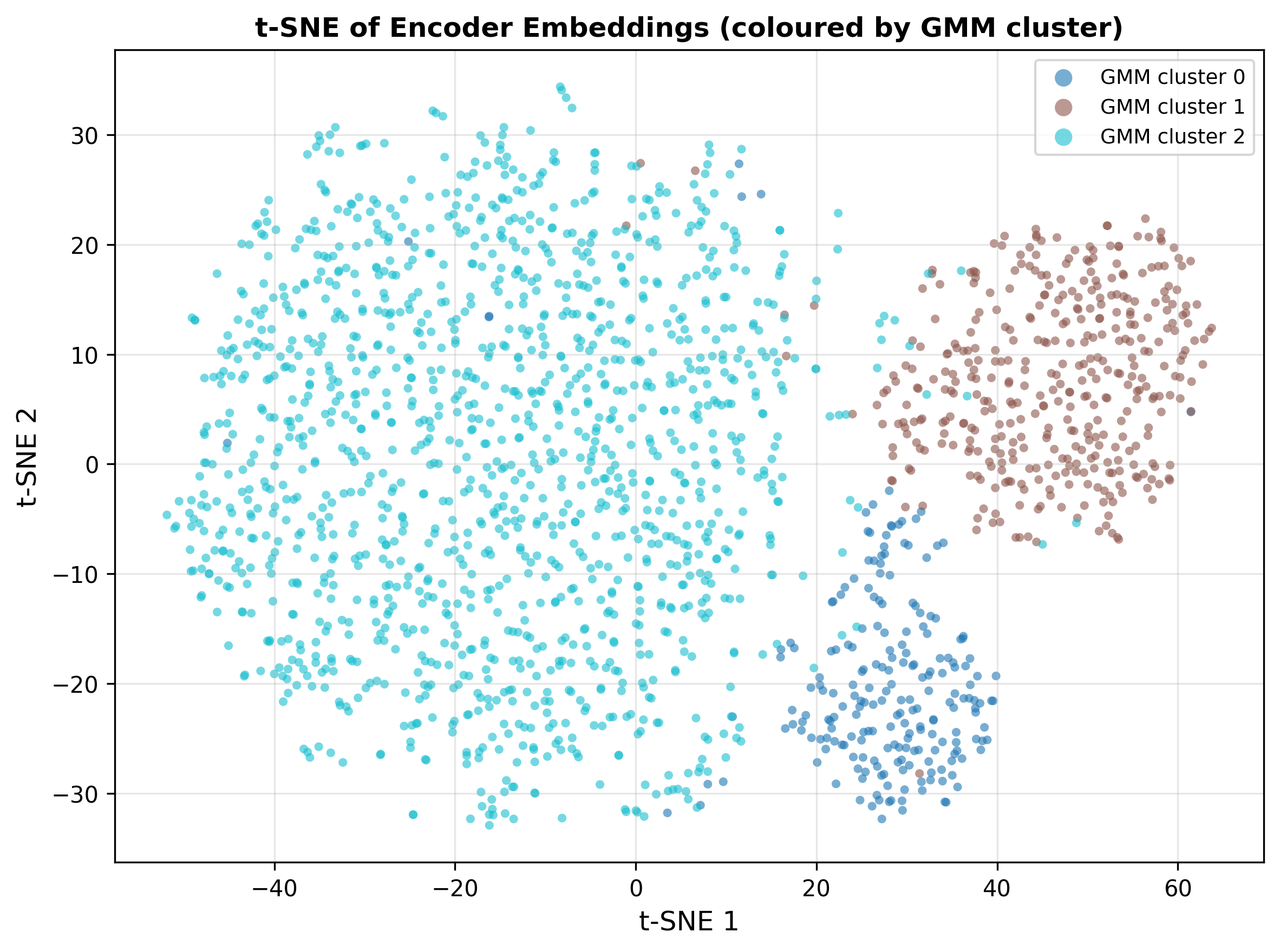} &
\includegraphics[width=0.14\textwidth]{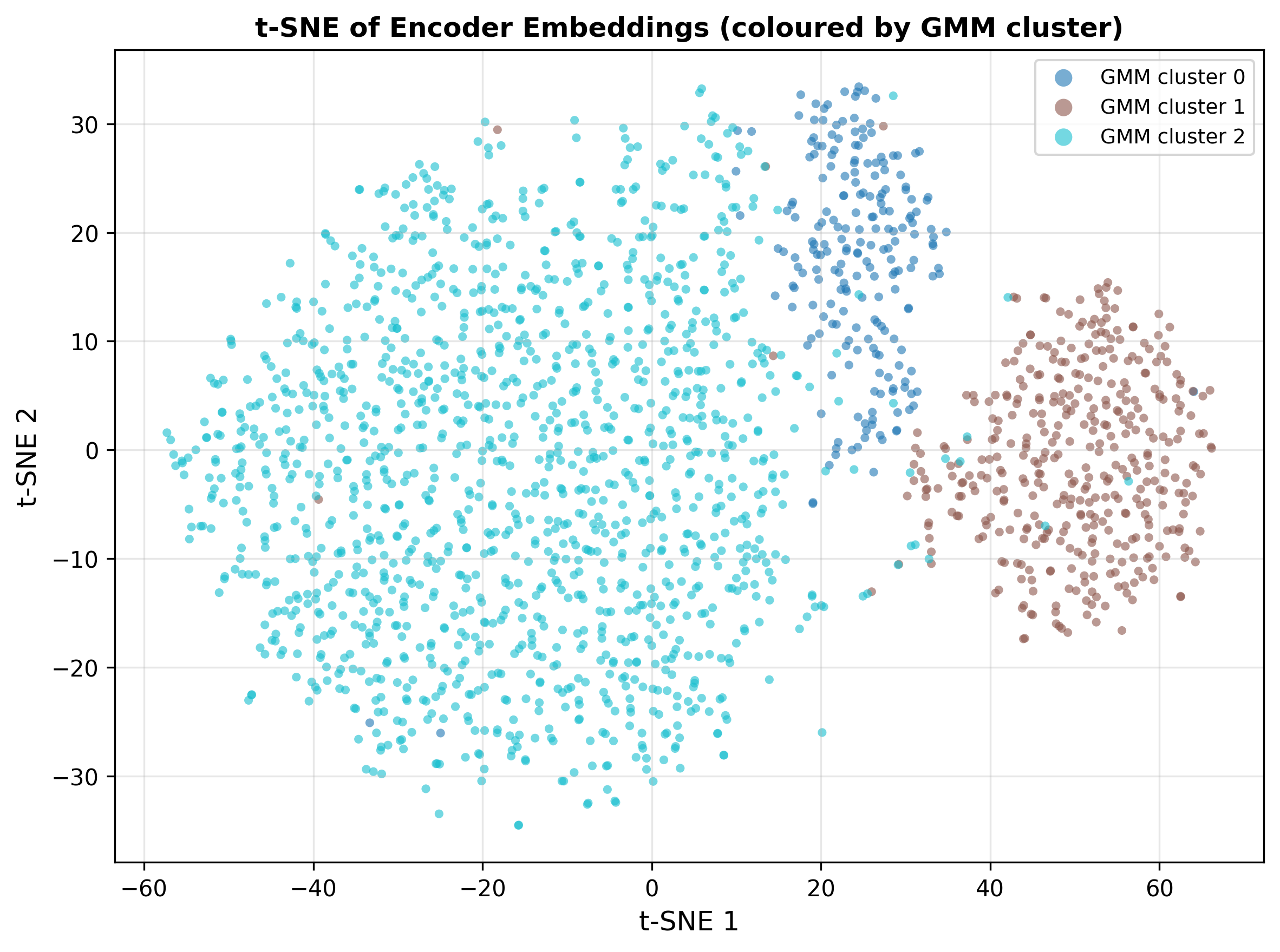} &
\includegraphics[width=0.14\textwidth]{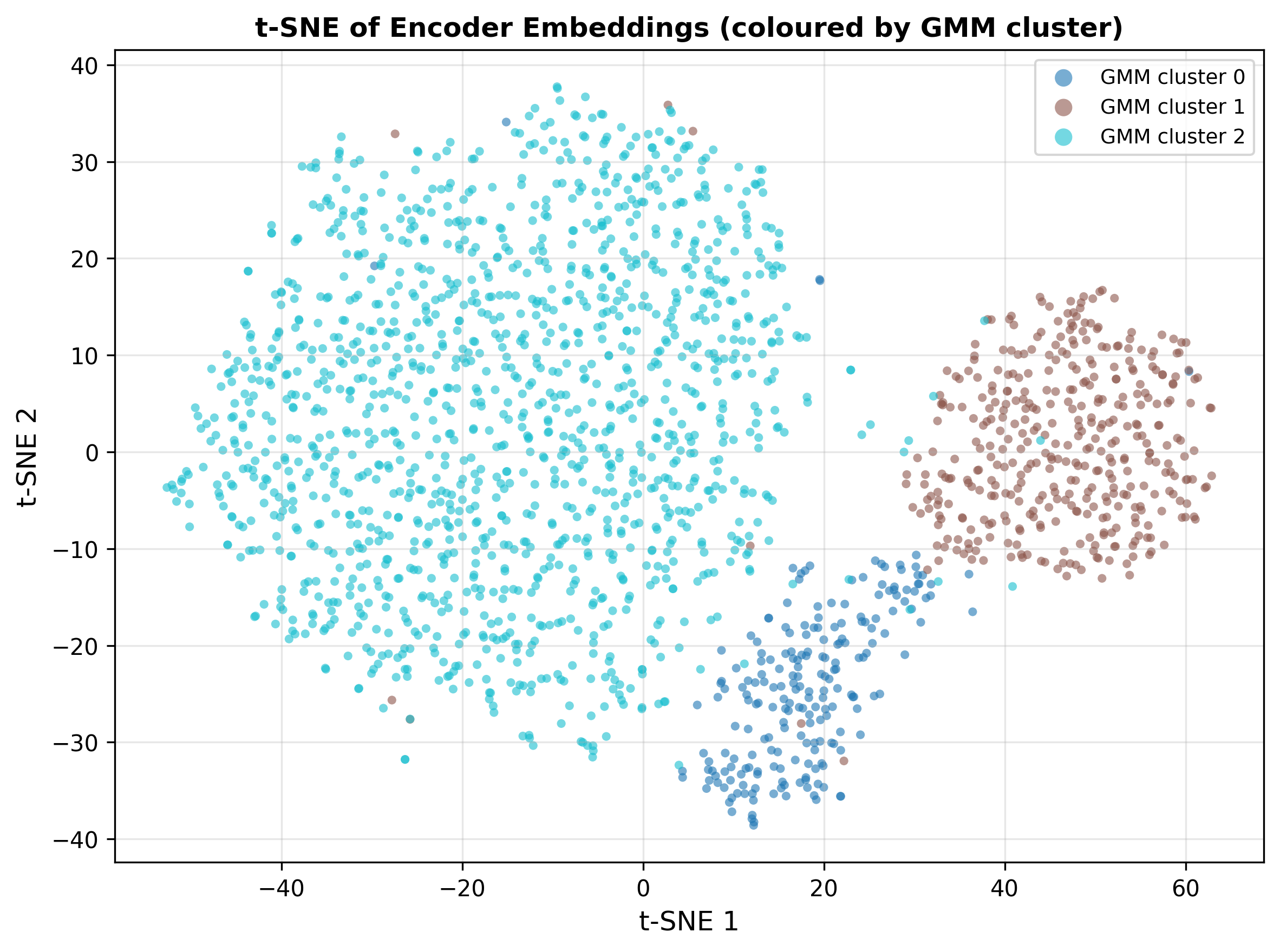} &
\includegraphics[width=0.14\textwidth]{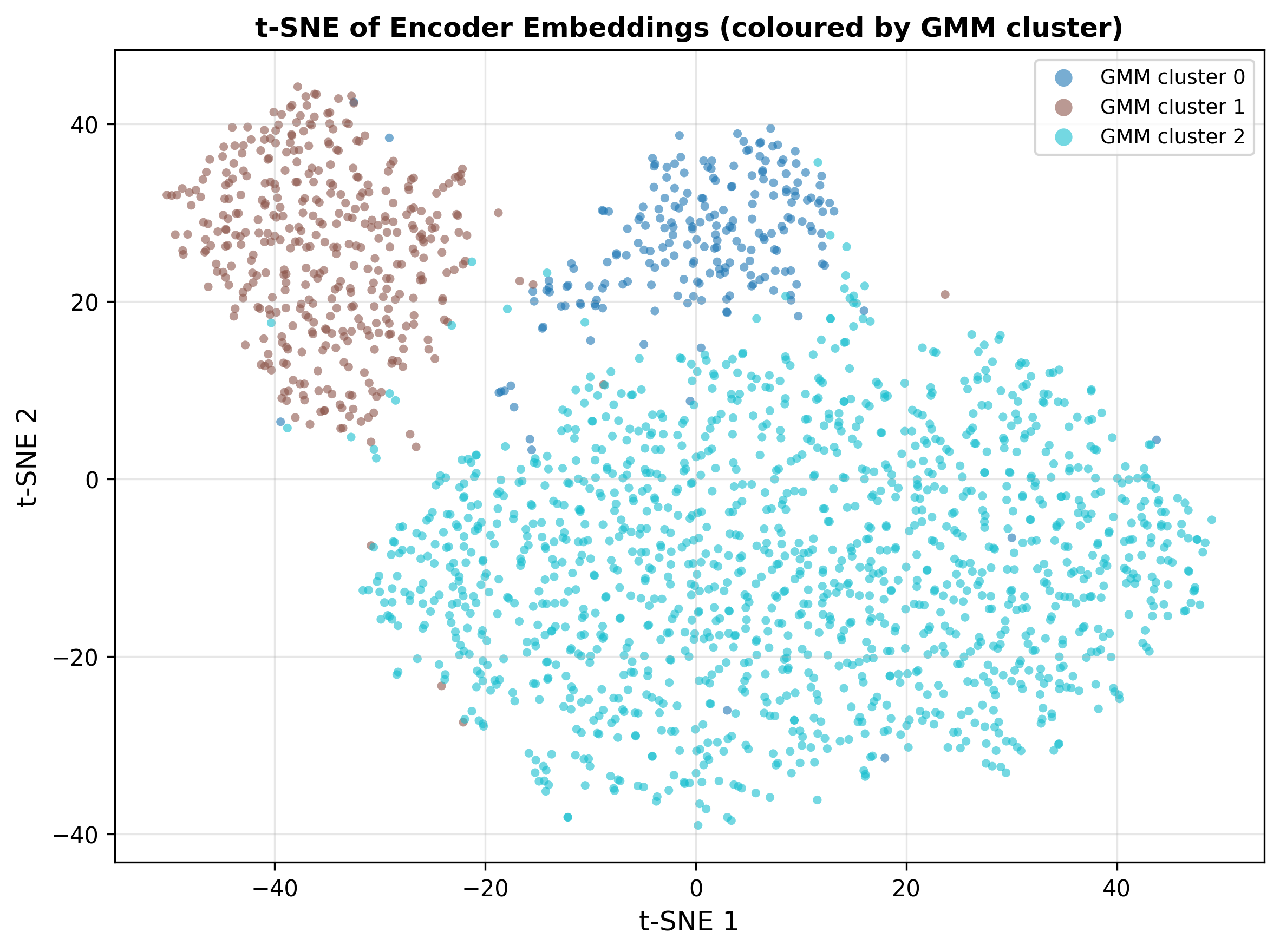} &
\includegraphics[width=0.14\textwidth]{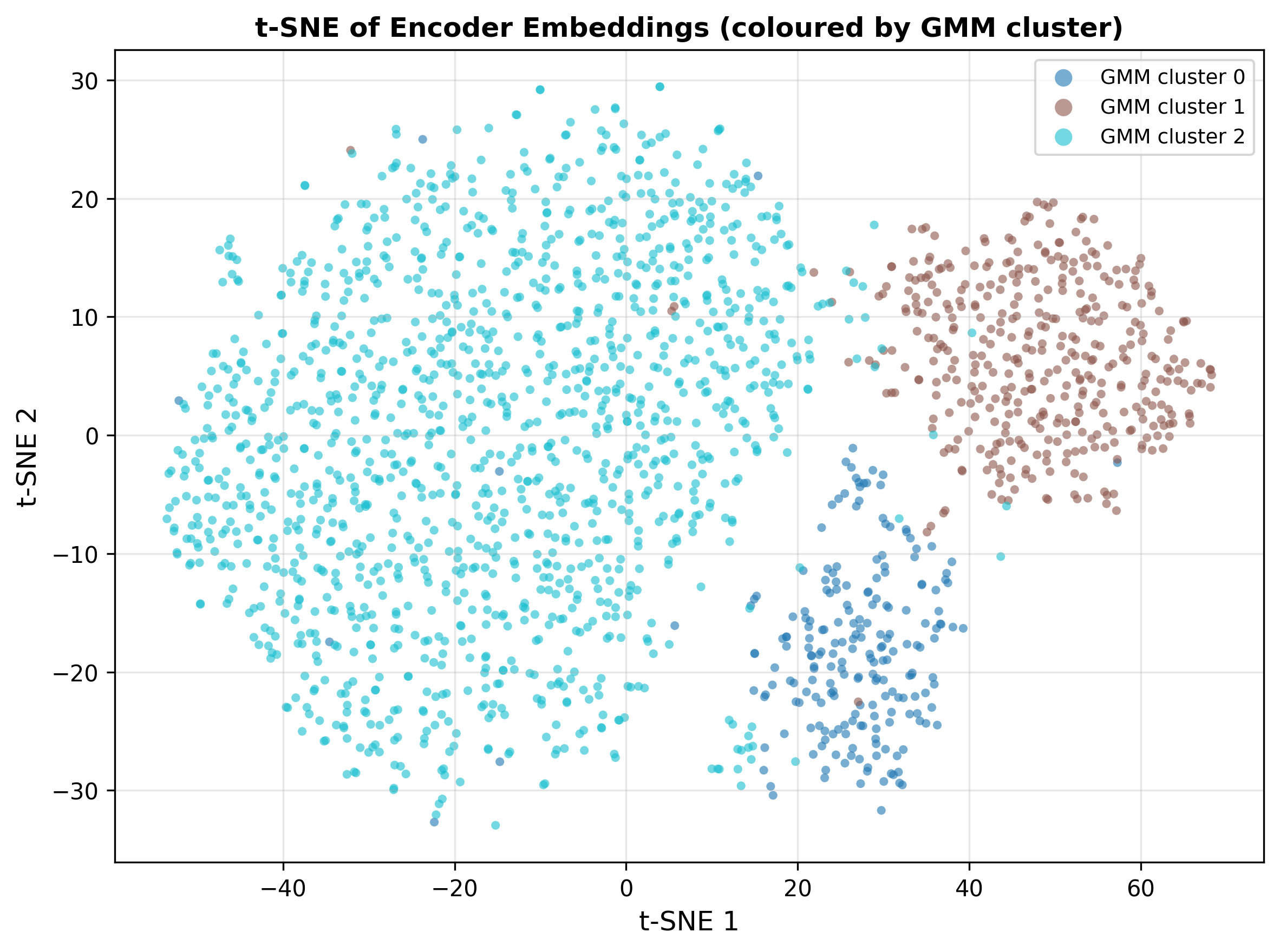}
\\
\multicolumn{6}{c}{$d=20$}
\end{tabular}

\caption{
t-SNE visualization of the learned representations for CI-Reg variants in the structured-factor simulation study with cluster separation 3.0. Rows correspond to encoder dimensions and columns correspond to different values of $\lambda_{\mathrm{CInd}}$. Points are colored according to the true latent Gaussian mixture component.
}
\label{fig:structured_factor_tsne_of_embeddings_sep3_ci_reg}
\end{figure}

\paragraph{Visualization of the learned representations.}

To further investigate the effect of the conditional independence regularizer on the learned latent representations, we visualize the embeddings using t-SNE \citep{JMLR:v9:vandermaaten08a}.
Figure \ref{fig:structured_factor_tsne_of_embeddings} shows the learned representations for all CI-Reg variants across the three encoder dimensions $d \in \{ 10, 15, 20\}$ in the structured-factor simulation study with cluster separation 5.0. Points are colored according to the true latent Gaussian mixture component.

Across all considered encoder dimensions and regularization strengths, the latent cluster structure remains clearly identifiable. This suggests that the conditional independence regularizer does not destroy meaningful latent structure, even for relatively large values of $\lambda_{\mathrm{CInd}}$. Instead, the learned representations continue to preserve the dominant factors governing the data-generating process, indicating that the proposed regularization could be incorporated without sacrificing the ability to recover the underlying latent clusters.

Figure \ref{fig:structured_factor_tsne_of_embeddings_sep3_ci_reg} presents the corresponding t-SNE visualizations for a more challenging setting with cluster separation 3.0. As expected, the learned latent clusters are less separated than in Figure \ref{fig:structured_factor_tsne_of_embeddings} due to the increased overlap between true mixture components. Importantly, the inclusion of the conditional independence regularizer does not appear to harm recovery of the latent structure. In several settings, moderate values of $\lambda_{\mathrm{CInd}}$ produce visibly clearer cluster separation than the unregularized model CI-Reg-0. Examples include $\lambda_{\mathrm{CInd}}=0.01$ for $d=10$, $\lambda_{\mathrm{CInd}} \in \{ 0.03, 0.30 \}$ for $d=15$, and $\lambda_{\mathrm{CInd}} \in \{ 0.01, 0.30, 0.90 \}$ for $d=20$. These visualizations suggest that the proposed regularizer can, in certain cases, help the learned representation better capture the underlying latent cluster structure while continuing to preserve the dominant factors governing the data-generating process.

\subsection{Real data experiments} \label{appendix:Additional results-Real data experiments}


Figures \ref{fig:real_ks_statistic}--\ref{fig:real_energy_distance_cat} provide additional evaluations of reconstruction performance on the real-world datasets. Figure \ref{fig:real_ks_statistic} reports the column-wise Kolmogorov--Smirnov (KS) statistic for numerical variables, Figure \ref{fig:real_combined_assoc_relative_frobenius} reports the combined association relative Frobenius error, and Figures \ref{fig:real_energy_distance_cont} and \ref{fig:real_energy_distance_cat} separately report the energy distances for the numerical and categorical components. For the real-world datasets considered in this paper, numerical features often have substantially different scales and variances. Therefore, all energy distances are computed after standardizing the numerical variables, ensuring that each numerical feature contributes equally to the metric and preventing variables with larger scales from dominating the evaluation.

For the KS statistic, CI-Reg and CI-Reg-0 achieve competitive performance across datasets and remain close to the DPA baseline. In most settings, all three distributional approaches outperform VAE and are generally comparable to or better than AE, indicating the benefit of explicitly modeling conditional distributions rather than focusing primarily on conditional means.

For the combined association relative Frobenius error, CI-Reg consistently achieves strong performance across all datasets. In contrast, AE, DPA, and especially VAE often exhibit substantially larger errors, reflecting limitations in recovering the full mixed-type conditional distribution.

\begin{figure}[htbp]
\centering

\begin{subfigure}[b]{0.48\textwidth}
    \centering
    \includegraphics[width=\textwidth]{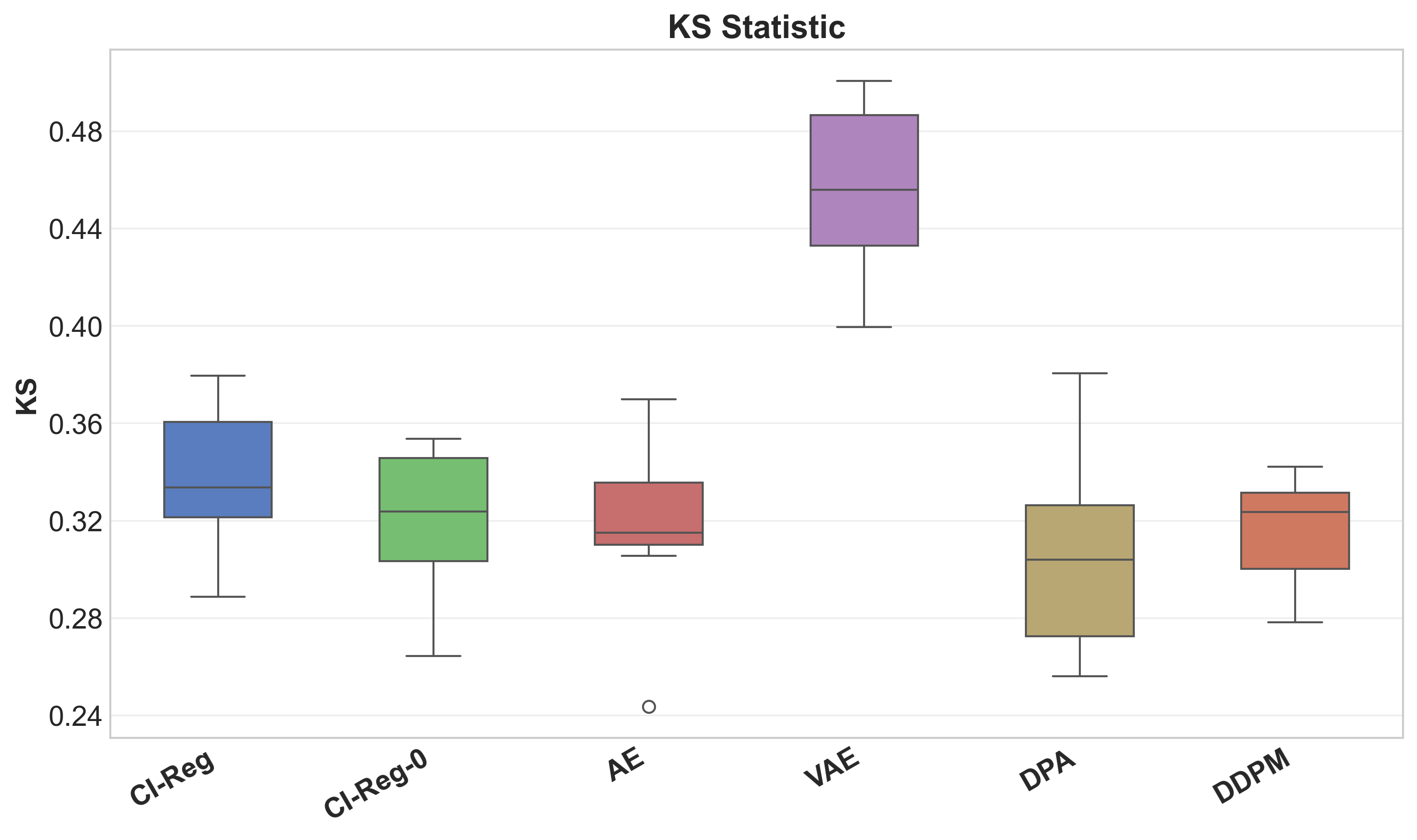}
    \caption{Adult}
\end{subfigure}
\hfill
\begin{subfigure}[b]{0.48\textwidth}
    \centering
    \includegraphics[width=\textwidth]{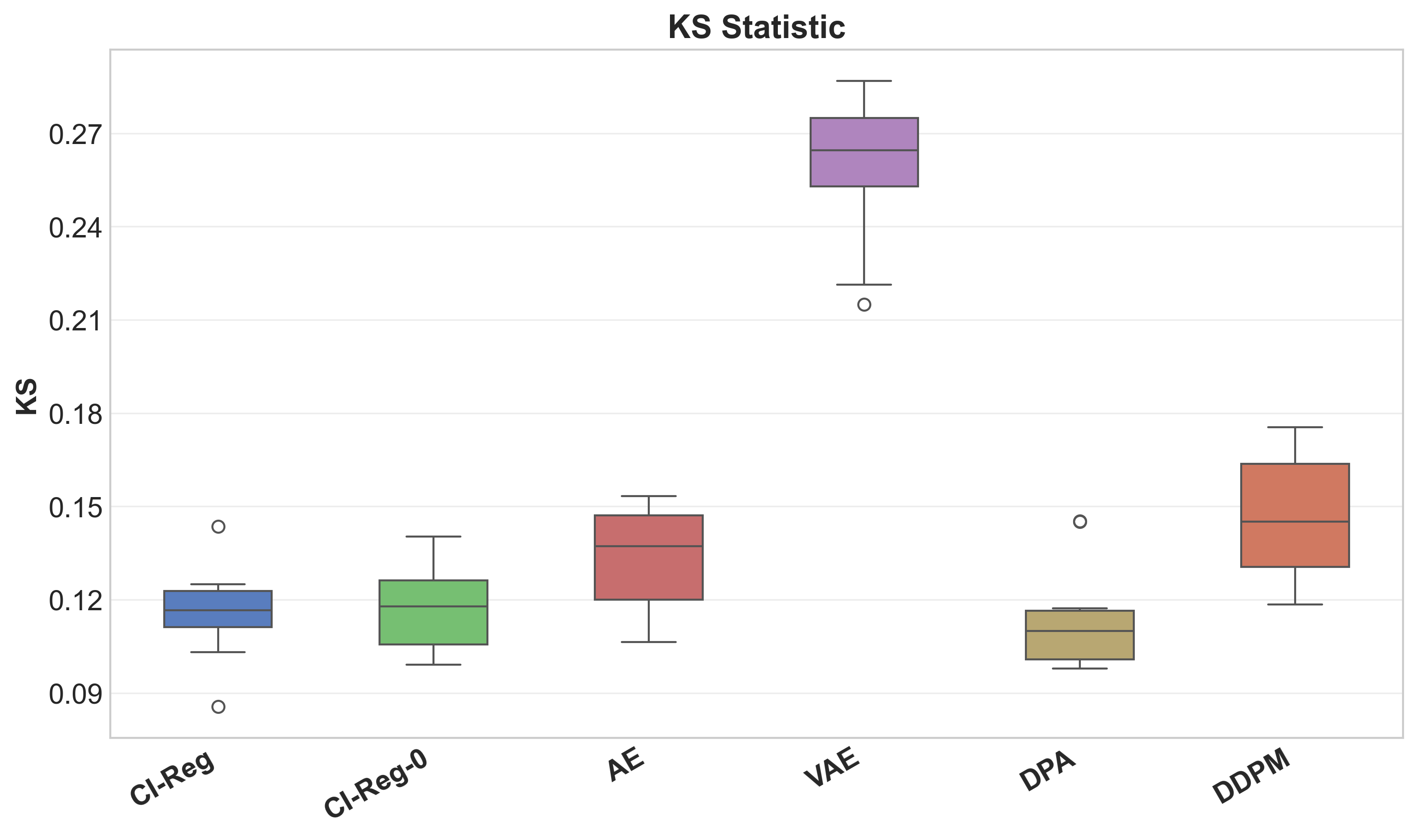}
    \caption{Insurance}
\end{subfigure}

\vspace{0.5em}

\begin{subfigure}[b]{0.48\textwidth}
    \centering
    \includegraphics[width=\textwidth]{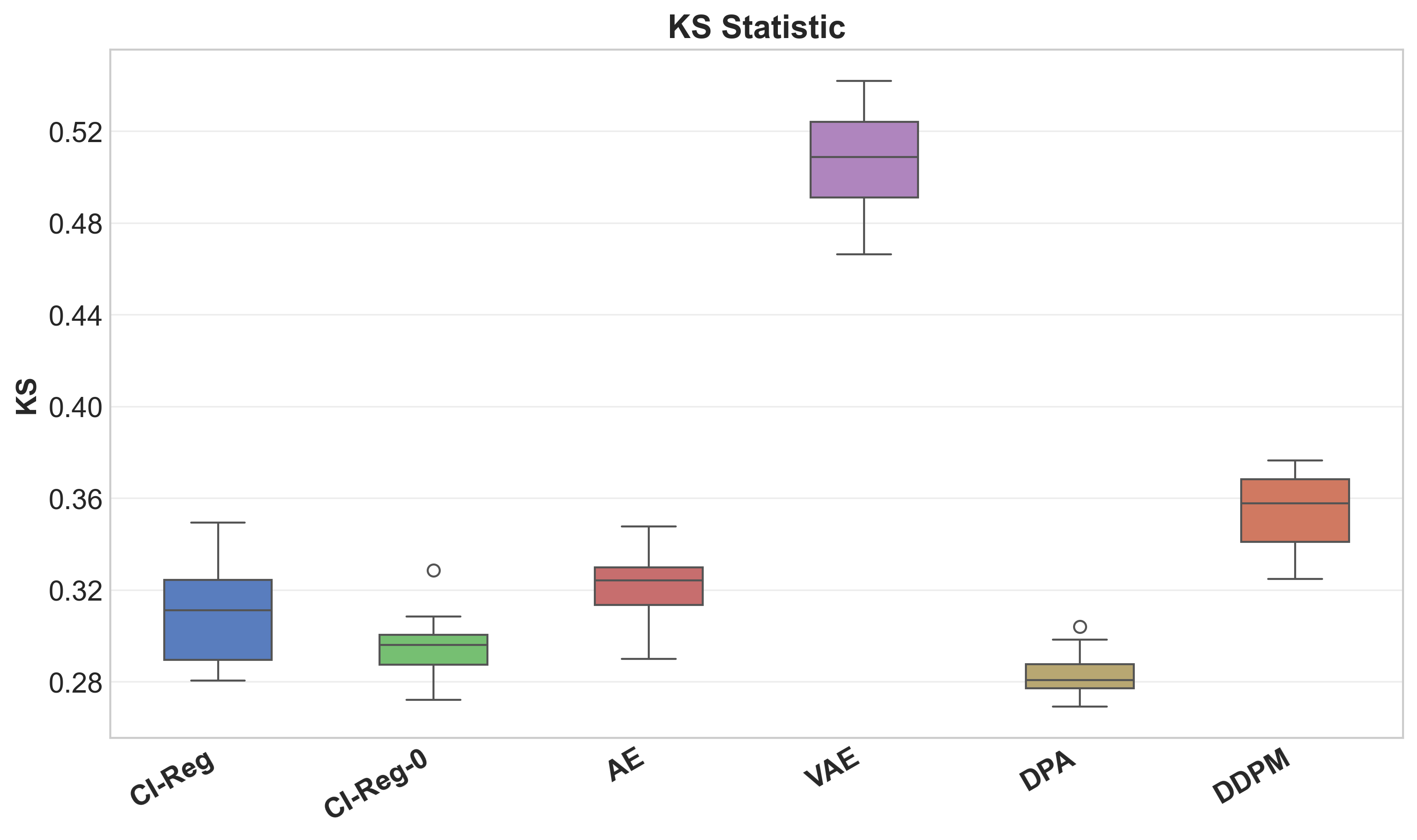}
    \caption{Shoppers}
\end{subfigure}
\hfill
\begin{subfigure}[b]{0.48\textwidth}
    \centering
    \includegraphics[width=\textwidth]{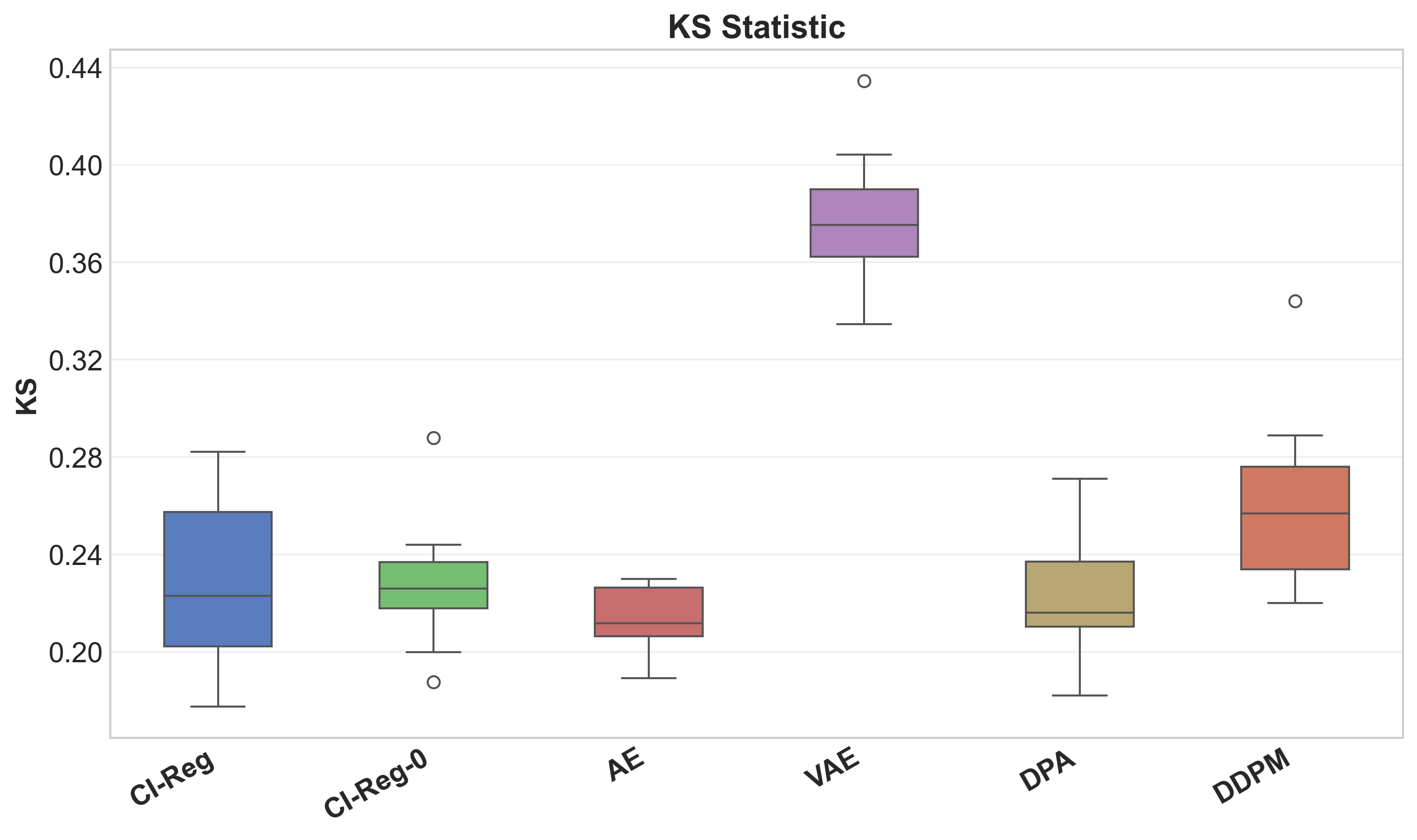}
    \caption{PM2.5}
\end{subfigure}

\caption{
 KS statistic for numerical components on the real-world datasets comparing CI-Reg, CI-Reg-0, AE, VAE, DPA, and DDPM. Lower values indicate better reconstruction performance.
}
\label{fig:real_ks_statistic}
\end{figure}

\begin{figure}[htbp]
\centering

\begin{subfigure}[b]{0.48\textwidth}
    \centering
    \includegraphics[width=\textwidth]{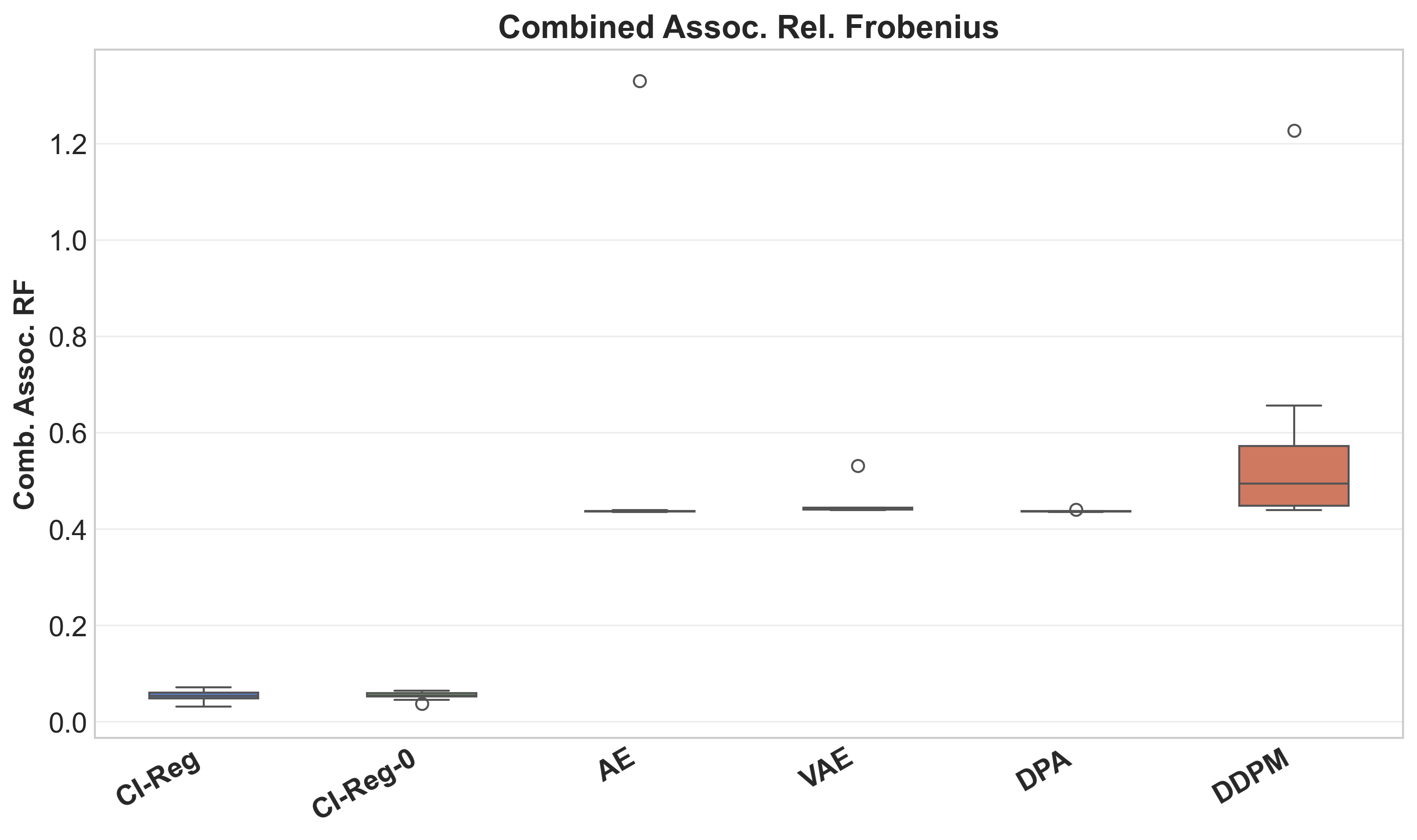}
    \caption{Adult}
\end{subfigure}
\hfill
\begin{subfigure}[b]{0.48\textwidth}
    \centering
    \includegraphics[width=\textwidth]{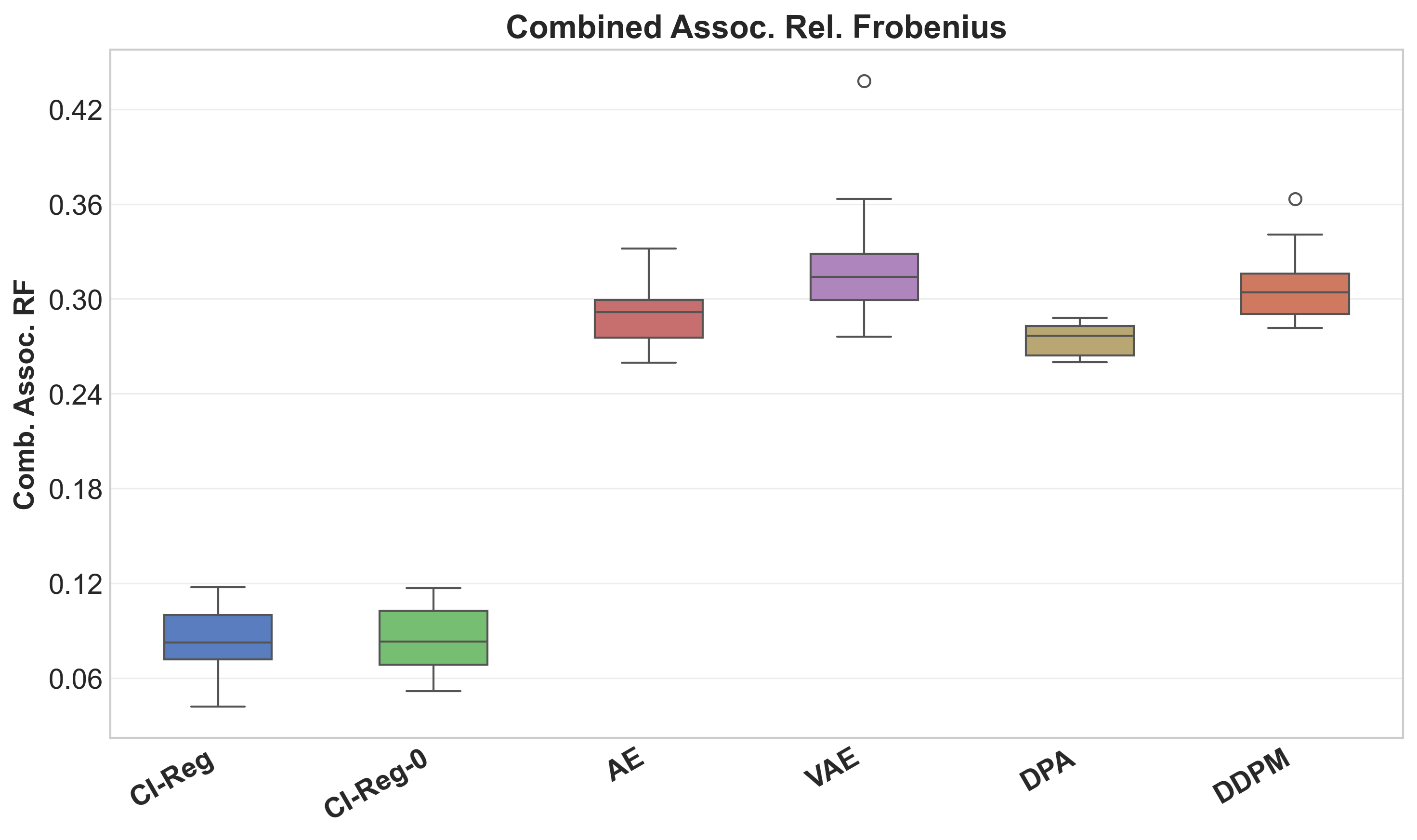}
    \caption{Insurance}
\end{subfigure}

\vspace{0.5em}

\begin{subfigure}[b]{0.48\textwidth}
    \centering
    \includegraphics[width=\textwidth]{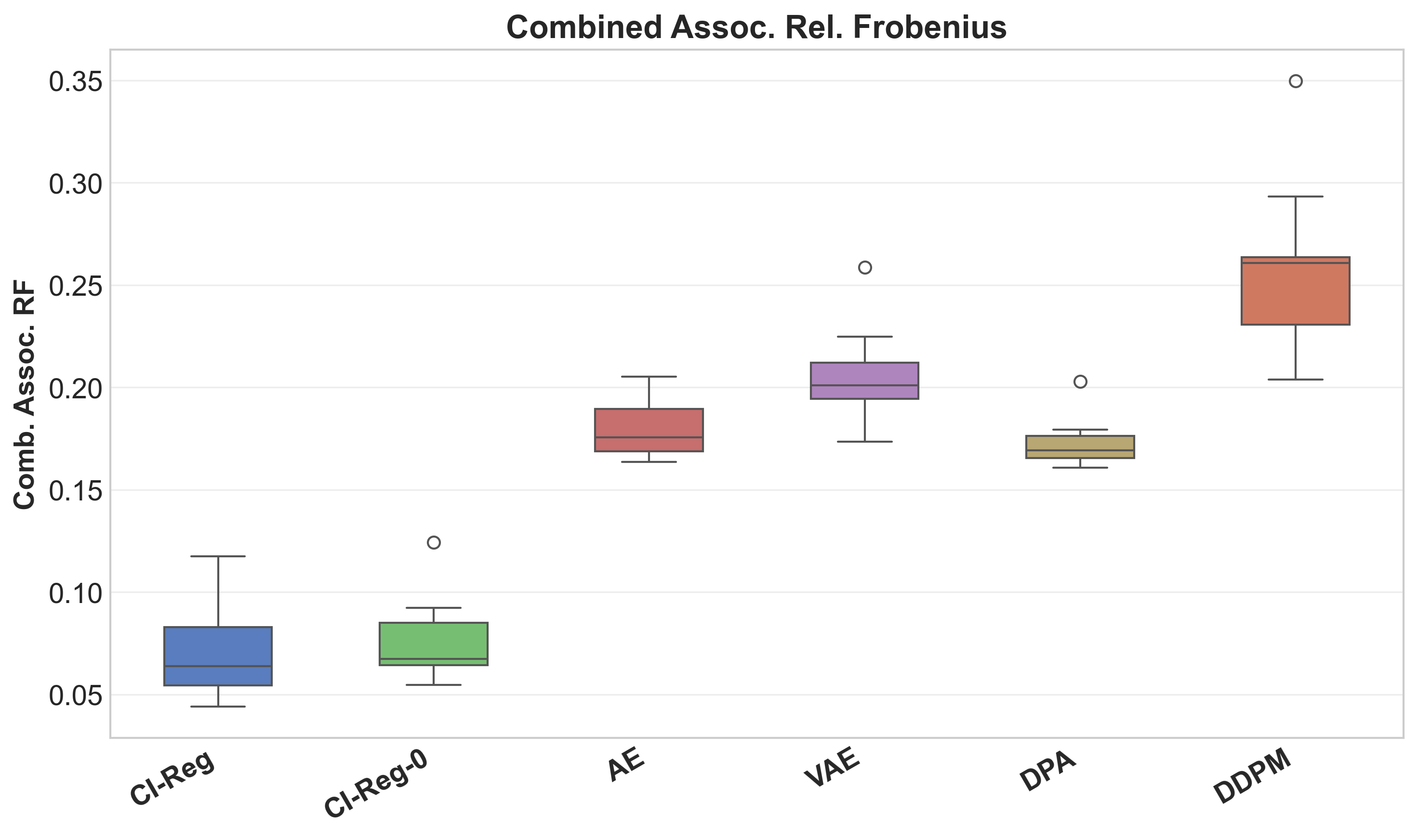}
    \caption{Shoppers}
\end{subfigure}
\hfill
\begin{subfigure}[b]{0.48\textwidth}
    \centering
    \includegraphics[width=\textwidth]{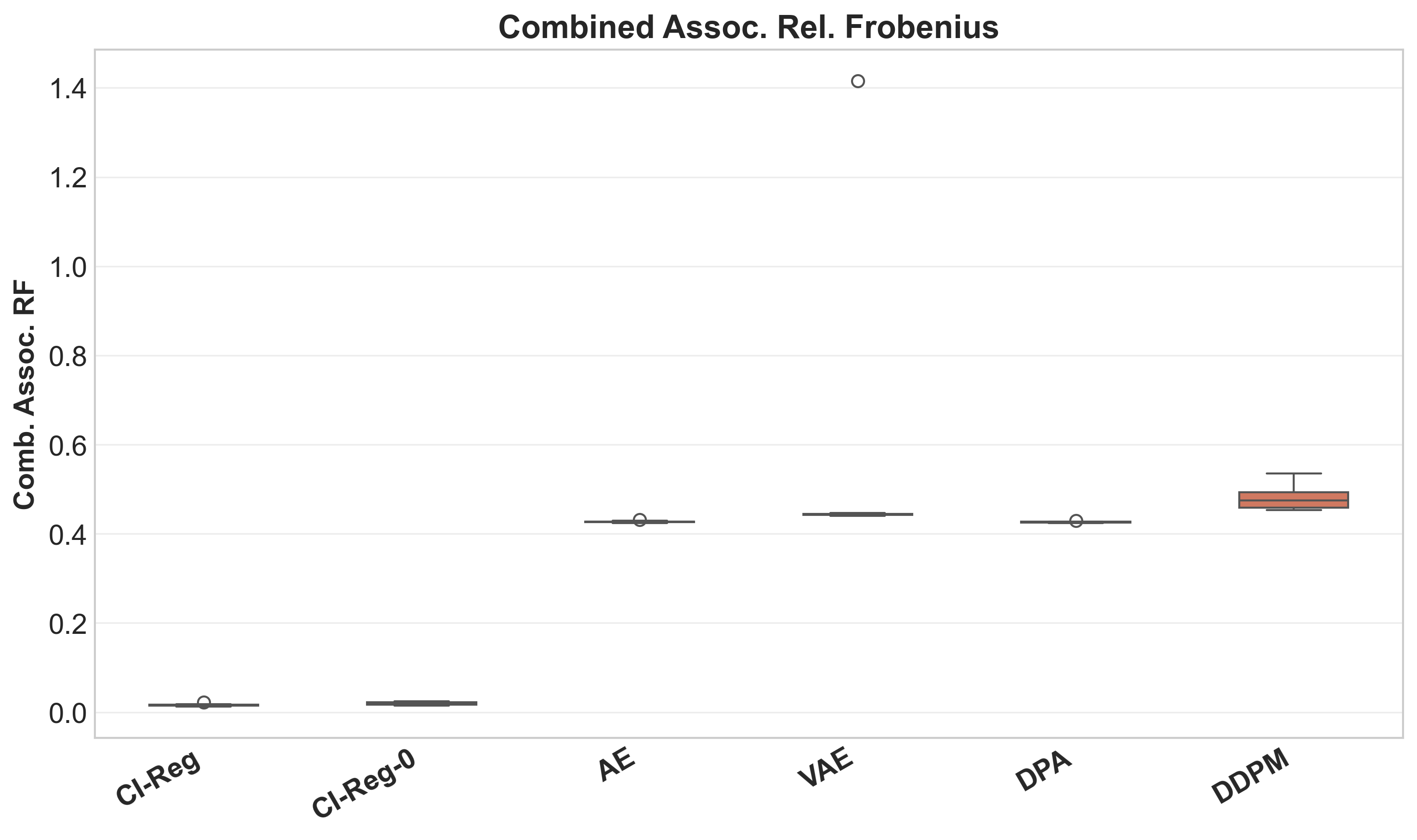}
    \caption{PM2.5}
\end{subfigure}

\caption{
Combined association relative Frobenius error on the real-world datasets comparing CI-Reg, CI-Reg-0, AE, VAE, DPA, and DDPM. Lower values indicate better reconstruction performance.
}
\label{fig:real_combined_assoc_relative_frobenius}
\end{figure}

\begin{figure}[htbp]
\centering

\begin{subfigure}[b]{0.48\textwidth}
    \centering
    \includegraphics[width=\textwidth]{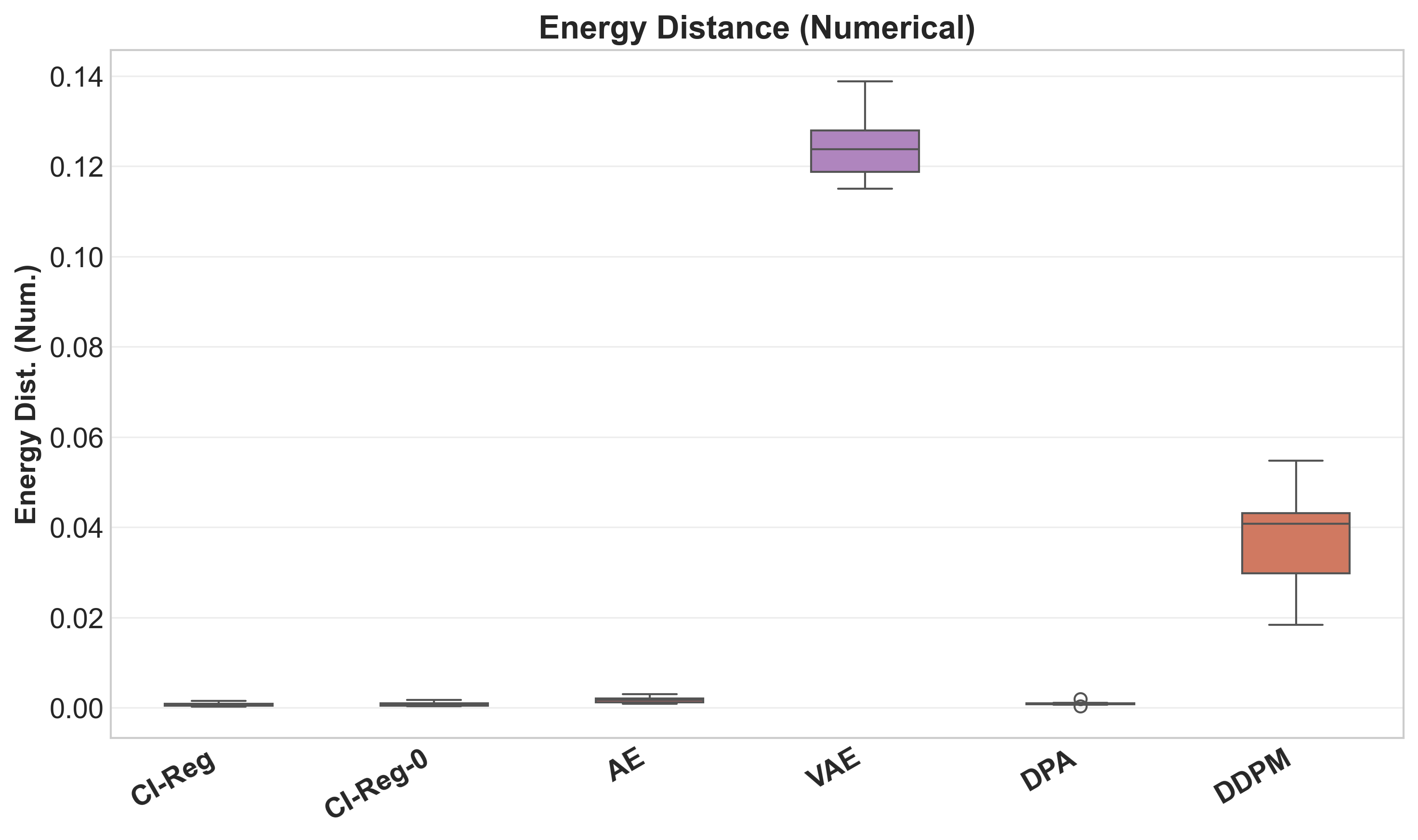}
    \caption{Adult}
\end{subfigure}
\hfill
\begin{subfigure}[b]{0.48\textwidth}
    \centering
    \includegraphics[width=\textwidth]{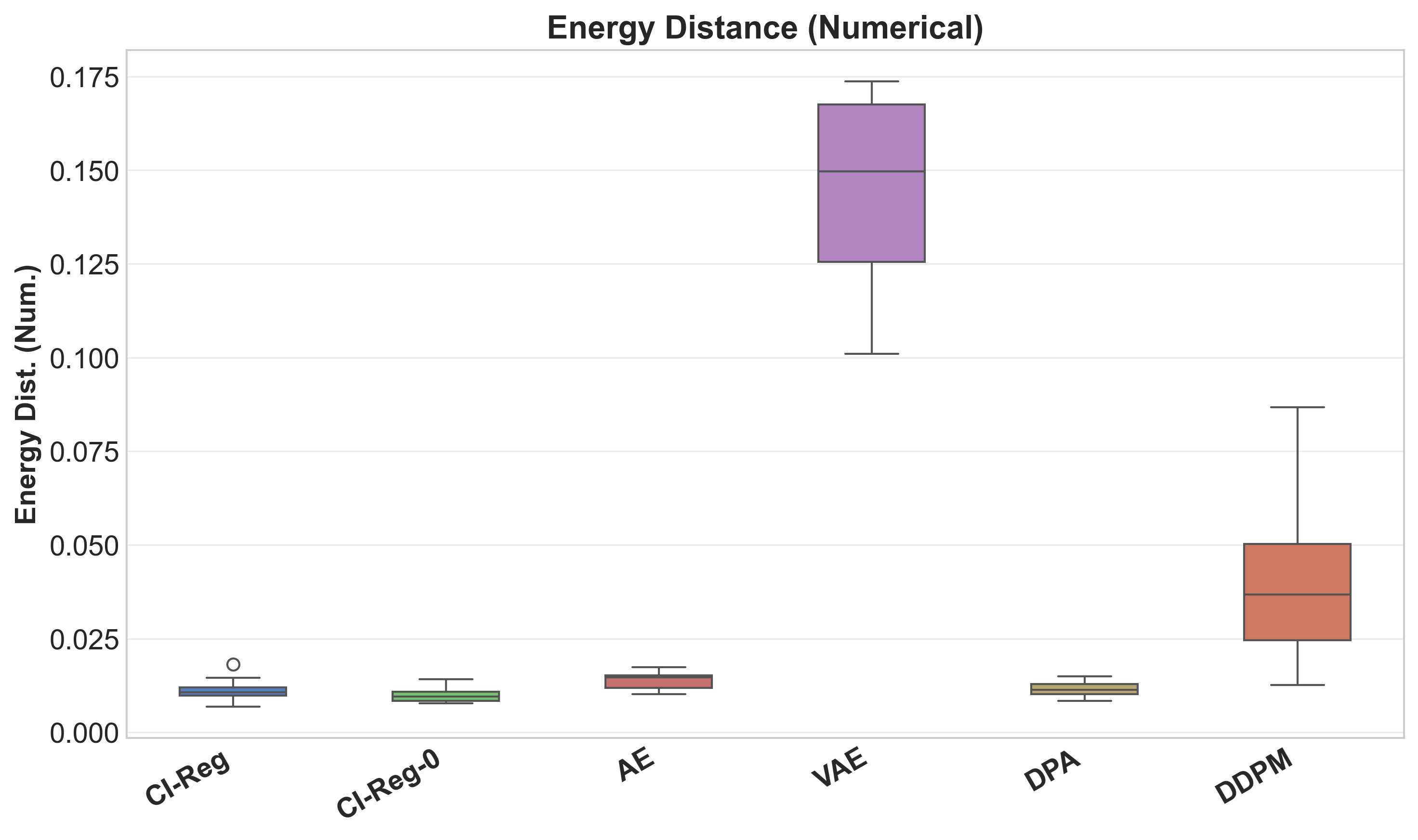}
    \caption{Insurance}
\end{subfigure}

\vspace{0.5em}

\begin{subfigure}[b]{0.48\textwidth}
    \centering
    \includegraphics[width=\textwidth]{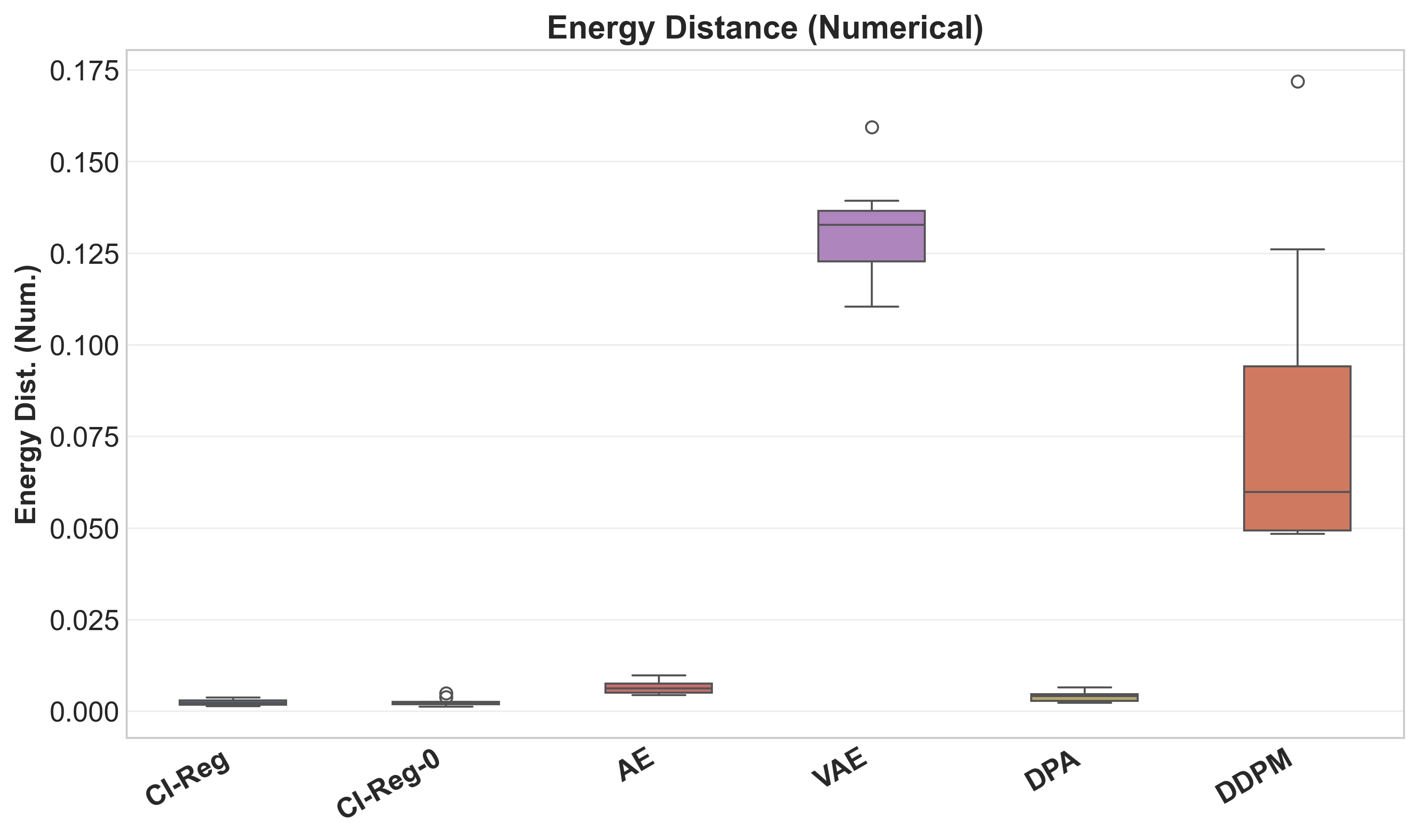}
    \caption{Shoppers}
\end{subfigure}
\hfill
\begin{subfigure}[b]{0.48\textwidth}
    \centering
    \includegraphics[width=\textwidth]{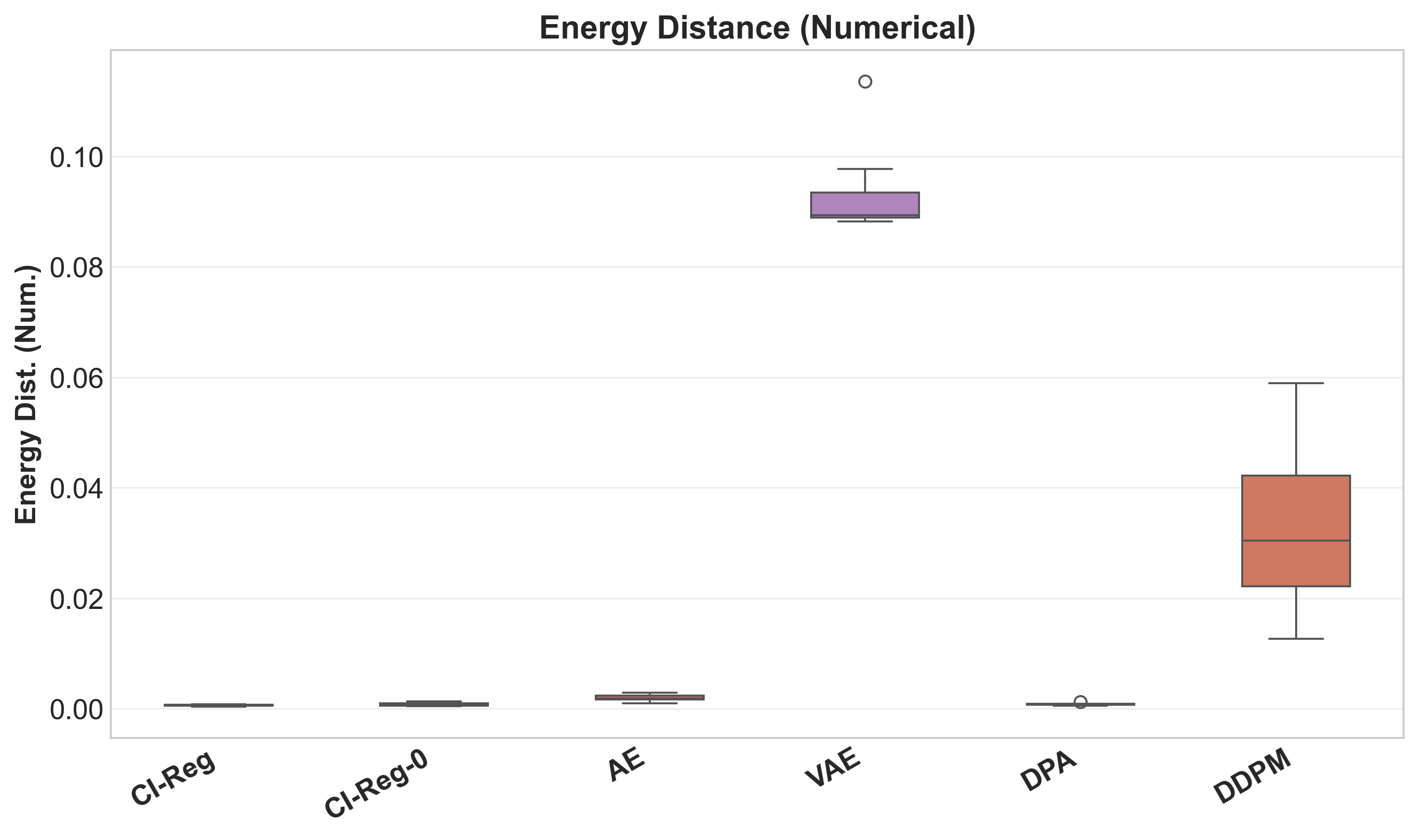}
    \caption{PM2.5}
\end{subfigure}

\caption{
Energy distance for the numerical part on the real-world datasets comparing CI-Reg, CI-Reg-0, AE, VAE, DPA, and DDPM. Lower values indicate better reconstruction performance.
}
\label{fig:real_energy_distance_cont}
\end{figure}

\begin{figure}[htbp]
\centering

\begin{subfigure}[b]{0.48\textwidth}
    \centering
    \includegraphics[width=\textwidth]{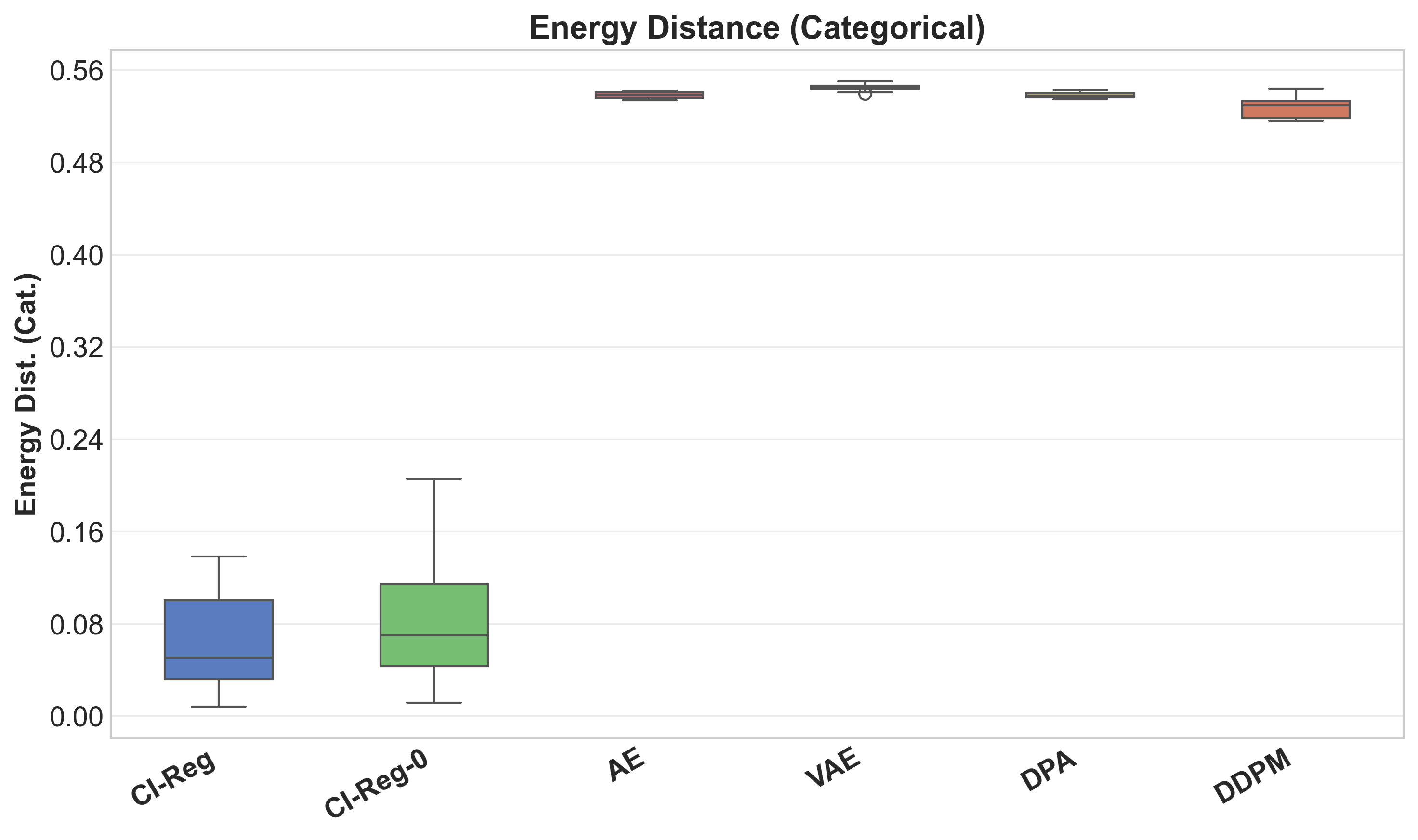}
    \caption{Adult}
\end{subfigure}
\hfill
\begin{subfigure}[b]{0.48\textwidth}
    \centering
    \includegraphics[width=\textwidth]{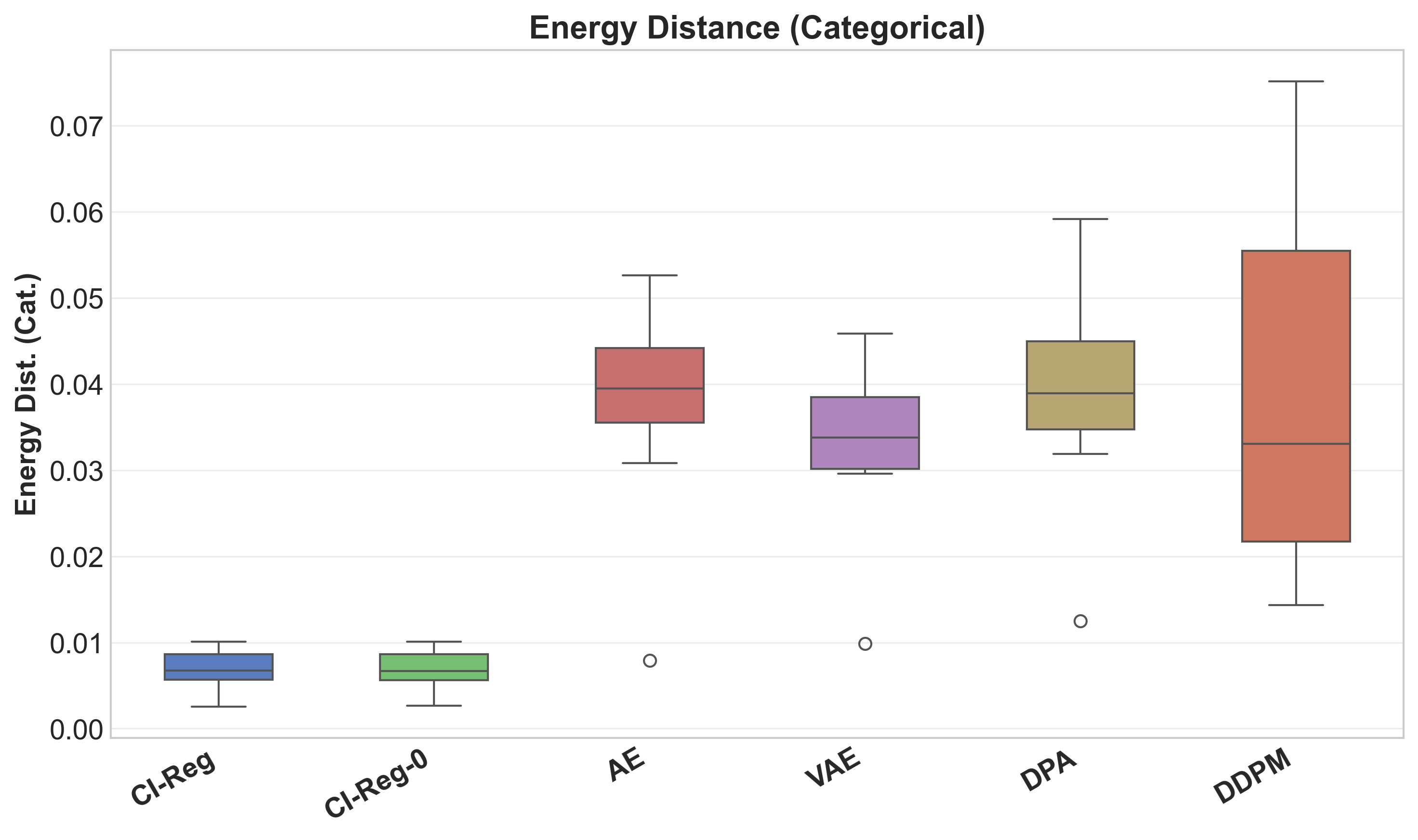}
    \caption{Insurance}
\end{subfigure}

\vspace{0.5em}

\begin{subfigure}[b]{0.48\textwidth}
    \centering
    \includegraphics[width=\textwidth]{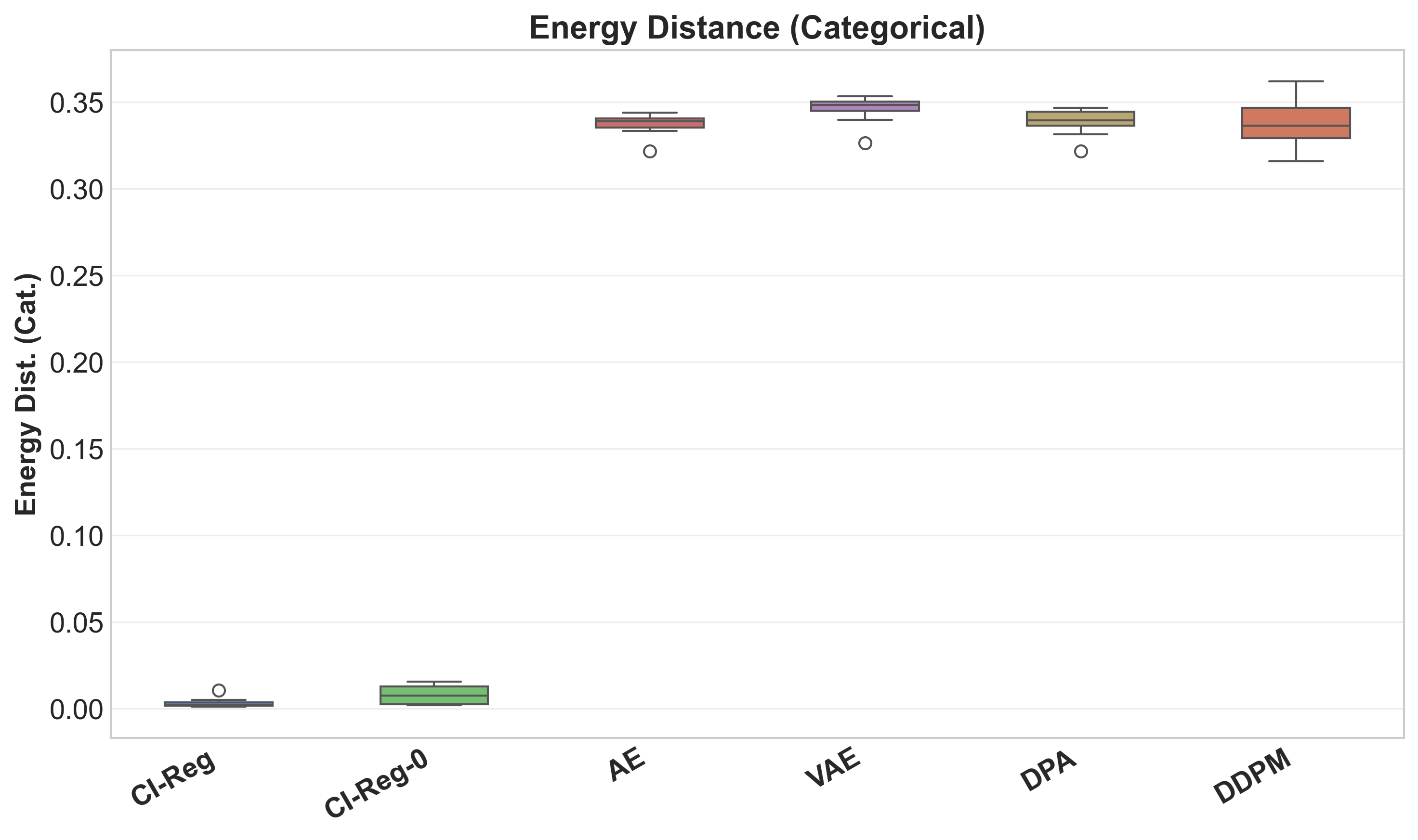}
    \caption{Shoppers}
\end{subfigure}
\hfill
\begin{subfigure}[b]{0.48\textwidth}
    \centering
    \includegraphics[width=\textwidth]{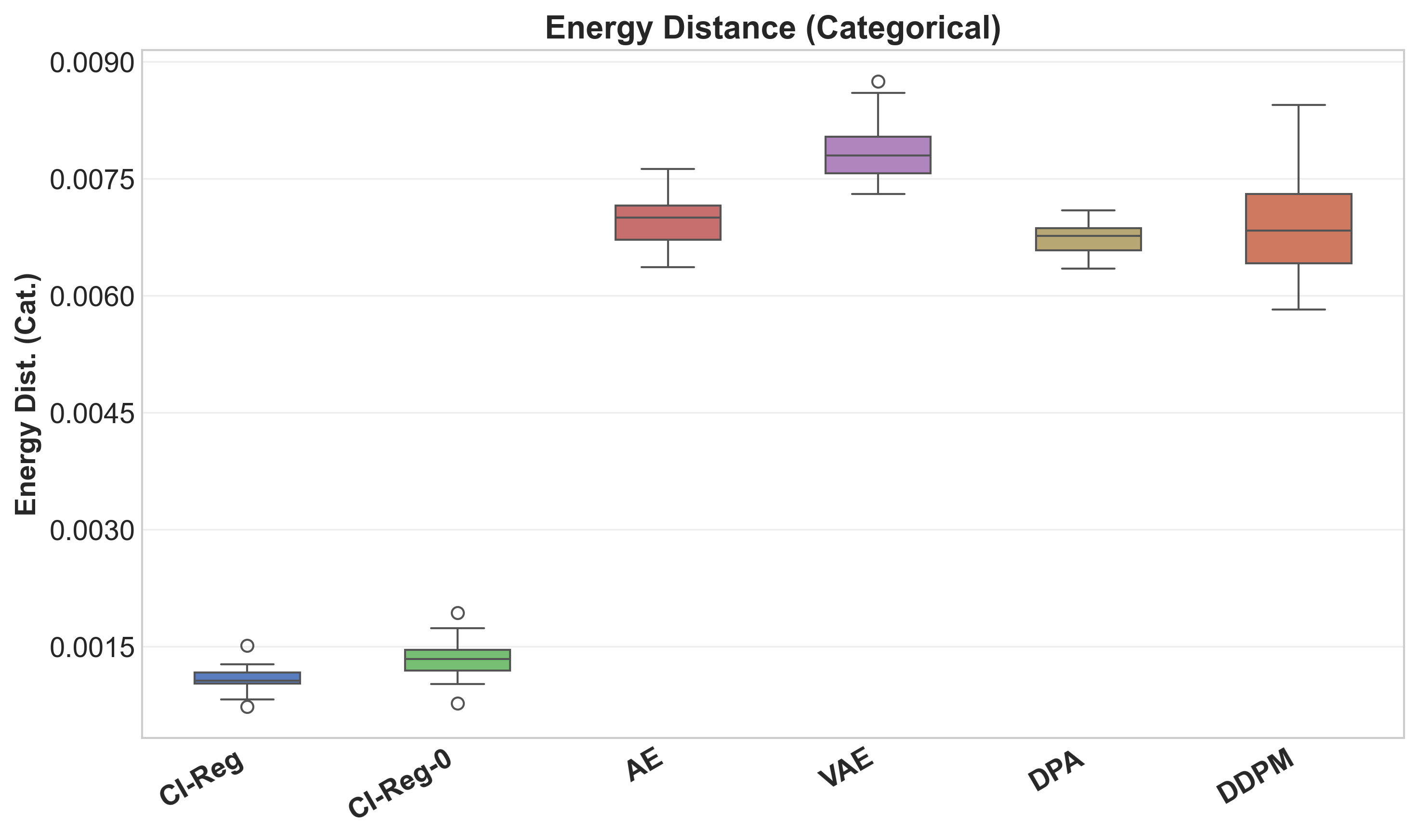}
    \caption{PM2.5}
\end{subfigure}

\caption{
Energy distance for the categorical part on the real-world datasets comparing CI-Reg, CI-Reg-0, AE, VAE, DPA, and DDPM. Lower values indicate better reconstruction performance.
}
\label{fig:real_energy_distance_cat}
\end{figure}

Figures \ref{fig:real_energy_distance_cont} and \ref{fig:real_energy_distance_cat} further illustrate the behavior of the proposed method on different data components. For the numerical component, CI-Reg and CI-Reg-0 remain competitive with DPA and generally outperform AE and VAE across datasets, indicating that the conditional independence regularization does not substantially compromise numerical distribution recovery. For the categorical component, however, CI-Reg and CI-Reg-0 consistently achieve substantially smaller energy distances than the competing baselines, often by a large margin. These results further support the main conclusion that explicitly separating numerical and categorical reconstruction objectives is particularly beneficial for recovering categorical conditional distributions in mixed-type data.

\end{document}